\documentclass[journal]{IEEEtran}

\def\citetechreport{0}    

\usepackage{color}

\usepackage{float}
\usepackage{amsmath}
\usepackage{amsthm}
\usepackage{amssymb}
\usepackage{setspace}
\usepackage{graphicx}
\usepackage{tabularx}

\usepackage{cite}
\usepackage{flushend}

\DeclareMathOperator*{\argmin}{arg\,min}

\usepackage[font={footnotesize}]{subcaption}
\graphicspath{{./image/}}

\newtheorem{definition}{Definition}
\newtheorem{assumption}{Assumption}
\newtheorem{theorem}{Theorem}

\newtheorem{lemma}{Lemma}

\newtheorem{proposition}{Proposition}

\usepackage[ruled, lined, noend, linesnumbered]{algorithm2e}
\SetKwBlock{RepeatNoEnd}{repeat}{}

\usepackage{enumitem}

\allowdisplaybreaks

\begin{document}

\title{Adaptive Federated Learning in Resource Constrained Edge Computing Systems
}

\makeatletter

\author{
\IEEEauthorblockN{Shiqiang Wang, Tiffany Tuor, Theodoros Salonidis, Kin K. Leung, \\ Christian Makaya, Ting He, Kevin Chan}
\thanks{
S. Wang, T. Salonidis, and C. Makaya are with IBM T. J. Watson Research Center, Yorktown Heights, NY, USA. Email: \{wangshiq, tsaloni\}@us.ibm.com, chrismak@ieee.org

T. Tuor and K. K. Leung are with Imperial College London, UK. Email: \{tiffany.tuor14, kin.leung\}@imperial.ac.uk

T. He is with Pennsylvania State University, University Park, PA, USA. Email: t.he@cse.psu.edu

K. Chan is with Army Research Laboratory, Adelphi, MD, USA. Email: kevin.s.chan.civ@mail.mil

This research was sponsored by the U.S. Army Research Laboratory and the U.K. Ministry of Defence under Agreement Number W911NF-16-3-0001. The views and conclusions contained in this document are those of the authors and should not be interpreted as representing the official policies, either expressed or implied, of the U.S. Army Research Laboratory, the U.S. Government, the U.K. Ministry of Defence or the U.K. Government. The U.S. and U.K. Governments are authorized to reproduce and distribute reprints for Government purposes notwithstanding any copyright notation hereon.

\if\citetechreport0
This paper (excluding appendices) has been accepted for publication in the IEEE Journal on Selected Areas in Communications.
\fi
A preliminary version of this work entitled ``When edge meets learning: adaptive control for resource-constrained distributed machine learning'' was presented at IEEE INFOCOM 2018~\cite{WangINFOCOM2018}. 
}
\vspace{-0.2in}
}

\maketitle

\begin{abstract}
Emerging technologies and applications including Internet of Things (IoT), social networking, and crowd-sourcing generate large amounts of data at the network edge. Machine learning models are often built from the collected data, to enable the detection, classification, and prediction of future events. Due to bandwidth, storage, and privacy concerns, it is often impractical to send all the data to a centralized location. In this paper, we consider the problem of learning model parameters from data distributed across multiple edge nodes, without sending raw data to a centralized place. Our focus is on a generic class of machine learning models that are trained using gradient-descent based approaches. We analyze the convergence bound of distributed gradient descent from a theoretical point of view, based on which we propose a control algorithm that determines the best trade-off between local update and global parameter aggregation to minimize the loss function under a given resource budget. The performance of the proposed algorithm is evaluated via extensive experiments with real datasets, both on a  networked prototype system and in a larger-scale simulated environment. The experimentation results show that our proposed approach performs near to the optimum with various machine learning models and different data distributions.
\end{abstract}


\IEEEpeerreviewmaketitle

\section{Introduction}

The rapid advancement of Internet of Things (IoT) and social networking applications results in an exponential growth of the data generated at the network edge. It has been predicted that the data generation rate will exceed the capacity of today's Internet in the near future \cite{ChiangFogOverview}.
Due to network bandwidth and data privacy concerns, it is impractical and often unnecessary to send all the data to a remote cloud. As a result, research organizations estimate that over $90\%$ of the data will be stored and processed locally   \cite{KellyABI}.
Local data storing and processing with global coordination is made possible by the emerging technology of mobile edge computing (MEC) \cite{MaoEdgeComptSurvey, MachEdgeComputSurvey}, where edge nodes, such as sensors, home gateways, micro servers, and small cells, are equipped with storage and computation capability. Multiple edge nodes work together with the remote cloud to perform large-scale distributed tasks that involve both local processing and remote coordination/execution.

To analyze large amounts of data and obtain useful information for the detection, classification, and prediction of future events, machine learning techniques are often applied. The definition of machine learning is very broad, ranging from simple data summarization with linear regression to multi-class classification with support vector machines (SVMs) and deep neural networks \cite{shalev2014understanding, Goodfellow-et-al-2016}.
The latter have shown very promising performance in recent years, for complex tasks such as image classification. One key enabler of machine learning is the ability to learn (train) models using a very large amount of data.
With the increasing amount of data being generated by new applications and with more applications becoming data-driven, one can foresee that machine learning tasks will become a dominant workload in distributed MEC systems in the future. However, it is challenging to perform distributed machine learning on resource-constrained MEC systems.

In this paper, we address the problem of how to efficiently utilize the limited computation and communication resources at the edge for the optimal learning performance. 
We consider a typical edge computing architecture where edge nodes are interconnected with the remote cloud via network elements, such as gateways and routers, as illustrated in Fig.~\ref{fig:architecture}. The raw data is collected and stored at multiple edge nodes, and a machine learning model is trained from the distributed data \emph{without} sending the raw data from the nodes to a central place. This variant of distributed machine learning (model training) from a federation of edge nodes is known as \emph{federated learning}~\cite{GoogleFederatedLearningBlog,mcmahan2016communication,WirelessNetworkIntelligence}.

\begin{figure}
\centering
\includegraphics[width=0.47\textwidth]{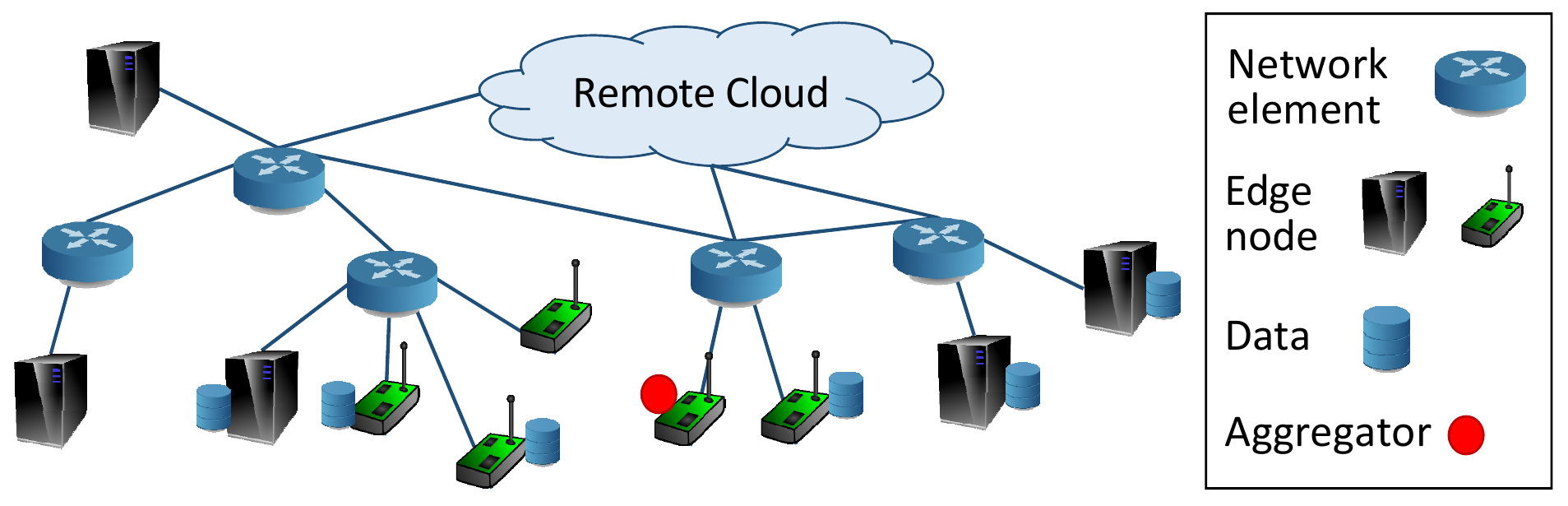}
\caption{System architecture.}
\label{fig:architecture}
\end{figure}

We focus on gradient-descent based federated learning algorithms, which have general applicability to a wide range of machine learning models. 
The learning process includes \emph{local update} steps where each edge node performs gradient descent to adjust the (local) model parameter to minimize the loss function defined on its own dataset. It also includes \emph{global aggregation} steps where model parameters obtained at different edge nodes are sent to an aggregator, which is a logical component that can run on the remote cloud, a network element, or an edge node. The aggregator aggregates these parameters (e.g., by taking a weighted average) and sends an updated parameter back to the edge nodes for the next round of iteration. 
The frequency of global aggregation is configurable; one can aggregate at an interval of one or multiple local updates. 
Each local update consumes computation resource of the edge node, and each global aggregation consumes communication resource of the network. 
The amount of consumed resources may vary over time, and there is a complex relationship among the frequency of global aggregation, the model training accuracy, and resource consumption.

We propose an algorithm to determine the frequency of global aggregation so that the available resource is most efficiently used. This is important because the training of machine learning models is usually resource-intensive, and a non-optimal operation of the learning task may waste a significant amount of resources.
Our main contributions in this paper are as follows:
\begin{enumerate}
\item We analyze the convergence bound of gradient-descent based federated learning from a theoretical perspective, and obtain a novel convergence bound that incorporates non-independent-and-identically-distributed (non-i.i.d.) data distributions among nodes and an arbitrary number of local updates between two global aggregations.
\item Using the above theoretical convergence bound, we propose a control algorithm that learns the data distribution, system dynamics, and model characteristics, based on which it dynamically adapts the frequency of global aggregation in real time to minimize the learning loss under a fixed resource budget.
\item We evaluate the performance of the proposed control algorithm via extensive experiments using real datasets both on a hardware prototype and in a simulated environment, which confirm that our proposed approach provides near-optimal performance for different data distributions, various machine learning models, and system configurations with different numbers of edge nodes. 
\end{enumerate}

\section{Related Work}

Existing work on MEC focuses on generic applications, where solutions have been proposed for application offloading~\cite{XiaoOffloading2017, TongOffloading2016}, workload scheduling \cite{TongHierarchical2016, TanInfocom2017}, and service migration triggered by user mobility \cite{WangICDCS2017, WangMicroCloudPredictedCost}. However, they do not address the relationship among communication, computation, and training accuracy for machine learning applications, which is important for optimizing the performance of machine learning tasks.

The concept of federated learning was first proposed in \cite{mcmahan2016communication}, which showed its effectiveness through experiments on various datasets. Based on the comparison of synchronous and asynchronous methods of distributed gradient descent in \cite{Chen2016}, it is proposed in \cite{mcmahan2016communication} that federated learning should use the synchronous approach because it is more efficient than asynchronous approaches. The approach in \cite{mcmahan2016communication} uses a fixed global aggregation frequency. It does not provide theoretical convergence guarantee and the experiments were not conducted in a network setting.  
Several extensions have been made to the original federated learning proposal recently. For example, a mechanism for secure global aggregation is proposed in~\cite{SecurityAggregationFederatedLearning}. Methods for compressing the information exchanged within one global aggregation step is proposed in~\cite{Jakub2016,hardy2017distributed}. Adjustments to the standard gradient descent procedure for better performance in the federated setting is studied in~\cite{FederatedOptimization}. Participant (client) selection for federated learning is studied in~\cite{nishio2018client}. An approach that shares a small amount of data with other nodes for better learning performance with non-i.i.d. data distribution is proposed in~\cite{zhao2018federated}.
These studies do not consider the adaptation of global aggregation frequency, and thus they are orthogonal to our work in this paper. To the best of our knowledge, the adaptation of global aggregation frequency for federated learning with resource constraints has not been studied in the literature.

An area related to federated learning is distributed machine learning in datacenters through the use of worker machines and parameter servers~\cite{li2014scaling}. The main difference between the datacenter environment and edge computing environment is that in datacenters, shared storage is usually used. The worker machines do not keep persistent data storage on their own, and they fetch the data from the shared storage at the beginning of the learning process. As a result, the data samples obtained by different workers are usually independent and identically distributed (i.i.d.). 
In federated learning, the data is collected at the edge directly and stored persistently at edge nodes, thus the data distribution at different edge nodes is usually non-i.i.d.
Concurrently with our work in this paper, optimization of synchronization frequency with running time considerations is studied in~\cite{WangSysML2019} for the datacenter setting. It does not consider characteristics of non-i.i.d. data distributions which is essential in federated learning.

Distributed machine learning across multiple datacenters in different geographical locations is studied in \cite{Gaia2017}, where a threshold-based approach to reduce the communication among different datacenters is proposed. Although the work in \cite{Gaia2017} is related to the adaptation of synchronization frequency with resource considerations, it focuses on peer-to-peer connected datacenters, which is different from the federated learning architecture that is not peer-to-peer. It also allows asynchronism among datacenter nodes, which is not the case in federated learning. 
In addition, the approach in \cite{Gaia2017} is designed empirically and does not consider a concrete theoretical objective, nor does it consider computation resource constraint which is important in MEC systems in addition to constrained communication resource.

From a theoretical perspective, bounds on the convergence of distributed gradient descent are obtained in~\cite{agarwal2011distributed,lian2015asynchronous,pmlr-v70-zheng17b}, which only allow one step of local update before global aggregation. Partial global aggregation is allowed in the decentralized gradient descent approach in \cite{lian2017can,lian2017asynchronous}, where after each local update step, parameter aggregation is performed over a non-empty subset of nodes, which does not apply in our federated learning setting where there is no aggregation at all after some of the local update steps. 
Multiple local updates before aggregation is possible in the bound derived in~\cite{Gaia2017}, but the number of local updates varies based on the thresholding procedure and cannot be specified as a given constant. Concurrently with our work, bounds with a fixed number of local updates between global aggregation steps are derived in~\cite{CooperativeSGD,YuAAAI2019}. However, the bound in~\cite{CooperativeSGD} only works with i.i.d. data distribution; the bound in~\cite{YuAAAI2019} is independent from how different the datasets are, which is inefficient because it does not capture the fact that training on i.i.d. data is likely to converge faster than training on non-i.i.d. data.
Related studies on distributed optimization that are applicable for machine learning applications also include \cite{zhang2012communication, arjevani2015communication, ma2017distributed}, where a separate solver is used to solve a local problem. The main focus of \cite{zhang2012communication, arjevani2015communication, ma2017distributed} is the trade-off between communication and optimality, where the complexity of solving the local problem (such as the number of local updates needed) is not studied. 
In addition, many of the existing studies either explicitly or implicitly assume i.i.d. data distribution at different nodes, which is inappropriate in federated learning.
To our knowledge, the convergence bound of distributed gradient descent in the federated learning setting, {which captures both the characteristics of different (possibly non-i.i.d. distributed) datasets and a given number of local update steps between two global aggregation steps, has not been studied in the literature.

In contrast to the above research, our work in this paper formally addresses the problem of dynamically determining the global aggregation frequency to \emph{optimize the learning with a given resource budget} for federated learning in MEC systems. 
This is a non-trivial problem due to the complex dependency between each learning step and its previous learning steps, which is hard to capture analytically.
It is also challenging due to non-i.i.d. data distributions at different nodes, where the data distribution is unknown beforehand and the datasets may have different degrees of similarities with each other, and the real-time dynamics of the system.
We propose an algorithm that is derived from theoretical analysis and adapts to real-time system dynamics.

We start with summarizing the basics of federated learning in the next section. In Section IV, we describe our problem formulation. The convergence analysis and control algorithm are presented in Sections V and VI, respectively. Experimentation results are shown in Section VII and the conclusion is presented in Section VIII.

\section{Preliminaries and Definitions \label{sec:distributedMachineLearning}}

\subsection{Loss Function}

Machine learning models include a set of parameters which are learned based on training data. A training data sample~$j$ usually consists of two parts. One is a vector $\mathbf{x}_j$ that is regarded as the input of the machine learning model (such as the pixels of an image); the other is a scalar $y_j$ that is the desired output of the model (such as the label of the image). To facilitate the learning, each model has a loss function defined on its parameter vector $\mathbf{w}$ for each data sample $j$. The loss function captures the error of the model on the training data, and the model learning process is to minimize the loss function on a collection of training data samples.
For each data sample $j$, we define the loss function as $f(\mathbf{w}, \mathbf{x}_j, y_j)$, which we write as $f_j(\mathbf{w})$ in short\footnote{Note that some unsupervised models (such as K-means) only learn on $\mathbf{x}_j$ and do not require the existence of $y_j$ in the training data. In such cases, the loss function value only depends on $\mathbf{x}_j$.}. 

\begin{table} \protect\caption{Loss functions for popular machine learning models}  \label{tab:learningModels} 

\vspace{-0.15in}

\renewcommand{\arraystretch}{1.4} 

\center{\footnotesize

\begin{tabularx}{\linewidth}
{>{\setlength\hsize{0.5\hsize}}X >{\setlength\hsize{1.5\hsize}}X} 
\hline Model & Loss function $f(\mathbf{w}, \mathbf{x}_j, y_j)$  ($\triangleq f_j (\mathbf{w})$)\\ 
\hline 

Squared-SVM & $\frac{\lambda}{2} \Vert \mathbf{w} \Vert^2  + \frac{1}{2} \max \left\{ 0; 1 - y_j \mathbf{w}^\mathrm{T} \mathbf{x}_j \right\}^2 $  ($\lambda$ is const.) \\
Linear regression & $\frac{1}{2} \Vert y_j - \mathbf{w}^\mathrm{T} \mathbf{x}_j \Vert^2 $  \\
K-means & $ \frac{1}{2} \min_l \Vert \mathbf{x}_j - \mathbf{w}_{(l)}  \Vert^2 $ where $\mathbf{w} \triangleq [\mathbf{w}_{(1)}^\mathrm{T}, \mathbf{w}_{(2)}^\mathrm{T}, ...]^\mathrm{T}$ \\
Convolutional neural network & Cross-entropy on cascaded linear and non-linear transforms, see \cite{Goodfellow-et-al-2016} \\
\hline
\end{tabularx} 

}

\end{table}

Examples of loss functions of popular machine learning models are summarized\footnote{While our focus is on non-probabilistic learning models, similar loss functions can be defined for probabilistic models where the goal is to minimize the negative of the log-likelihood function, for instance.} in Table~\ref{tab:learningModels} \cite{shalev2014understanding,Goodfellow-et-al-2016,bottou2010large}.
For convenience, we assume that all vectors are column vectors in this paper and use $\mathbf{x}^\mathrm{T}$ to denote the transpose of $\mathbf{x}$. We use ``$\triangleq$'' to denote ``is defined to be equal to'' and use $\Vert \cdot \Vert$ to denote the $\mathcal{L}^2$ norm.

Assume that we have $N$ edge nodes with local datasets $\mathcal{D}_{1},\mathcal{D}_{2},...,\mathcal{D}_{i},...,\mathcal{D}_{N}$. For each dataset $\mathcal{D}_{i}$ at node~$i$, the loss function on the collection of data samples at this node is
\begin{equation}
F_{i}(\mathbf{w}) \triangleq \frac{1}{\left|\mathcal{D}_{i}\right|}\sum_{j\in\mathcal{D}_{i}}f_{j}(\mathbf{w}).
\label{eq:localLossFuncAllSamples}
\end{equation}
We define $D_{i}\triangleq\left|\mathcal{D}_{i}\right|$, where $| \cdot |$ denotes the size of the set, and $D\triangleq\sum_{i=1}^{N}D_{i}$.
Assuming $\mathcal{D}_{i} \cap \mathcal{D}_{i'} = \emptyset$ for $i \neq i'$,  we define the global loss function on all the distributed datasets as
\begin{equation}
F(\mathbf{w}) \triangleq  \frac{\sum_{j\in\cup_i\mathcal{D}_{i}}f_{j}(\mathbf{w})}{\left|\cup_i\mathcal{D}_{i}\right|} = \frac{\sum_{i=1}^{N}D_{i}F_{i}(\mathbf{w})}{D}.
\label{eq:globalLossFuncAllSamples}
\end{equation}
Note that $F(\mathbf{w})$ \emph{cannot} be directly computed without sharing information among multiple nodes.

\subsection{The Learning Problem}

The learning problem is to minimize $F(\mathbf{w})$, i.e., to find
\begin{equation}
\mathbf{w}^{*} \triangleq \arg\min F(\mathbf{w}).
\label{eq:learningProblem}
\end{equation}
Due to the inherent complexity of most machine learning models, it is usually impossible to find a closed-form solution to (\ref{eq:learningProblem}). Thus, (\ref{eq:learningProblem}) is often solved using gradient-descent techniques.

\subsection{Distributed Gradient Descent}
\label{subsec:distGradDescent}

We present a canonical distributed gradient-descent algorithm to solve (\ref{eq:learningProblem}), which is widely used in state-of-the-art federated learning systems (e.g., \cite{mcmahan2016communication}).
Each node~$i$ has its local model parameter $\mathbf{w}_{i}{(t)}$, where $t=0,1,2,...$ denotes the iteration index. At $t=0$, the local parameters 
for all nodes $i$ are initialized to the same value. For $t>0$, new values of $\mathbf{w}_{i}{(t)}$ are computed according to a gradient-descent update rule on the local loss function, based on the parameter value in the previous iteration $t-1$. This gradient-descent step on the local loss function (defined on the local dataset) at each node is referred to as the \emph{local update}.
After one or multiple local updates, a \emph{global aggregation} is performed through the aggregator to update the local parameter at each node to the weighted average of all nodes' parameters.
We define that each \emph{iteration} includes a local update step which is possibly followed by a global aggregation step.

After global aggregation, the local parameter $\mathbf{w}_{i}{(t)}$ at each node~$i$ usually changes. For convenience, we use $\widetilde{\mathbf{w}}_{i}{(t)}$ to denote the parameter at node~$i$ after possible global aggregation. 
If no aggregation is performed at iteration $t$, we have $\widetilde{\mathbf{w}}_{i}{(t)}=\mathbf{w}_{i}{(t)}$.
If aggregation is performed at iteration $t$, then generally
$\widetilde{\mathbf{w}}_{i}{(t)} \neq \mathbf{w}_{i}{(t)}$ and we set $\widetilde{\mathbf{w}}_{i}{(t)} = \mathbf{w}{(t)}$, where $\mathbf{w}{(t)}$ is a weighted average of $\mathbf{w}_{i}{(t)}$ defined in (\ref{eq:globalAverage}) below. An example of these definitions is shown in Fig.~\ref{fig:wVariables}.

\begin{figure}
\centering
\includegraphics[width=0.48\textwidth]{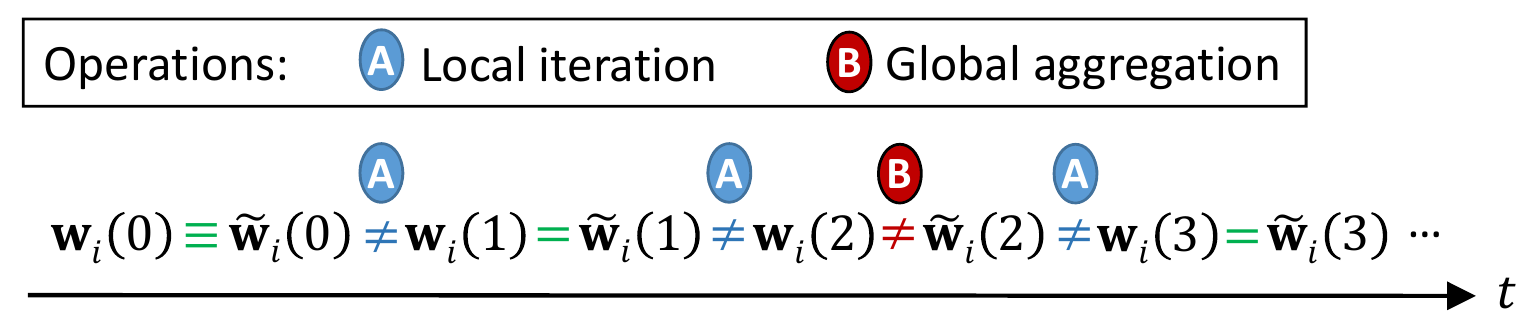}
\caption{Illustration of the values of $\mathbf{w}_i(t)$ and $\widetilde{\mathbf{w}}_i(t)$ at node $i$.}
\label{fig:wVariables}
\end{figure}

The local update in each iteration is performed on the parameter after possible global aggregation in the previous iteration. For each node~$i$, the update rule is as follows:
\begin{align}
\mathbf{w}_{i}{(t)}=\widetilde{\mathbf{w}}_{i}{(t-1)}-\eta\nabla F_{i}\left(\widetilde{\mathbf{w}}_{i}{(t-1)}\right)\label{eq:localUpdate}
\end{align} 
where $\eta > 0$ is the step size.
For any iteration $t$ (which may or may not include a global aggregation step), we define 
\begin{equation}
\mathbf{w}{(t)}=\frac{\sum_{i=1}^{N}D_{i}\mathbf{w}_{i}{(t)}}{D}\label{eq:globalAverage}.
\end{equation}
This global model parameter $\mathbf{w}{(t)}$ is \emph{only observable to nodes in the system if  global aggregation is performed at iteration $t$}, but we define it for all $t$ to facilitate the analysis later. 

We define that the system performs $\tau$ steps of local updates at each node between every two global aggregations. We define $T$ as the total number of local iterations at each node.
For ease of presentation, we assume that $T$ is an integer multiple of $\tau$ in the theoretical analysis, which will be relaxed when we discuss practical aspects in Section~\ref{subsec:controlAlgImpl}.
The logic of distributed gradient descent is presented in Algorithm~\ref{alg:distGradDescent}, which ignores aspects related to the communication between the aggregator and edge nodes. Such aspects will be discussed later in Section~\ref{subsec:controlAlgImpl}.

\begin{algorithm}[t]
\caption{Distributed gradient descent (logical view)} 
\label{alg:distGradDescent} 

{\footnotesize

\KwIn{$\tau$, $T$}
\KwOut{Final model parameter $\mathbf{w}^\mathrm{f}$}
Initialize $\mathbf{w}^\mathrm{f}$, $\mathbf{w}_{i}{(0)}$ and $\widetilde{\mathbf{w}}_{i}{(0)}$ to the same value for all $i$;

\For{$t = 1,2,...,T$}
{
    For each node $i$ \emph{in parallel}, compute local update  using (\ref{eq:localUpdate});   \label{alg:distGradDescent:LocalUpdate}
            
    \If{$t$ is an integer multiple of $\tau$}
    {
        Set $\widetilde{\mathbf{w}}_{i}{(t)} \leftarrow \mathbf{w}{(t)}$ for all $i$, where $\mathbf{w}{(t)}$ is defined in (\ref{eq:globalAverage});  //Global aggregation    \label{alg:distGradDescent:TildeSetGlobalAgg}

        Update $\mathbf{w}^\mathrm{f} \leftarrow \arg\min_{\mathbf{w}\in\{\mathbf{w}^\mathrm{f},\mathbf{w}{(t)}\}}F(\mathbf{w})$;  \label{alg:distGradDescent:minStep}
    }
    \Else
    {  Set $\widetilde{\mathbf{w}}_{i}{(t)} \leftarrow \mathbf{w}_{i}{(t)}$ for all $i$;  //No global aggregation   \label{alg:distGradDescent:TildeSet}
    }

}

}
\end{algorithm}

The final model parameter $\mathbf{w}^\mathrm{f}$ obtained from Algorithm~\ref{alg:distGradDescent} is the one that has produced the minimum global loss after each global aggregation throughout the entire execution of the algorithm. We use $\mathbf{w}^\mathrm{f}$ instead of $\mathbf{w}(T)$, to align with the theoretical convergence bound that will be presented in Section~\ref{sec:mainResults}. In practice, we have seen that $\mathbf{w}^\mathrm{f}$ and $\mathbf{w}(T)$ are usually the same, but using $\mathbf{w}^\mathrm{f}$ provides theoretical rigor in terms of convergence guarantee so we use $\mathbf{w}^\mathrm{f}$ in this paper. Note that $F(\mathbf{w})$ in Line~\ref{alg:distGradDescent:minStep} of Algorithm~\ref{alg:distGradDescent} is computed in a distributed manner according to (\ref{eq:globalLossFuncAllSamples}); the details will be presented later.

The rationale behind Algorithm~\ref{alg:distGradDescent} is that when $\tau=1$, i.e., when we perform global aggregation after every local update step, the distributed gradient descent (ignoring communication aspects) is equivalent to the centralized gradient descent, where the latter assumes that all data samples are available at a centralized location and the global loss function and its gradient can be observed directly. 
This is due to the linearity of the gradient operator. 
\if\citetechreport1
Due to space limitation, see our online technical report \cite[Appendix~\ref{append:optimalityOfDistributedGradDescent}]{JournalTechReport} as well as \cite{FUSION2018} for detailed discussions about this.
\else
See Appendix~\ref{append:optimalityOfDistributedGradDescent} as well as \cite{FUSION2018} for detailed discussions about this.
\fi

The main notations in this paper are summarized in Table~\ref{tab:MainNotations}.

\begin{table}[t]
\caption {Summary of main notations} \label{tab:MainNotations} 
\vspace{-0.15in}
{\footnotesize
\begin{center}
\begin{tabularx}{\linewidth}
{>{\setlength\hsize{0.2\hsize}}X >{\setlength\hsize{1.8\hsize}}X} 
\hline
$F(\mathbf{w})$ & Global loss function \\
$F_i(\mathbf{w})$ & Local loss function for node $i$ \\
$t$ & Iteration index \\
$\mathbf{w}_i(t)$ & Local model parameter at node $i$ in iteration $t$ \\
$\mathbf{w}(t)$ & Global model parameter in iteration $t$ \\
$\mathbf{w}^\mathrm{f}$ & Final model parameter obtained at the end of learning process\\
$\mathbf{w}^*$ & True optimal model parameter that minimizes $F(\mathbf{w})$\\
$\eta$ & Gradient descent step size\\
$\tau$ & Number of local update steps between two global aggregations \\
$T$ & Total number of local update steps at each node \\
$K$ & Total number of global aggregation steps, equal to $T/\tau$ \\
$M$ ($m$)\! & Total number of resource types (the $m$-th type of resource) \\
$R_m$ & Total budget of the $m$-th type of resource \\
$c_m$ & Consumption of type-$m$ resource in one lo\underline{c}al update step\\
$b_m$ & Consumption of type-$m$ resource in one glo\underline{b}al aggregation step\\
$\rho$ & Lipschitz parameter of $F_i(\mathbf{w})$ ($\forall i$) and $F(\mathbf{w})$ \\
$\beta$ & Smoothness parameter of $F_i(\mathbf{w})$ ($\forall i$) and $F(\mathbf{w})$ \\
$\delta$ & Gradient divergence \\
$h(\tau)$ & Function defined in (\ref{eq:hFuncDef}), gap between the model parameters obtained from distributed and centralized gradient descents\\
$\varphi$ & Constant defined in Lemma~\ref{lemma:convergenceUpperBound}, control parameter \\
$G(\tau)$ & Function defined in (\ref{eq:GTauDef}), control objective \\
$\tau^*$ & Optimal $\tau$ obtained by minimizing $G(\tau)$ \\
\hline
\end{tabularx}
\end{center}
}
\end{table}

\section{Problem Formulation}
\label{sec:ProblemFormulation}

When there is a large amount of data (which is usually needed for training an accurate model) distributed at a large number of nodes, the federated learning process can consume a significant amount of resources. The notion of ``resources'' here is generic and can include time, energy, monetary cost etc. \emph{related to both computation and communication}.
One often has to limit the amount of resources used for learning each model, in order not to backlog the system and to keep the operational cost low.
This is particularly important in edge computing environments where the computation and communication resources are not as abundant as in datacenters.

Therefore, a natural question is how to make efficient use of a given amount of resources to minimize the loss function of model training.
For the distributed gradient-descent based learning approach presented above, the question narrows down to determining the optimal values of $T$ and $\tau$, so that the global loss function is minimized subject to a given resource constraint for this learning task.

We use $K$ to denote the total number of global aggregations within $T$ iterations. Because we assumed earlier that $T$ is an integer multiple of $\tau$, we have $K=\frac{T}{\tau}$.
We define 
\begin{equation}
\label{eq:wminDef}
\mathbf{w}^\mathrm{f} \triangleq \argmin_{\mathbf{w} \in \{\mathbf{w}(k\tau): k=0, 1, 2, ..., K\}} F(\mathbf{w}).
\end{equation}
It is easy to verify that this definition is equivalent to $\mathbf{w}^\mathrm{f}$ found from Algorithm~\ref{alg:distGradDescent}.

To compute $F(\mathbf{w})$ in (\ref{eq:wminDef}), each node $i$ first computes $F_i(\mathbf{w})$ and sends the result to the aggregator, then the aggregator computes $F(\mathbf{w})$ according to (\ref{eq:globalLossFuncAllSamples}). Since each node only knows the value of $\mathbf{w}(k\tau)$ after the $k$-th global aggregation,  $F_i(\mathbf{w}(k\tau))$ at node $i$ will be sent back to the aggregator at the $(k+1)$-th global aggregation, and the aggregator computes $F(\mathbf{w}(k\tau))$ afterwards. 
To compute the last loss value $F(\mathbf{w}(K\tau)) = F(\mathbf{w}(T))$, an additional round of local and global update is performed at the end.
We assume that at each node, local update consumes the same amount of resource no matter whether only the local loss is computed (in the last round) or both the local loss and gradient are computed (in all the other rounds), because the loss and gradient computations can usually be based on the same intermediate result. For example, the back propagation approach for computing gradients in neural networks requires a forward propagation procedure that essentially obtains the loss as an intermediate step \cite{Goodfellow-et-al-2016}.

We consider $M$ different types of resources. For example, one type of resource can be time, another type can be energy, a third type can be communication bandwidth, etc. For each $m\in\{1,2,...,M\}$, we define that each lo\underline{c}al update step at \emph{all} nodes consumes $c_m$ units of type-$m$ resource, and each glo\underline{b}al aggregation step consumes $b_m$ units of type-$m$ resource, where $c_m\geq 0$ and $b_m\geq 0$ are both \emph{finite} real numbers. 
For given $T$ and $\tau$, the total amount of consumed type-$m$ resource is $(T+1)c_m + \left(K + 1\right)b_m$, where the additional ``$+1$'' is for computing $F(\mathbf{w}(K\tau))$, as discussed above.

Let $R_m$ denote the total budget of type-$m$ resource. 
We seek the solution to the following problem:
\begin{align}
\min_{\tau, K \in \{1,2,3,...\}} & \,\, F(\mathbf{w}^\mathrm{f}) \label{eq:optimizationProblem} \\
\textrm{s.t.} & \,\, (T\!+\! 1)c_m \!+\! \left(K\! +\! 1\right)b_m \leq R_m \, , \,\, \forall m\!\in\! \{\!1,...,M\!\} \nonumber \\
& \,\,T=K\tau . \nonumber
\end{align}
To solve (\ref{eq:optimizationProblem}), we need to find out how  $\tau$ and $K$ (and thus $T$) affect the loss function $F(\mathbf{w}^\mathrm{f})$ computed on the final model parameter $\mathbf{w}^\mathrm{f}$.
It is generally impossible to find an exact analytical expression to relate $\tau$ and $K$ with $F(\mathbf{w}^\mathrm{f})$, because it depends on the convergence property of gradient descent (for which only upper/lower bounds are known \cite{convex}) and the impact of the global aggregation frequency on the convergence. Further, the resource consumptions $c_m$ and $b_m$ can be time-varying in practice which makes the problem even more challenging than (\ref{eq:optimizationProblem}) alone.

We analyze the convergence bound of distributed gradient descent (Algorithm~\ref{alg:distGradDescent}) in Section~\ref{sec:convergenceAnalysis}, then use this bound to approximately solve (\ref{eq:optimizationProblem}) and propose a control algorithm for adaptively choosing the best values of $\tau$ and $T$ to achieve near-optimal resource utilization in Section~\ref{sec:schedulingAlgorithm}.

\section{Convergence Analysis}
\label{sec:convergenceAnalysis}

We analyze the convergence of Algorithm~\ref{alg:distGradDescent} in this section and find an upper bound of $F(\mathbf{w}^\mathrm{f})-F(\mathbf{w}^*)$. To facilitate the analysis, we first introduce some notations.

\subsection{Definitions}

We can divide the $T$ iterations into $K$ different intervals, as shown in Fig.~\ref{fig:interval}, with only the first and last iterations in each interval containing global aggregation. We use the shorthand notations $[k]$ to denote the iteration interval\footnote{With slight abuse of notation, we use $[(k-1)\tau,k\tau]$ to denote the integers contained in the interval for simplicity. We use the same convention in other parts of the paper as long as there is no ambiguity.} $[(k-1)\tau,k\tau]$, for $k=1,2,...,K$.

\begin{figure}
\centering
\includegraphics[width=0.47\textwidth]{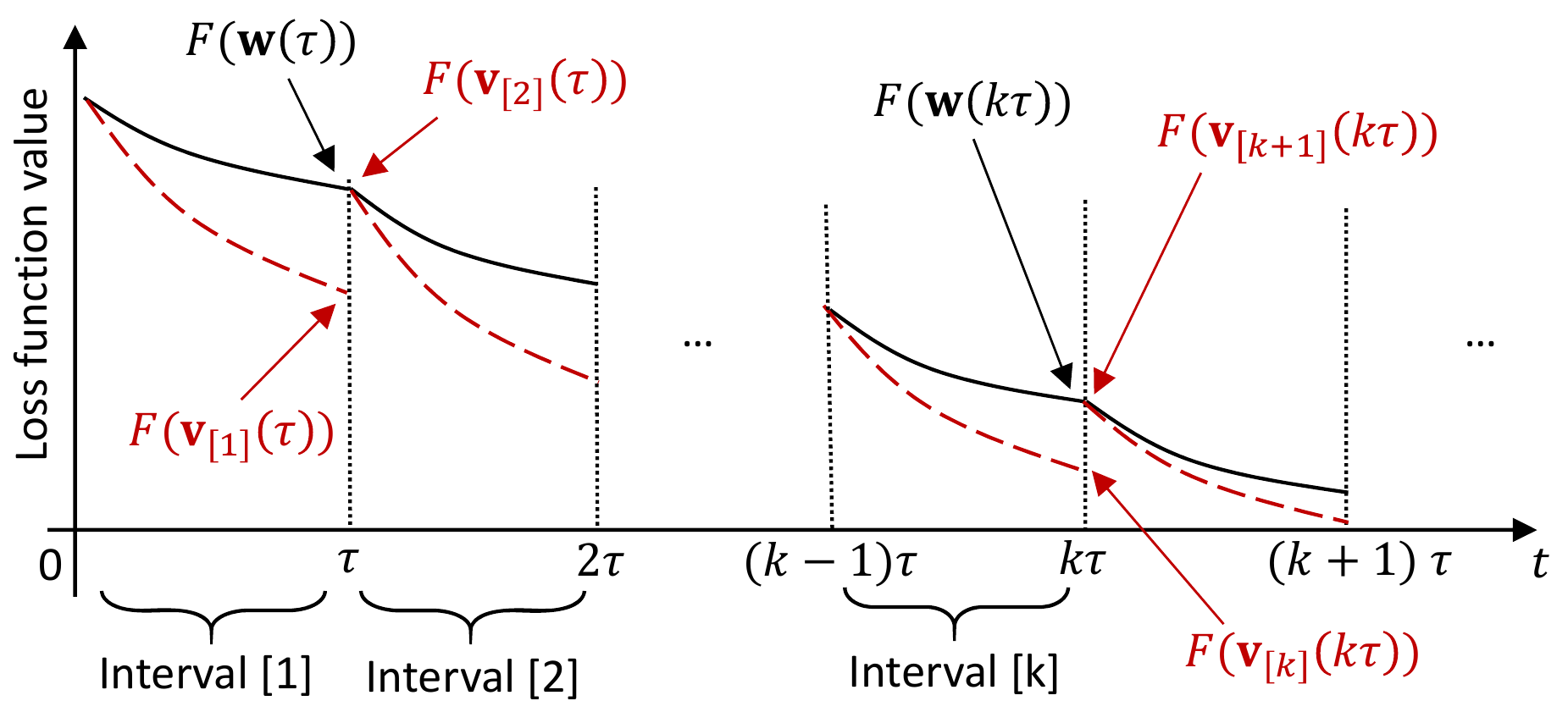}
\caption{Illustration of definitions in different intervals.}
\label{fig:interval}
\end{figure}

For each interval $[k]$, we use $\mathbf{v}_{[k]}(t)$ to denote an auxiliary parameter vector that follows a \emph{centralized} gradient descent according to
\begin{align}
\mathbf{v}_{[k]}{(t)}=\mathbf{v}_{[k]}{(t-1)}-\eta\nabla F(\mathbf{v}_{[k]}{(t-1)})\label{eq:updateV}
\end{align}
where $\mathbf{v}_{[k]}(t)$ is only defined for $t \in [(k-1)\tau,k\tau]$ for a given $k$. 
This update rule is based on the global loss function $F(\mathbf{w})$ which is only observable when all data samples are available at a central place (thus we call it centralized gradient descent), whereas the iteration in (\ref{eq:localUpdate}) is on the local loss function $F_i(\mathbf{w})$.

We define that $\mathbf{v}_{[k]}(t)$ is ``synchronized'' with $\mathbf{w}(t)$ at the beginning of each interval $[k]$, i.e., $\mathbf{v}_{[k]}((k-1)\tau) \triangleq \mathbf{w}((k-1)\tau)$, where $\mathbf{w}(t)$ is the average of local parameters defined in (\ref{eq:globalAverage}).
Note that we also have $\widetilde{\mathbf{w}}_{i}{((k-1)\tau)}=\mathbf{w}{((k-1)\tau)}$ for all~$i$ because the global aggregation (or initialization when $k=1$) is performed in iteration $(k-1)\tau$.

The above definitions enable us to find the convergence bound of Algorithm~\ref{alg:distGradDescent} by taking a two-step approach. The first step is to find the gap between $\mathbf{w}(k\tau)$ and $\mathbf{v}_{[k]}(k\tau)$ for each $k$, which is the difference between the distributed and centralized gradient descents after $\tau$ steps of local updates without global aggregation. The second step is to combine this gap with the convergence bound of $\mathbf{v}_{[k]}(t)$ within each interval $[k]$ to obtain the convergence bound of $\mathbf{w}(t)$.

For the purpose of the analysis, we make the following assumption to the loss function.
\begin{assumption}
\label{assumption:Convex}
We assume the following for all $i$:
\begin{enumerate}
\item $F_{i}(\mathbf{w})$ is convex
\item $F_{i}(\mathbf{w})$ is $\rho$-Lipschitz, i.e., $\Vert F_{i}(\mathbf{w}) - F_{i}(\mathbf{w}')  \Vert \leq \rho \Vert \mathbf{w} - \mathbf{w'}  \Vert$ for any $\mathbf{w}, \mathbf{w}'$
\item $F_{i}(\mathbf{w})$ is $\beta$-smooth, i.e., $\left\Vert \nabla F_{i}(\mathbf{w}) - \nabla F_{i}(\mathbf{w}')  \right\Vert \leq \beta \left\Vert \mathbf{w} - \mathbf{w}'  \right\Vert$ for any $\mathbf{w}, \mathbf{w'}$
\end{enumerate}
\end{assumption}
Assumption~\ref{assumption:Convex} is satisfied for squared-SVM and linear regression (see Table~\ref{tab:learningModels}). The experimentation results that will be presented in Section~\ref{sec:experimentation} show that our control algorithm also works well for models (such as neural network) whose loss functions do not satisfy Assumption~\ref{assumption:Convex}.

\begin{lemma}
$F(\mathbf{w})$ is convex, $\rho$-Lipschitz, and $\beta$-smooth.
\end{lemma}
\begin{proof}
Straightforwardly from Assumption~\ref{assumption:Convex}, the definition of $F(\mathbf{w})$, and triangle inequality.
\end{proof}

We also define the following metric to capture the \emph{divergence} between the gradient of a local loss function and the gradient of the global loss function. This divergence is \emph{related to how the data is distributed at different nodes}.

\begin{definition}
\label{def:gradientDivergence}
(Gradient Divergence)
For any  $i$ and $\mathbf{w}$, we define $\delta_i$  as an upper bound of $\left\Vert \nabla F_{i}(\mathbf{w})-\nabla F(\mathbf{w})\right\Vert$, i.e.,
\begin{equation}
\left\Vert \nabla F_{i}(\mathbf{w})-\nabla F(\mathbf{w})\right\Vert \leq \delta_i  .
\label{eq:deltaiDef}
\end{equation}
We also define $\delta \triangleq \frac{\sum_{i}D_{i}\delta_{i}}{D}$.
\end{definition}

\subsection{Main Results}
\label{sec:mainResults}

The below theorem gives an upper bound on the difference between $\mathbf{w}(t)$ and $\mathbf{v}_{[k]}(t)$ when $t$ is within the interval $[k]$.

\begin{theorem} \label{theorem:wBound}
For any interval $[k]$ and $t\in [k]$, we have
\begin{equation}
\left\Vert \mathbf{w}{(t)}-\mathbf{v}_{[k]}{(t)}\right\Vert \leq h(t-(k-1)\tau)
\label{eq:wBound}
\end{equation}
where
\begin{equation}
h(x) \triangleq \frac{\delta}{\beta}\left((\eta\beta+1)^{x}-1\right)-\eta\delta x 
\label{eq:hFuncDef}
\end{equation}
for any $x=0, 1, 2, ...$. 

Furthermore, as $F(\cdot )$ is $\rho$-Lipschitz, we have
$F(\mathbf{w}{(t)})-F(\mathbf{v}_{[k]}(t))\leq \rho h(t - (k-1)\tau)$.
\end{theorem}

\begin{proof}
We first obtain an upper bound of $\left\Vert \widetilde{\mathbf{w}}_{i}{(t)}-\mathbf{v}_{[k]}{(t)}\right\Vert$ for each node $i$, based on which the final result is obtained.
\if\citetechreport1
See our online technical report \cite[Appendix~\ref{append:proofWBound}]{JournalTechReport} for details.
\else
For details, see Appendix~\ref{append:proofWBound}.
\fi
\end{proof}

Note that we always have $\eta>0$ and $\beta>0$ because otherwise the gradient descent procedure or the loss function becomes trivial. Therefore, we have $(\eta\beta + 1)^x \geq \eta\beta x + 1$ for $x = 0, 1, 2, ...$ due to Bernoulli's inequality. Substituting this into (\ref{eq:hFuncDef}) confirms that we always have $h(x) \geq 0$.

It is easy to see that $h(0) = h(1) = 0$.
Therefore, when $t = (k-1)\tau$, i.e., at the beginning of the interval $[k]$, the upper bound in (\ref{eq:wBound}) is zero. This is consistent with the definition of $\mathbf{v}_{[k]}((k-1)\tau) = \mathbf{w}((k-1)\tau)$ for any $k$. When $t = (k-1)\tau + 1$ (i.e., the second iteration in interval $[k]$), the upper bound in (\ref{eq:wBound}) is also zero. This agrees with the discussion at the end of Section~\ref{subsec:distGradDescent}, showing that there is no gap between distributed and centralized gradient descents when only one local update is performed after the global aggregation.
If $\tau = 1$, then $t - (k-1)\tau$ is either $0$ or $1$ for any interval $[k]$ and $t \in [k]$. Hence, the upper bound in (\ref{eq:wBound}) becomes exact for $\tau = 1$.

For $\tau > 1$, the value of $x = t - (k-1)$ can be larger. When $x$ is large, the exponential term with $(\eta\beta+1)^{x}$ in (\ref{eq:hFuncDef}) becomes dominant, and the gap between $\mathbf{w}(t)$ and $\mathbf{v}_{[k]}(t)$ can increase exponentially with $t$ for $t \in [k]$.
We also note that $h(x)$ is proportional to the gradient divergence $\delta$ (see (\ref{eq:hFuncDef})), which is intuitive because the more the local gradient is different from the global gradient (for the same parameter $\mathbf{w}$), the larger the gap will be. The gap is caused by the difference in the local gradients at different nodes starting at the second local update after each global aggregation. In an extreme case when all nodes have exactly the same data samples (and thus the same local loss functions), the gradients will be always the same and $\delta = 0$, in which case $\mathbf{w}(t)$ and $\mathbf{v}_{[k]}(t)$ are always equal.

Theorem~\ref{theorem:wBound} gives an upper bound of the difference between distributed and centralized gradient descents for each iteration interval $[k]$, assuming that $\mathbf{v}_{[k]}(t)$ in the centralized gradient descent is synchronized with $\mathbf{w}(t)$ at the beginning of each $[k]$.
Based on this result, we first obtain the following lemma.

\begin{lemma}
\label{lemma:convergenceUpperBound}
When all the following conditions are satisfied:
\begin{enumerate}
\item $\eta \leq \frac{1}{\beta}$
\item $\eta\varphi  -\frac{\rho h(\tau)}{\tau\varepsilon^2} > 0$
\item $ F\left(\mathbf{v}_{[k]}(k\tau)\right)-F(\mathbf{w}^*) \geq \varepsilon$ for all $k$ 
\item $ F\left(\mathbf{w}(T)\right)-F(\mathbf{w}^*) \geq \varepsilon$
\end{enumerate}
for some $\varepsilon > 0$, where we define $\varphi \triangleq \omega \left(1-\frac{\beta \eta}{2}\right)$ and $\omega \triangleq \min_k \frac{1}{{\left\Vert \mathbf{v}_{[k]}((k-1)\tau)-\mathbf{w}^*\right\Vert^2 }}$,
then the convergence upper bound of Algorithm~\ref{alg:distGradDescent} after $T$ iterations is given by
\begin{equation}
F(\mathbf{w}(T))-F(\mathbf{w}^*) \leq \frac{1}{T\left(\eta\varphi -\frac{\rho h(\tau)}{\tau\varepsilon^2}\right)}.
\label{eq:convergenceUpperBound}
\end{equation}
\end{lemma}

\begin{proof}
We first analyze the convergence of $F\left(\mathbf{v}_{[k]}(t)\right)$ within each interval $[k]$. Then, we combine this result with the gap between $F(\mathbf{w}{(t)})$ and $F(\mathbf{v}_{[k]}(t))$ from Theorem~\ref{theorem:wBound} to obtain the final result. 
\if\citetechreport1
See our online technical report \cite[Appendix~\ref{append:proofConvergenceUpperBound}]{JournalTechReport} for details.
\else
For details, see Appendix~\ref{append:proofConvergenceUpperBound}.
\fi
\end{proof}

We then have the following theorem.

\begin{theorem}
\label{theorem:convergenceUpperBoundFinal}
When $\eta \leq \frac{1}{\beta}$, we have
\begin{equation}
F(\mathbf{w}^\mathrm{f}) -F(\mathbf{w}^*)\leq \frac{1}{2\eta\varphi T} + \sqrt{\frac{1}{4\eta^2\varphi^2 T^2} + \frac{\rho h(\tau)}{\eta\varphi\tau}} + \rho h(\tau). \label{eq:convergenceUpperBoundFinal}
\end{equation}
\end{theorem}
\begin{proof}
Condition 1 in Lemma~\ref{lemma:convergenceUpperBound} is always satisfied due to the condition $\eta \leq \frac{1}{\beta}$ in this theorem.

When $\rho h(\tau) = 0$, we can choose $\varepsilon$ to be arbitrarily small (but greater than zero) so that conditions 2--4 in Lemma~\ref{lemma:convergenceUpperBound} are satisfied.
We see that the right-hand sides of (\ref{eq:convergenceUpperBound}) and (\ref{eq:convergenceUpperBoundFinal}) are equal in this case (when $\rho h(\tau) = 0$), and the result in (\ref{eq:convergenceUpperBoundFinal}) follows directly from Lemma~\ref{lemma:convergenceUpperBound} because $F(\mathbf{w}^\mathrm{f}) -F(\mathbf{w}^*)\leq F(\mathbf{w}(T))-F(\mathbf{w}^*)$ according to the definition of $\mathbf{w}^\mathrm{f}$ in (\ref{eq:wminDef}).

We consider $\rho h(\tau) > 0$ in the following.
Consider the right-hand side of (\ref{eq:convergenceUpperBound}) and let
\begin{equation}
\varepsilon_0 = \frac{1}{T\left(\eta\varphi -\frac{\rho h(\tau)}{\tau\varepsilon_0^2}\right)}.
\label{eq:convergenceUpperBoundFinalProof1}
\end{equation}
Solving for $\varepsilon_0$, we obtain
\begin{equation}
\varepsilon_0 = \frac{1}{2\eta\varphi T} + \sqrt{\frac{1}{4\eta^2\varphi^2 T^2} + \frac{\rho h(\tau)}{\eta\varphi\tau}}
\label{eq:convergenceUpperBoundFinalProof2}
\end{equation}
where the negative solution is ignored because $\varepsilon > 0$ in Lemma~\ref{lemma:convergenceUpperBound}.
Because $\varepsilon_0 > 0$ according to (\ref{eq:convergenceUpperBoundFinalProof2}), the denominator of (\ref{eq:convergenceUpperBoundFinalProof1}) is greater than zero, thus condition 2 in Lemma~\ref{lemma:convergenceUpperBound} is satisfied for any $\varepsilon \geq \varepsilon_0$, where we note that $\eta\varphi  -\frac{\rho h(\tau)}{\tau\varepsilon^2}$ increases with $\varepsilon$ when $\rho h(\tau) > 0$.

Suppose that there exists $\varepsilon > \varepsilon_0$ satisfying conditions 3 and 4 in Lemma~\ref{lemma:convergenceUpperBound}, so that all the conditions in Lemma~\ref{lemma:convergenceUpperBound} are satisfied. Applying Lemma~\ref{lemma:convergenceUpperBound} and considering (\ref{eq:convergenceUpperBoundFinalProof1}), we have
\begin{equation*}
F(\mathbf{w}(T))-F(\mathbf{w}^*) \leq \frac{1}{T\!\left(\eta\varphi \!-\!\frac{\rho h(\tau)}{\tau\varepsilon^2}\right)} < \frac{1}{T\!\left(\eta\varphi \! -\!\frac{\rho h(\tau)}{\tau\varepsilon_0^2}\right)} = \varepsilon_0
\end{equation*}
which contradicts with condition 4 in Lemma~\ref{lemma:convergenceUpperBound}. 
Therefore,  there \emph{does not} exist $\varepsilon > \varepsilon_0$ that satisfy both conditions 3 and 4 in Lemma~\ref{lemma:convergenceUpperBound}. This means that either 
1) $\exists k$ such that $ F\left(\mathbf{v}_{[k]}(k\tau)\right)-F(\mathbf{w}^*) \leq \varepsilon_0$ or 2) $ F\left(\mathbf{w}(T)\right)-F(\mathbf{w}^*) \leq \varepsilon_0$.
It follows that 
\begin{equation}
\min\left\{\min_{k=1,2,...,K} F\left(\mathbf{v}_{[k]}(k\tau)\right); F\left(\mathbf{w}(T)\right)\right\} -F(\mathbf{w}^*) \leq \varepsilon_0.
\label{eq:convergenceUpperBoundFinalProof3}
\end{equation}
From Theorem~\ref{theorem:wBound}, $F(\mathbf{w}{(k\tau)}) \leq F(\mathbf{v}_{[k]}(k\tau)) + \rho h(\tau)$ for any $k$. Combining with (\ref{eq:convergenceUpperBoundFinalProof3}), we get
\begin{equation*}
\min_{k=1,2,...,K} F\left(\mathbf{w}{(k\tau)}\right) -F(\mathbf{w}^*) \leq \varepsilon_0 + \rho h(\tau)
\end{equation*}
where we recall that $T=K\tau$.
Using (\ref{eq:wminDef}) and (\ref{eq:convergenceUpperBoundFinalProof2}), we obtain the result in (\ref{eq:convergenceUpperBoundFinal}).
\end{proof}

We note that the bound in (\ref{eq:convergenceUpperBoundFinal}) has no restriction on how the data is distributed at different nodes. The impact of different data distribution is captured by the gradient divergence $\delta$, which is included in $h(\tau)$.
It is easy to see from (\ref{eq:hFuncDef}) that $h(\tau)$ is non-negative, non-decreasing in $\tau$, and proportional to $\delta$. Thus, as one would intuitively expect, for a given total number of local update steps $T$, the optimality gap (i.e., $F(\mathbf{w}^\mathrm{f}) -F(\mathbf{w}^*)$) becomes larger when $\tau$ and $\delta$ are larger. For given $\tau$ and $\delta$, the optimality gap becomes smaller when $T$ is larger. 
When $\tau=1$, we have $h(\tau)=0$, and the optimality gap converges to zero as  $T \rightarrow \infty$.  When $\tau > 1$, we have $h(\tau) > 0$, and we can see from (\ref{eq:convergenceUpperBoundFinal}) that in this case, convergence is only guaranteed to a non-zero optimality gap as $T\rightarrow \infty$. This means that when we have unlimited budget for all types of resources (i.e., $R_m \rightarrow \infty, \forall m$), it is always optimal to set $\tau=1$ and perform global aggregation after every step of local update. However, when the resource budget $R_m$ is limited for some $m$, the training will be terminated after a finite number of iterations, thus the value of $T$ is finite. In this case, it may be better to perform global aggregation less frequently so that more resources can be used for local update, as we will see later in this paper.

\section{Control Algorithm}
\label{sec:schedulingAlgorithm}

We propose an algorithm that approximately solves (\ref{eq:optimizationProblem}) in this section. We first assume that the resource consumptions $c_m$ and $b_m$ ($\forall m$) are known, and we solve for the values of $\tau$ and $T$. Then, we consider practical scenarios where $c_m$, $b_m$, and some other parameters are unknown and may vary over time, and we propose a control algorithm that estimates the parameters and dynamically adjusts the value of $\tau$ in real time.

\subsection{Approximate Solution to (\ref{eq:optimizationProblem})}
\label{subsec:approxSolutionControlAlg}

We assume that $\eta$ is chosen small enough such that $\eta \leq \frac{1}{\beta}$, and use the upper bound in (\ref{eq:convergenceUpperBoundFinal}) as an approximation of $F(\mathbf{w}^\mathrm{f})-F(\mathbf{w}^*)$. Because for a given global loss function $F(\mathbf{w})$, its minimum value $F(\mathbf{w}^*)$ is a constant, the minimization of $F(\mathbf{w}^\mathrm{f})$ in (\ref{eq:optimizationProblem}) is equivalent to minimizing $F(\mathbf{w}^\mathrm{f})-F(\mathbf{w}^*)$.
With this approximation and rearranging the inequality constraints in (\ref{eq:optimizationProblem}), we can rewrite (\ref{eq:optimizationProblem}) as
\begin{align}
\min_{\tau, K \in \{1,2,3,...\}} & \,\, \frac{1}{2\eta\varphi T} + \sqrt{\frac{1}{4\eta^2\varphi^2 T^2} + \frac{\rho h(\tau)}{\eta\varphi\tau}} + \rho h(\tau) \label{eq:optimizationProblemApprox} \\
\textrm{s.t.} & \,\, K \leq \frac{R'_m }{c_m\tau + b_m}, \quad \forall m \in \{1,...,M\} \nonumber \\
& \,\, T = K\tau \nonumber
\end{align}
where $R'_m \triangleq R_m-b_m-c_m$.

It is easy to see that the objective function in (\ref{eq:optimizationProblemApprox}) decreases with $T$, thus it also decreases with $K$ because $T = K\tau$. Therefore, for any $\tau$, the optimal value of $K$ is $\left\lfloor \min_m \frac{R'_m }{c_m\tau + b_m} \right\rfloor$, i.e., the largest value of $K$ that does not violate any inequality constraint in (\ref{eq:optimizationProblemApprox}), where $\lfloor\cdot\rfloor$ denotes the floor function for rounding down to integer. To simplify the analysis, we approximate by ignoring the rounding operation and substituting $T = K\tau \approx \min_m \frac{R'_m \tau}{c_m\tau + b_m} =  1 \big/ \max_m \frac{c_m\tau + b_m}{R'_m \tau} $ into the objective function in (\ref{eq:optimizationProblemApprox}), yielding
\begin{equation}
G(\tau) \triangleq \frac{\max_m \!\frac{c_m\tau + b_m}{R'_m \tau}}{2\eta\varphi } + \sqrt{\!\frac{\left(\!\max_m \!\frac{c_m\tau + b_m}{R'_m \tau}\!\right)^{\!\!2}\!}{4\eta^2\varphi^2 }\! + \!\frac{\rho h(\!\tau\!)}{\eta\varphi\tau}} + \rho h(\!\tau\!) 
\label{eq:GTauDef}
\end{equation}
and we can define the (approximately) optimal $\tau$ as
\begin{equation}
\tau^* = \argmin_{\tau \in \{1,2,3,...\}} G(\tau)
\label{eq:optimizationGTau}
\end{equation}
from which we can directly obtain the (approximately) optimal $K$ as $K^* =\left\lfloor \min_m\frac{R'_m}{c_m \tau^* + b_m} \right\rfloor$, and the (approximately) optimal $T$ as $T^* = K^* \tau^* =\left\lfloor \min_m \frac{R'_m}{c_m\tau^* + b_m} \right\rfloor \tau^*$.

\begin{proposition}
When $\eta \leq \frac{1}{\beta}$, $\rho > 0$, $\beta > 0$, $\delta > 0$, we have $\lim_{R_{\textnormal{min}} \rightarrow \infty} \tau^* = 1$, where $R_{\textnormal{min}} \triangleq \min_m R_m$.
\label{prop:TauOptOneInfR}
\end{proposition}
\begin{proof}
Because $R_{\textnormal{min}} \rightarrow \infty \iff R_{m} \rightarrow \infty, \forall m \iff R'_{m} \rightarrow \infty, \forall m$, we have $\lim_{R_{\textnormal{min}} \rightarrow \infty} \max_m \frac{c_m\tau + b_m}{R'_m \tau} = \max_m \lim_{R'_{m} \rightarrow \infty} \frac{c_m\tau + b_m}{R'_m \tau} = 0$. Thus,
$\lim_{R_{\textnormal{min}} \rightarrow \infty} G(\tau) =  \sqrt{\frac{\rho h(\tau)}{\eta\varphi\tau}} + \rho h(\tau)$.
Let $B \triangleq \eta\beta + 1$. With a slight abuse of notation, we consider continuous values of $\tau \geq 1$. We have
\begin{align*}
d \left(\frac{ h(\tau)}{\tau}\right) \bigg/ d\tau & = \frac{\delta}{\beta\tau^2}\left( B^\tau \log B^\tau - (B^\tau - 1) \right) \\
& \geq \frac{\delta}{\beta\tau^2}\left( B^\tau \left( 1 - \frac{1}{B^\tau} \right) - B^\tau - 1 \right) \geq 0
\end{align*}
where the first inequality is from a lower bound of logarithmic function \cite{topsok2006some}. We also have
\begin{align*}
\frac{d h(\tau)}{d\tau} & = \frac{\delta}{\beta} (B^\tau \log B - \eta\beta) \geq \frac{\delta}{\beta} \left(\frac{2\eta\beta B^\tau}{2+\eta\beta}  - \eta\beta \right) \\
& = \frac{\delta (2\eta\beta B^\tau - 2\eta\beta - \eta^2\beta^2 )}{\beta (2+\eta\beta )} \\
& \geq \frac{\delta (2\eta\beta B - 2\eta\beta - \eta^2\beta^2 )}{\beta (2+\eta\beta )} = \frac{\delta \eta^2 \beta^2 }{\beta (2+\eta\beta )} > 0
\end{align*}
where the first inequality is from a lower bound of $\log B$ \cite{topsok2006some}, the second inequality is because $B > 1$ and $\tau \geq 1$.

Thus, for any $\tau \geq 1$, $h(\tau)$ increases with $\tau$, and $\frac{h(\tau)}{\tau}$ is non-decreasing with $\tau$. We also note that $\sqrt{x}$ increases with $x$ for any $x\geq 0$, and $h(1)=0$. It follows that $\lim_{R_{\textnormal{min}} \rightarrow \infty} G(\tau)$ increases with $\tau$  for any $\tau \geq 1$. Hence, $\lim_{R_{\textnormal{min}} \rightarrow \infty} \tau^* = 1$.
\end{proof}

Combining Proposition~\ref{prop:TauOptOneInfR} with Theorem~\ref{theorem:convergenceUpperBoundFinal}, we know that using $\tau^*$ found from (\ref{eq:optimizationGTau}) guarantees convergence with zero optimality gap as $R_\textnormal{min} \rightarrow \infty$ (and thus $R'_m \rightarrow \infty, \forall m$ and $T^* \rightarrow \infty$), because $\lim_{R_\textnormal{min} \rightarrow \infty} \tau^* = 1$ and $h(1)=0$.
For general values of $R_m$ (and $R'_m$), we have the following result.

\begin{proposition}
When $\eta \! \leq \!\frac{1}{\beta}$, $\rho \!>\! 0$, $\beta \!>\! 0$, $\delta \!>\! 0$, there exists a \emph{finite} value $\tau_0$, which only depends on $\eta$, $\beta$, $\rho$, $\delta$, $\varphi$, $c_m$, $b_m$, $R'_m$ ($\forall m$), such that $\tau^*  \leq \tau_0$. The quantity $\tau_0$ is defined as
\begin{align*}
\tau_0   \triangleq  \max \Bigg\{ &  \max_m  \frac{b_m R'_\nu - b_\nu R'_m}{c_\nu R'_m -c_m R'_\nu}  \, ;  \,\, \frac{\varphi(2 \!+\! \eta\beta)}{2 \rho\delta}\!  \left( \frac{2c_\nu b_\nu}{C_2} \!+\! \frac{2b_\nu^2}{C_2} \right) \!; \\
&   \frac{1}{\rho\delta\eta \log B} \left(\!\frac{b_\nu}{C_1} \! + \!\rho\eta\delta\!\right) \!  -\! \frac{1}{\eta\beta} ; \,\,\,\, \frac{1}{\eta\beta} + \frac{1}{2} \Bigg\}
\end{align*}
where index $\nu \triangleq \arg\max_{m\in V} \frac{b_m}{R'_m}$ (set $V \triangleq \arg\max_m \frac{c_m}{R'_m}$), $B \triangleq \eta\beta + 1$, $C_1 \triangleq 2\eta\varphi R'_\nu$, $C_2 \triangleq 4\eta^2\varphi^2 R'^2_\nu$. Here, for convenience, we allow $\arg\max$ to interchangeably return a set and an arbitrary value in that set, we also define $\frac{0}{0} \triangleq 0$. 

We also note that $0 < \eta\beta \leq 1$, thus $\tau_0 \geq \frac{1}{\eta\beta} + \frac{1}{2} > 1$.
\label{prop:TauOptBounded}
\end{proposition}
\begin{proof}
We can show that $\max_m  \frac{b_m R'_\nu - b_\nu R'_m}{c_\nu R'_m -c_m R'_\nu}$ is finite according to the definition of $\nu$ and $\frac{0}{0}$, then it is easy to see that $\tau_0$ is finite. We then show $\arg\max_m \frac{c_m\tau + b_m}{R'_m \tau} = \nu$ for any $\tau > \tau_0$, in which case the maximization over $m$ in  (\ref{eq:GTauDef}) becomes fixing $m=\nu$. Then, the proof separately considers the terms inside and outside the square root in (\ref{eq:GTauDef}). It shows that the first order derivatives of both parts are always larger than zero when $\tau > \tau_0$. Because the square root is an increasing function, $G(\tau)$ increases with $\tau$ for $\tau > \tau_0$, and thus $\tau^* \leq \tau_0$. 
\if\citetechreport1
See our online technical report \cite[Appendix~\ref{append:proofGTauOptBounded}]{JournalTechReport} for details.
\else
See Appendix~\ref{append:proofGTauOptBounded} for details.
\fi
\end{proof}

There is no closed-form solution for $\tau^*$ because $G(\tau)$ includes both polynomial and exponential terms of $\tau$, where the exponential term is embedded in $h(\tau)$. Because $\tau^*$ can only be a positive integer, according to Proposition~\ref{prop:TauOptBounded}, we can compute $G(\tau)$ within a finite range of $\tau$ to find $\tau^*$ that minimizes $G(\tau)$.

\subsection{Adaptive Federated Learning}

\label{subsec:controlAlgImpl}

In this subsection, we present the complete control algorithm for adaptive federated learning, which recomputes $\tau^*$ in every global aggregation step based on the most recent system state.
We use the theoretical results above to guide the design of the algorithm.

As mentioned earlier, the local updates run on edge nodes and the global aggregation is performed through the assistance of an aggregator, where the aggregator is a logical component that may also run on one of the edge nodes. The complete procedures at the aggregator and each edge node are presented in Algorithms~\ref{alg:protocolAggregator} and \ref{alg:protocolEdgeNode}, respectively, where Lines~\ref{alg:protocolEdgeNode:startLocalUpdt}--\ref{alg:protocolEdgeNode:endLocalUpdt} of Algorithm~\ref{alg:protocolEdgeNode} are for local updates and the rest is considered as part of global aggregation, initialization, or final operation. We assume that the aggregator initiates the learning process, and the initial model parameter $\mathbf{w}(0)$ is sent by the aggregator to all edge nodes. 
We note that instead of transmitting the entire model parameter vector in every global aggregation step, one can also transmit compressed or quantized model parameters to further save the communication bandwidth, where the compression or quantization can be performed using techniques described in~\cite{Jakub2016,hardy2017distributed}, for instance.

\subsubsection{Estimation of Parameters in $G(\tau)$}
\label{subsec:ControlAlgParamEst}

The expression of $G(\tau)$, which includes $h(\tau)$, has parameters which need to be estimated in practice.
Among these parameters, $c_m$ and $b_m$ ($\forall m$) are related to resource consumption, $\rho$, $\beta$, and $\delta$ are related to the loss function characteristics. These parameters are estimated in real time during the learning process. 

The values of $c_m$ and $b_m$ ($\forall m$) are estimated based on measurements of resource consumptions at the edge nodes and the aggregator (Line~\ref{alg:protocolAggregator:resourceConsumptionEst} of Algorithm~\ref{alg:protocolAggregator}). The  estimation depends on the type of resource under consideration. For example, when the type-$m$ resource is energy, the sum energy consumption (per local update) at all nodes is considered as $c_m$; when the type-$m$ resource is time, the maximum computation time (per local update) at all nodes is considered as $c_m$.
The aggregator also monitors the total resource consumption of each resource type $m$ based on the estimates, and compares the total resource consumption against the resource budget $R_m$ (Line~\ref{alg:protocolAggregator:checkRemainingResource} of Algorithm~\ref{alg:protocolAggregator}). If the consumed resource is at the budget limit for some $m$, it stops the learning and returns the final result.

The values of $\rho$, $\beta$, and $\delta$ are estimated based on the local and global losses and gradients computed at $\mathbf{w}(t)$ and $\mathbf{w}_i(t)$, see Line~\ref{alg:protocolAggregator:receiveBetaGrad} and Lines~\ref{alg:protocolAggregator:firstEstimation}--\ref{alg:protocolAggregator:estDelta} of Algorithm~\ref{alg:protocolAggregator} and Lines~\ref{alg:protocolEdgeNode:RhoEstimate}, \ref{alg:protocolEdgeNode:BetaEstimate}, and \ref{alg:protocolEdgeNode:sendBetaAndGrad} of Algorithm~\ref{alg:protocolEdgeNode}. To perform the estimation, each edge node needs to have access to both its local model parameter $\mathbf{w}_i(t)$ and the global model parameter $\mathbf{w}(t)$ for the same iteration $t$ (see Lines~\ref{alg:protocolEdgeNode:RhoEstimate} and \ref{alg:protocolEdgeNode:BetaEstimate} of Algorithm~\ref{alg:protocolEdgeNode}), which is only possible when global aggregation is performed in iteration~$t$. 
Because $\mathbf{w}(t)$ is only observable by each node after global aggregation,
estimated values of $\rho$, $\beta$, and $\delta$ are only available for recomputing $\tau^*$ starting from the second global aggregation step after initialization, which uses estimates obtained in the previous global aggregation step\footnote{See the condition in Line~\ref{alg:protocolAggregator:t0Cond} of Algorithm~\ref{alg:protocolAggregator} and Lines~\ref{alg:protocolEdgeNode:tCond} and \ref{alg:protocolEdgeNode:t0Cond} of Algorithm~\ref{alg:protocolEdgeNode}. Also note that the parameters $\hat{\rho}_i$, $\hat{\beta}_i$, $F_{i}(\mathbf{w}(t_0))$, $\nabla F_{i}(\mathbf{w}(t_0))$ sent in Line~\ref{alg:protocolEdgeNode:sendBetaAndGrad} of Algorithm~\ref{alg:protocolEdgeNode} are obtained at the previous global aggregation step ($t_0$, $\hat{\rho}_i$, and $\hat{\beta}_i$ are obtained in Lines~\ref{alg:protocolEdgeNode:t0Def}--\ref{alg:protocolEdgeNode:BetaEstimate} of Algorithm~\ref{alg:protocolEdgeNode}).}.

\begin{algorithm}[t]
\caption{Procedure at the aggregator} 
\label{alg:protocolAggregator} 
{\footnotesize

\KwIn{Resource budget $R$, control parameter $\varphi$, search range parameter $\gamma$, maximum $\tau$ value $\tau_\textrm{max}$}

\KwOut{$\mathbf{w}^\mathrm{f}$}

Initialize $\tau^* \leftarrow 1$, $t \leftarrow 0$, $s \leftarrow 0$; \quad //$s$ is a resource counter

Initialize $\mathbf{w}(0)$ as a constant or a random vector;

Initialize $\mathbf{w}^\mathrm{f} \leftarrow \mathbf{w}(0)$;

\RepeatNoEnd
{
    Send $\mathbf{w}(t)$ and $\tau^*$ to all edge nodes, also send \texttt{STOP} if it is set; \label{alg:protocolAggregator:sendWTau} 
    
    $t_0 \leftarrow t$; \quad //Save iteration index of last transmission of $\mathbf{w}(t)$

    $t \leftarrow t+\tau^*$; \quad //Next global aggregation is after $\tau$ iterations
    
    Receive $\mathbf{w}_i(t)$, $\hat{c}_i$ from each node $i$;
    
    Compute $\mathbf{w}(t)$ according to (\ref{eq:globalAverage});   \label{alg:protocolAggregator:computeWGlobal} 
    
    \If{$t_0 > 0$ \label{alg:protocolAggregator:t0Cond} }
    {
        Receive $\hat{\rho}_i$, $\hat{\beta}_i$, $F_{i}(\mathbf{w}(t_0))$, $\nabla F_{i}(\mathbf{w}(t_0))$ from each node $i$;  \label{alg:protocolAggregator:receiveBetaGrad}
        
        Compute $F(\mathbf{w}(t_0))$ according to (\ref{eq:globalLossFuncAllSamples})
        
        \If{$F(\mathbf{w}(t_0)) < F(\mathbf{w}^\mathrm{f})$   \label{alg:protocolAggregator:minStepStart}}
        {
            $\mathbf{w}^\mathrm{f} \leftarrow \mathbf{w}(t_0)$;   \label{alg:protocolAggregator:minStepEnd}
        }

        \If{\texttt{STOP} flag is set}
        {
            \textbf{break}; \quad //Break out of the loop here if \texttt{STOP} is set
        }

        Estimate $\hat{\rho} \leftarrow \frac{\sum_{i=1}^{N}D_{i}\hat{\rho}_{i}}{D}$;  \label{alg:protocolAggregator:firstEstimation}

        Estimate $\hat{\beta} \leftarrow \frac{\sum_{i=1}^{N}D_{i}\hat{\beta}_{i}}{D}$;
                        
        Compute $\nabla F(\mathbf{w}(t_0)) \leftarrow \frac{\sum_{i=1}^{N}D_{i}\nabla F_{i}(\mathbf{w}(t_0))}{D}$,
        estimate $\hat{\delta}_i \leftarrow \Vert \nabla F_{i}(\mathbf{w}(t_0)) - \nabla F(\mathbf{w}(t_0))  \Vert$ for each $i$,
        from which we estimate $\hat{\delta} \leftarrow \frac{\sum_{i=1}^{N}D_{i}\hat{\delta}_{i}}{D}$;   \label{alg:protocolAggregator:estDelta}
        
        Compute new value of $\tau^*$ according to (\ref{eq:optimizationGTau}) via linear search on integer values of $\tau$ within $[1, \tau_\textrm{m}]$, where we set $\tau_\textrm{m} \leftarrow \min\{\gamma\tau^*; \tau_\textrm{max}\}$; \label{alg:protocolAggregator:TauCompute} 
    }
    
    \For{$m = 1,2,...,M$}
    {
        Estimate resource consumptions $\hat{c}_m$, $\hat{b}_m$, using $\hat{c}_{m,i}$ received from all nodes $i$ and local measurements at the aggregator; \label{alg:protocolAggregator:resourceConsumptionEst} 

        $s_m \leftarrow s_m + \hat{c}_m\tau + \hat{b}_m$;
    }

    \If{$\exists m$ such that $s_m + \hat{c}_m(\tau + 1) + 2\hat{b}_m \geq R_m$ \label{alg:protocolAggregator:checkRemainingResource} }
    {
        Decrease $\tau^*$ to the maximum possible value such that the estimated resource consumption for remaining iterations is within budget $R_m$ for all $m$, set \texttt{STOP} flag;   \label{alg:protocolAggregator:decreaseTauToRemainWithinBudget}
    }

}

Send $\mathbf{w}(t)$ to all edge nodes;   \label{alg:protocolAggregator:FinalStart}

Receive $F_{i}(\mathbf{w}(t))$ from each node $i$;

Compute $F(\mathbf{w}(t))$ according to (\ref{eq:globalLossFuncAllSamples})
        
\If{$F(\mathbf{w}(t)) < F(\mathbf{w}^\mathrm{f})$ \label{alg:protocolAggregator:FinalLossCompare}}
{
    $\mathbf{w}^\mathrm{f} \leftarrow \mathbf{w}(t)$;   \label{alg:protocolAggregator:FinalEnd}
}

}

\end{algorithm}

\begin{algorithm}[t]
\caption{Procedure at each edge node $i$} 
\label{alg:protocolEdgeNode} 
{\footnotesize

Initialize $t \leftarrow 0$;

\Repeat{\texttt{STOP} flag is received}{
    Receive $\mathbf{w}(t)$ and new $\tau^*$ from aggregator, set $\widetilde{\mathbf{w}}_i(t) \leftarrow \mathbf{w}(t)$;   \label{alg:protocolEdgeNode:receiveWTauSetTilde} 
    
    $t_0 \leftarrow t$; \quad //Save iteration index of last transmission of $\mathbf{w}(t)$ \label{alg:protocolEdgeNode:t0Def}

    \If{$t>0$    \label{alg:protocolEdgeNode:tCond} }
    {
        Estimate $\hat{\rho}_i \leftarrow \left\Vert F_{i}(\mathbf{w}_i(t)) - F_{i}(\mathbf{w}(t))  \right\Vert / \left\Vert \mathbf{w}_i(t) - \mathbf{w}(t)  \right\Vert$; \label{alg:protocolEdgeNode:RhoEstimate} 

        Estimate $\hat{\beta}_i \leftarrow \left\Vert \nabla F_{i}(\mathbf{w}_i(t)) - \nabla F_{i}(\mathbf{w}(t))  \right\Vert / \left\Vert \mathbf{w}_i(t) - \mathbf{w}(t)  \right\Vert$; \label{alg:protocolEdgeNode:BetaEstimate}

    }
    
    \For{$\mu = 1,2,..., \tau^*$  \label{alg:protocolEdgeNode:startLocalUpdt} }
    {
        $t \leftarrow t+1$; \quad //Start of next iteration

        Perform local update and obtain $\mathbf{w}_i(t)$ according to (\ref{eq:localUpdate});   \label{alg:protocolEdgeNode:LocalUpdateGradDescent} 

        \If{$\mu < \tau^*$}
        {
            Set $\widetilde{\mathbf{w}}_i(t) \leftarrow \mathbf{w}_i(t)$;  \label{alg:protocolEdgeNode:endLocalUpdt}   \label{alg:protocolEdgeNode:LocalUpdateTildeSet} 
        }
    }  
    
    \For{$m = 1,2,...,M$}
    {
        Estimate type-$m$ resource consumption $\hat{c}_{m,i}$ for one local update at node $i$;
    }
    
    Send $\mathbf{w}_i(t)$, $\hat{c}_{m,i}$ ($\forall m$) to aggregator;

    \If{$t_0 > 0$ \label{alg:protocolEdgeNode:t0Cond} }
    {
        Send $\hat{\rho}_i$, $\hat{\beta}_i$, $F_{i}(\mathbf{w}(t_0))$, $\nabla F_{i}(\mathbf{w}(t_0))$ to aggregator; \label{alg:protocolEdgeNode:sendBetaAndGrad} 

    }
}

Receive $\mathbf{w}(t)$ from aggregator;   \label{alg:protocolEdgeNode:FinalStart}

Send $F_{i}(\mathbf{w}(t))$ to aggregator;    \label{alg:protocolEdgeNode:FinalEnd}

}
\end{algorithm}

\emph{Remark:} In the extreme case where $\mathbf{w}_i (t) = \mathbf{w}(t)$ in Lines~\ref{alg:protocolEdgeNode:RhoEstimate} and \ref{alg:protocolEdgeNode:BetaEstimate} of Algorithm~\ref{alg:protocolEdgeNode}, we estimate $\hat{\rho}_i$ and $\hat{\beta}_i$ as zero. When $\delta = \beta = 0$ and $\frac{\delta}{\beta}$ in $h(\tau)$ is undefined, we define that $h(\tau) = 0$ for all $\tau \geq 1$. This is because for $t>0$, $\mathbf{w}_i (t) = \mathbf{w}(t)$ only occurs when different nodes have extremely similar (often equal) datasets, in which case a large value of $\tau$ does not make the convergence worse than a small value of $\tau$, thus it makes sense to define $h(\tau) = 0$ in this case.

The parameter $\eta$ is the gradient-descent step size which is pre-specified and known. The remaining parameter $\varphi$ includes $\omega$ which is non-straightforward to estimate because the algorithm does not know $\mathbf{w}^*$, thus we regard $\varphi$ as a control parameter that is manually chosen and remains fixed for the same machine learning model\footnote{Although $\varphi$ is related to $\beta$ and we estimate $\beta$ separately, we found that it is good to keep $\varphi$ a constant value that does not vary with the estimated value of $\beta$ in practice, because there can be occasions where the estimated $\beta$ is large causing $\varphi < 0$, which causes abnormal behavior when computing $\tau^*$ from $G(\tau)$.}.
Experimentation results presented in the next section show that a fixed value of $\varphi$ works well across different data distributions, various numbers of nodes, and various resource consumptions/budgets. 
If we multiply both sides of (\ref{eq:GTauDef}) by $\varphi$, we can see that a larger value of $\varphi$ gives a higher weight to the terms with $h(\tau)$, yielding a smaller value of $\tau^*$ (because $h(\tau)$ increases with $\tau$), and vice versa. Therefore, in practice, it is not hard to tune the value of $\varphi$ on a small and simple setup, which can then be applied to general cases. See also the results on the sensitivity of $\varphi$ in Section~\ref{subsubsec:sensitivityOfPhiResult}.

\subsubsection{Recomputing $\tau^*$}

The value of $\tau^*$ is recomputed by the aggregator during each global aggregation step, based on the most updated parameter estimations. 
When searching for $\tau^*$, we use the following search range instead of the range in Proposition~\ref{prop:TauOptBounded} due to practical considerations of estimation error.
As shown in Line~\ref{alg:protocolAggregator:TauCompute} of Algorithm~\ref{alg:protocolAggregator}, we search for new values of $\tau^*$ up to $\gamma$ times the current value of $\tau^*$, and find $\tau^*$ that minimizes $G(\tau)$, where $\gamma > 0$ is a fixed parameter. The presence of $\gamma$ limits the search space and also avoids $\tau^*$ from growing too quickly as initial parameter estimates may be inaccurate. We also impose a maximum value of $\tau$, denoted by $\tau_\textrm{max}$, because if $\tau^*$ is too large, it is more likely for the system to operate beyond the resource budget due to inaccuracies in the estimation of local resource consumption, see Line~\ref{alg:protocolAggregator:checkRemainingResource} of Algorithm~\ref{alg:protocolAggregator}. 
The new value of $\tau^*$ is sent to each node together with $\mathbf{w}(t)$ (Line~\ref{alg:protocolAggregator:sendWTau} of Algorithm~\ref{alg:protocolAggregator}).

\subsubsection{Distributed Gradient Descent}
The local update steps of distributed gradient descent at the edge node include Lines~\ref{alg:protocolEdgeNode:startLocalUpdt}--\ref{alg:protocolEdgeNode:endLocalUpdt} of Algorithm~\ref{alg:protocolEdgeNode}, where Line~\ref{alg:protocolEdgeNode:LocalUpdateGradDescent} of Algorithm~\ref{alg:protocolEdgeNode} corresponds to Line~\ref{alg:distGradDescent:LocalUpdate} of Algorithm~\ref{alg:distGradDescent} and Line~\ref{alg:protocolEdgeNode:LocalUpdateTildeSet} of Algorithm~\ref{alg:protocolEdgeNode} corresponds to Line~\ref{alg:distGradDescent:TildeSet} of Algorithm~\ref{alg:distGradDescent}. When global aggregation is performed, Line~\ref{alg:protocolAggregator:computeWGlobal} of Algorithm~\ref{alg:protocolAggregator} computes the global model parameter $\mathbf{w}(t)$ at the aggregator, which is sent to the edge nodes in Line~\ref{alg:protocolAggregator:sendWTau} of Algorithm~\ref{alg:protocolAggregator}, and each edge node receives $\mathbf{w}(t)$ in Line~\ref{alg:protocolEdgeNode:receiveWTauSetTilde} of Algorithm~\ref{alg:protocolEdgeNode} and sets $\widetilde{\mathbf{w}}_i(t) \leftarrow \mathbf{w}(t)$ to use $\mathbf{w}(t)$ as the initial model parameter for the next round of local update; this corresponds to Line~\ref{alg:distGradDescent:TildeSetGlobalAgg} of Algorithm~\ref{alg:distGradDescent}. 

The final model parameter $\mathbf{w}^\mathrm{f}$ that minimizes $F(\mathbf{w})$ is obtained at the aggregator in Lines~\ref{alg:protocolAggregator:minStepStart}--\ref{alg:protocolAggregator:minStepEnd} of Algorithm~\ref{alg:protocolAggregator}, corresponding to Line~\ref{alg:distGradDescent:minStep} of Algorithm~\ref{alg:distGradDescent}.
As discussed in Section~\ref{sec:ProblemFormulation}, the computation of $\mathbf{w}^\mathrm{f}$ lags for one round of global aggregation, because for any iteration $t_0$ that includes a global aggregation step, $F(\mathbf{w}(t_0))$ can only be computed after each edge node has received $\mathbf{w}(t_0)$ and sent the local loss $F_i(\mathbf{w}(t_0))$ to the aggregator in the next round of global aggregation. To take into account the final value of $\mathbf{w}(t)$ in the computation of $\mathbf{w}^\mathrm{f}$, Lines~\ref{alg:protocolAggregator:FinalStart}--\ref{alg:protocolAggregator:FinalEnd} of Algorithm~\ref{alg:protocolAggregator} and Lines~\ref{alg:protocolEdgeNode:FinalStart}--\ref{alg:protocolEdgeNode:FinalEnd} of Algorithm~\ref{alg:protocolEdgeNode} perform an additional round of computation of the loss and $\mathbf{w}^\mathrm{f}$, as also discussed in Section~\ref{sec:ProblemFormulation}.

Overall, when global aggregation is executed for $K$ times in total, the computational complexity of Algorithm~\ref{alg:protocolAggregator} is $O(K(NM+\tau_\textrm{max}))$, because each global aggregation step includes the computation of global parameters from the local parameters collected from $N$ different nodes for $M$ resource types and the linear search step in Line~\ref{alg:protocolAggregator:TauCompute}  of Algorithm~\ref{alg:protocolAggregator} which has at most $\tau_\textrm{max}$ steps. When $T$ steps of local updates are performed in total, Algorithm~\ref{alg:protocolEdgeNode} has a computational complexity of $O(T+KM)$, where the additional term $KM$ corresponds to the additional local processing (at each node) in global aggregation steps.

\subsection{Extension to Stochastic Gradient Descent}

When the amount of training data is large, it is usually computationally prohibitive to compute the gradient of the loss function defined on the entire (local) dataset. In such cases, stochastic gradient descent (SGD) is often used \cite{shalev2014understanding,Goodfellow-et-al-2016,bottou2010large}, which uses the gradient computed on the loss function defined on a randomly sampled subset (referred to as a mini-batch) of data to approximate the real gradient. Although the theoretical analysis in this paper is based on deterministic gradient descent (DGD), the proposed approach can be directly extended to SGD. As discussed in \cite{FUSION2018}, SGD can be seen as an approximation to DGD.

When using SGD with our proposed algorithm, all losses and their gradients are computed on mini-batches. Each local iteration step corresponds to a step of gradient descent where the gradient is computed on a mini-batch of local training data. The mini-batch changes for every step of local iteration, i.e., for each new local iteration, a new mini-batch of a given size is randomly selected from the local training data. However, to reduce errors introduced by random data sampling when estimating the parameters $\rho$, $\beta$, and $\delta$, the first iteration after global aggregation uses the same mini-batch as the last iteration before global aggregation. When $\tau=1$, the mini-batch changes if the same mini-batch has already been used in two iterations, to ensure that different mini-batches are used for training over time.

To avoid approximation errors caused by mini-batch sampling when determining $\mathbf{w}^\mathrm{f}$, when using SGD, the aggregator informs the edge nodes whether the current $\mathbf{w}(t_0)$ is selected as $\mathbf{w}^\mathrm{f}$ using an additional flag sent together with the message in Line~\ref{alg:protocolAggregator:sendWTau} of Algorithm~\ref{alg:protocolAggregator}. The edge nodes save their own copies of $\mathbf{w}^\mathrm{f}$. When an edge node computes $F_i(\mathbf{w}(t_0))$ that is sent in Line~\ref{alg:protocolEdgeNode:sendBetaAndGrad}  of Algorithm~\ref{alg:protocolEdgeNode}, it also recomputes $F_i(\mathbf{w}^\mathrm{f})$ using the same mini-batch as for computing $F_i(\mathbf{w}(t_0))$. It then sends both $F_i(\mathbf{w}^\mathrm{f})$ and $F_i(\mathbf{w}(t_0))$ to the aggregator in Line~\ref{alg:protocolEdgeNode:sendBetaAndGrad}  of Algorithm~\ref{alg:protocolEdgeNode}. The aggregator recomputes $F(\mathbf{w}^\mathrm{f})$ based on the most recently received $F_i(\mathbf{w}^\mathrm{f})$. In this way, the values of $F(\mathbf{w}^\mathrm{f})$ and $F(\mathbf{w}(t_0))$ used for the comparison in Lines~\ref{alg:protocolAggregator:minStepStart} and \ref{alg:protocolAggregator:FinalLossCompare} of Algorithm~\ref{alg:protocolAggregator} are computed on the same mini-batch at each edge node.

\section{Experimentation Results}
\label{sec:experimentation}

\subsection{Setup}

To evaluate the performance of our proposed adaptive federated learning algorithm, we conducted experiments both on networked prototype system with $5$ nodes and in a simulated environment with the number of nodes varying from $5$ to $500$.
The prototype system consists of three Raspberry Pi (version 3) devices and two laptop computers, which are all interconnected via Wi-Fi in an office building. 
This represents an edge computing environment where the computational capabilities of edge nodes are heterogeneous. All these $5$ nodes have local datasets on which model training is conducted. The aggregator is located on one of the laptop computers, and hence co-located with one of the local datasets.

\subsubsection{Resource Definition}
\label{subsec:experimentation:resource}

For ease of presentation and interpretation of results, we let $M=1$ and consider time as the single resource type in our experiments.
For the prototype system, we train each model for a fixed amount of time budget. The values of $c$ and $b$ (we omit the subscript $m=1$ for simplicity) correspond to the actual time used for each local update and global aggregation, respectively.
The simulation environment performs model training with simulated resource consumptions, which are randomly generated according to Gaussian distribution with mean and standard deviation values
\if\citetechreport1
(see \cite[Appendix~\ref{append:ExperimentSimParam}]{JournalTechReport} for these values)
\else
(see Appendix~\ref{append:ExperimentSimParam} for these values)
\fi
obtained from measurements of the squared-SVM model on the prototype. See Section~\ref{subsec:experimentation:modelsDatasets} below for definitions of models and datasets.

\subsubsection{Baselines}
\label{subsec:experimentation:baseline}

We compare with the following baseline approaches:
\begin{enumerate}
\item[(a)] Centralized gradient descent \cite{shalev2014understanding, Goodfellow-et-al-2016}, where the entire training dataset is stored on a single edge node and the model is trained directly on that node using a standard (centralized) gradient descent procedure;
\item[(b)] Canonical federated learning approach presented in~\cite{mcmahan2016communication}, which is equivalent to using a fixed (non-adaptive) value of $\tau$ in our setting;
\item[(c)] Synchronous distributed gradient descent \cite{Chen2016}, which is equivalent to fixing $\tau = 1$ in our setting.
\end{enumerate}

For a fair comparison, we implement the estimation of resource consumptions for all baselines and the training stops when we have reached the resource (time) budget.
When conducting experiments on the prototype system, the centralized gradient descent is performed on a Raspberry Pi device. To avoid resource consumption related to loss computation, centralized gradient descent uses the last model parameter $\mathbf{w}(T)$ (instead of $\mathbf{w}^\mathrm{f}$) as the result, because convergence of $\mathbf{w}(T)$ can be proven in the centralized case \cite{convex}.
We do not explicitly distinguish the baselines (b) and (c) above because they both correspond to an approach with non-adaptive $\tau$ of a certain value. When $\tau$ is non-adaptive, we use the same protocol as in Algorithms~\ref{alg:protocolAggregator} and \ref{alg:protocolEdgeNode}, but remove any parts related to parameter estimation and recomputation of $\tau$.

\subsubsection{DGD and SGD}
We consider both DGD and SGD in the experiments to evaluate the general applicability of the proposed algorithm. For SGD, the mini-batch sampling uses the same initial random seed at all nodes, which means that when the datasets at all nodes are identical, the mini-batches at all nodes are also identical in the same iteration (while they are generally different across different iterations). This setup is for a better consideration of the differences between equal and non-equal data distributions (see Section~\ref{subsec:DataDistributionCases} below).

\begin{figure*}
    \centering
    \begin{subfigure}{0.2\textwidth}
        \centering
        \includegraphics[width=0.8\linewidth]{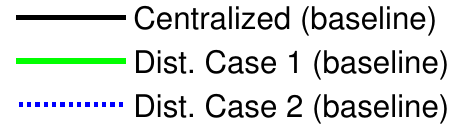}
    \end{subfigure}%
    ~
    \begin{subfigure}{0.2\textwidth}
        \centering
        \includegraphics[width=0.8\linewidth]{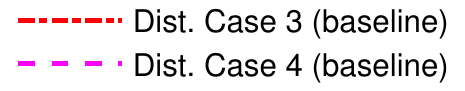}
    \end{subfigure}%
    ~
    \begin{subfigure}{0.2\textwidth}
        \centering
        \includegraphics[width=0.75\linewidth]{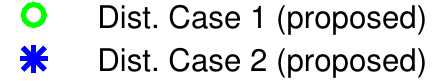}
    \end{subfigure}%
    ~
    \begin{subfigure}{0.2\textwidth}
        \centering
        \includegraphics[width=0.75\linewidth]{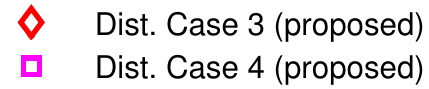}
    \end{subfigure}%
    \vspace{0.05in}
        
    \begin{subfigure}[b]{0.17\textwidth}
        \centering
        \includegraphics[width=1\linewidth]{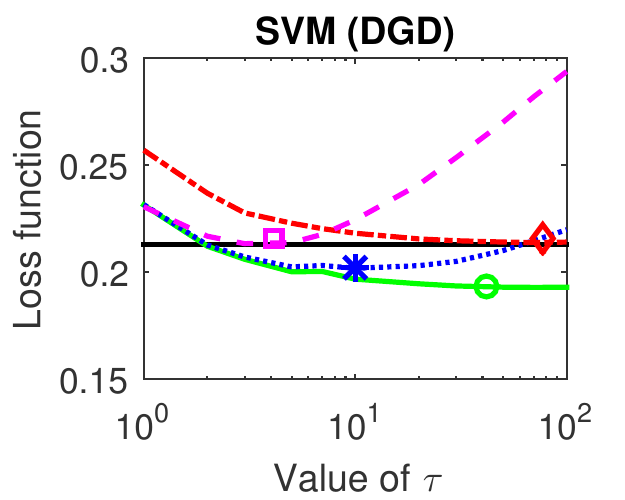}
    \end{subfigure}%
    ~\hspace{-0.1in}
    \begin{subfigure}[b]{0.17\textwidth}
        \centering
        \includegraphics[width=1\linewidth]{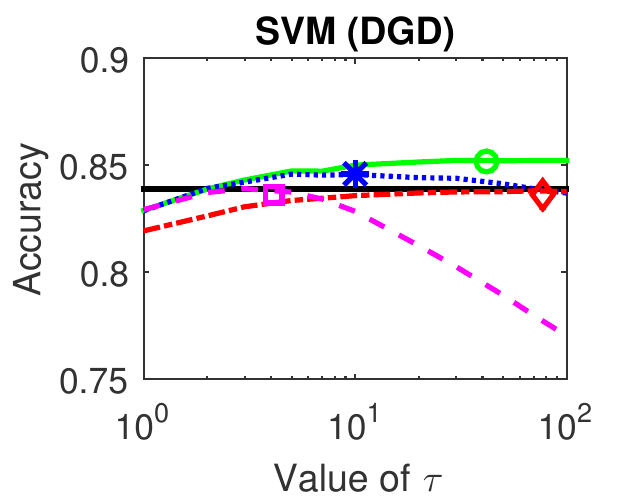}
    \end{subfigure}%
    ~\hspace{-0.1in}
    \begin{subfigure}[b]{0.17\textwidth}
        \centering
        \includegraphics[width=1\linewidth]{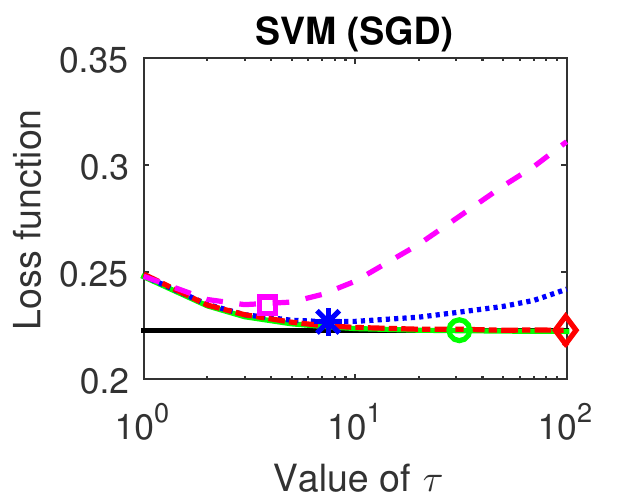}
    \end{subfigure}%
    ~\hspace{-0.1in}
    \begin{subfigure}[b]{0.17\textwidth}
        \centering
        \includegraphics[width=1\linewidth]{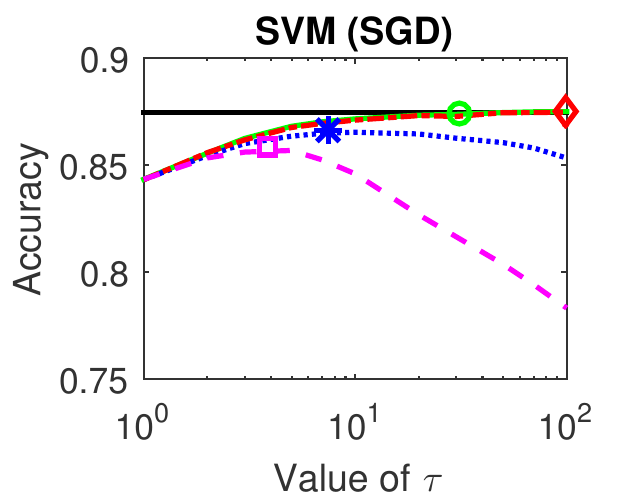}
    \end{subfigure}%
    ~\hspace{-0.1in}
    \begin{subfigure}[b]{0.17\textwidth}
        \includegraphics[width=1\linewidth]{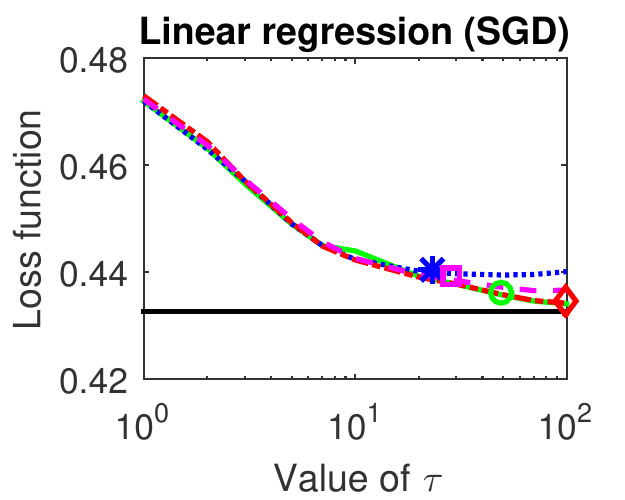}
    \end{subfigure}%
    ~\hspace{-0.1in}
    \begin{subfigure}[b]{0.17\textwidth}
        \includegraphics[width=1\linewidth]{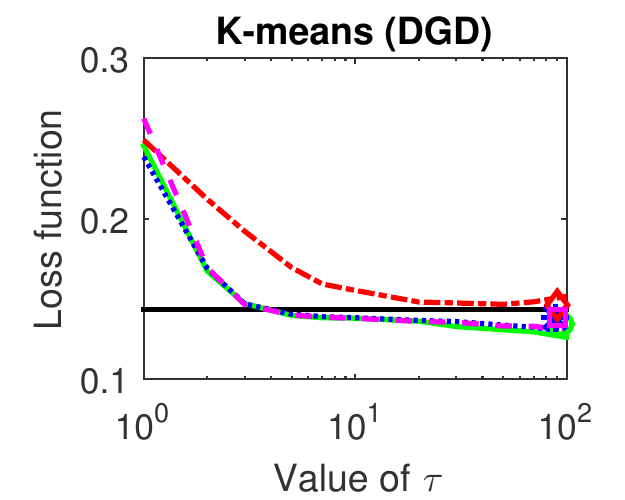}
    \end{subfigure}%
    \vspace{0.05in}

    \begin{subfigure}[b]{0.17\textwidth}
        \centering
        \includegraphics[width=1\linewidth]{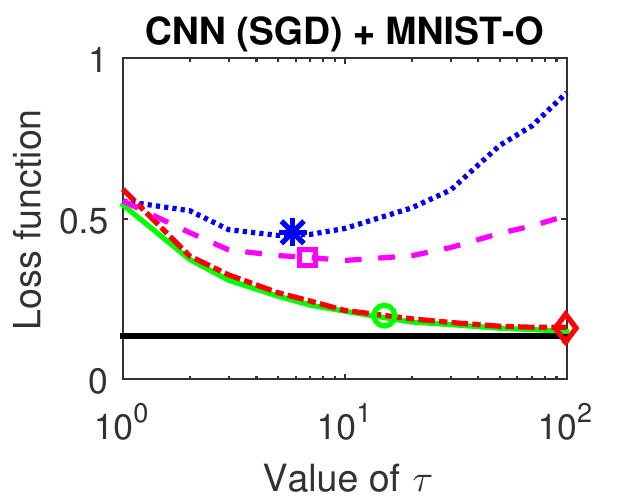}
    \end{subfigure}%
    ~\hspace{-0.1in}
    \begin{subfigure}[b]{0.17\textwidth}
        \centering
        \includegraphics[width=1\linewidth]{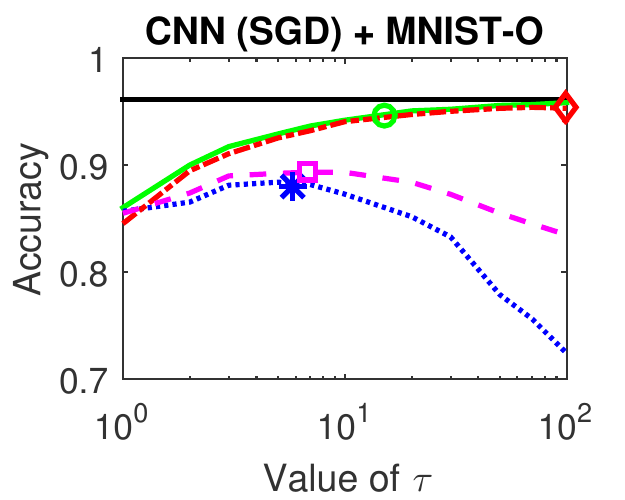}
    \end{subfigure}%
    ~\hspace{-0.1in}
    \begin{subfigure}[b]{0.17\textwidth}
        \centering
        \includegraphics[width=1\linewidth]{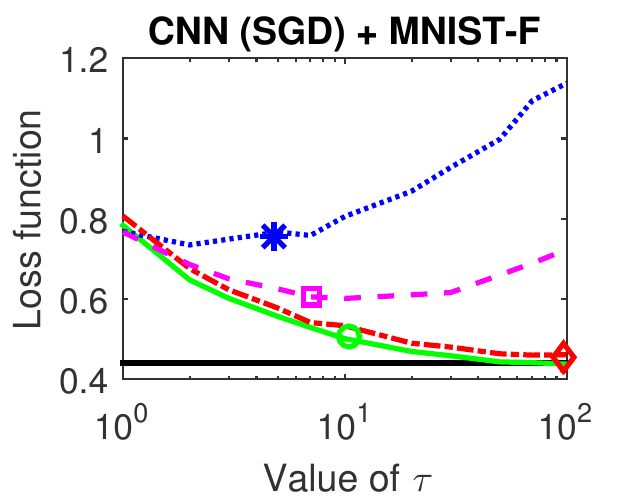}
    \end{subfigure}%
    ~\hspace{-0.1in}
    \begin{subfigure}[b]{0.17\textwidth}
        \centering
        \includegraphics[width=1\linewidth]{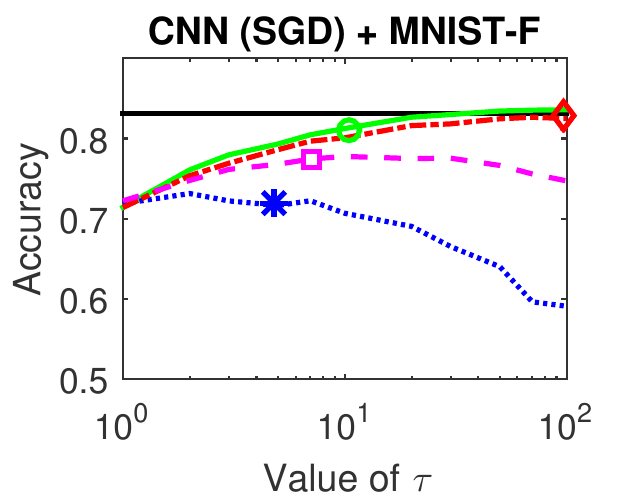}
    \end{subfigure}%
    ~\hspace{-0.1in}
    \begin{subfigure}[b]{0.17\textwidth}
        \centering
        \includegraphics[width=1\linewidth]{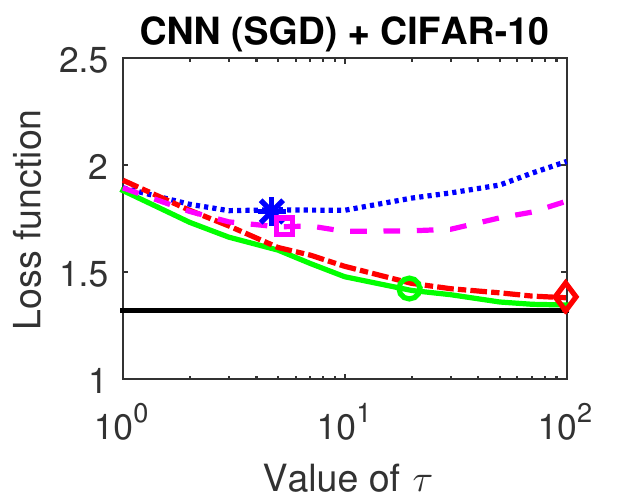}
    \end{subfigure}%
    ~\hspace{-0.1in}
    \begin{subfigure}[b]{0.17\textwidth}
        \centering
        \includegraphics[width=1\linewidth]{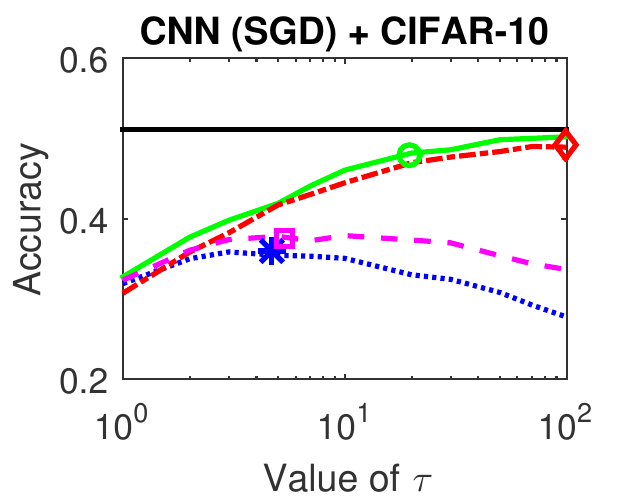}
    \end{subfigure}%

\caption{Loss function values and classification accuracy with different $\tau$. Only SVM and CNN classifiers have accuracy values. The curves show the results from the baseline with different fixed values of $\tau$. Our proposed solution (represented by a single marker for each case) gives an average $\tau$ and loss/accuracy that is close to the optimum in all cases.}
\label{fig:LossAndAccuracyExperimentResults}
\vspace{-0.2in}
\end{figure*}

\subsubsection{Models and Datasets}
\label{subsec:experimentation:modelsDatasets}

We evaluate the training of four different models on five different datasets, which represent a large variety of both small and large models and datasets, as one can expect all these variants to exist in edge computing scenarios. 
The models include squared-SVM, linear regression, K-means, and deep convolutional neural networks (CNN)\footnote{The CNN has $9$ layers with the following structure: $5 \times 5 \times 32$ Convolutional $\rightarrow$ $2 \times 2$ MaxPool $\rightarrow$ Local Response Normalization $\rightarrow$ $5 \times 5 \times 32$ Convolutional $\rightarrow$ Local Response Normalization $\rightarrow$ $2 \times 2$ MaxPool $\rightarrow$ $z \times 256$ Fully connected $\rightarrow$ $256 \times 10$ Fully connected $\rightarrow$ Softmax, where $z$ depends on the input image size and $z=1568$ for MNIST-O and MNIST-F and $z=2048$ for CIFAR-10. This configuration is similar to what is suggested in the TensorFlow tutorial \cite{TensorFlowTutorial}.}. See Table~\ref{tab:learningModels} for a summary of the loss functions of these models, and see \cite{shalev2014understanding, Goodfellow-et-al-2016,bottou2010large} for more details. Among them, the loss functions for squared-SVM (which we refer to as SVM in short in the following) and linear regression satisfy Assumption~\ref{assumption:Convex}, whereas the loss functions for K-means and CNN are non-convex and thus do not satisfy Assumption~\ref{assumption:Convex}.

SVM is trained on the original MNIST dataset (referred to as \emph{MNIST-O}) \cite{lecun1998gradient}, which contains gray-scale images of $70,000$ handwritten digits ($60,000$ for training and $10,000$ for testing). The SVM outputs a binary label that corresponds to whether the digit is even or odd.
We consider both DGD and SGD variants of SVM. The DGD variant only uses $1,000$ training and $1,000$ testing data samples out of the entire dataset in each simulation round, because DGD cannot process a large amount of data. The SGD variant uses the entire MNIST dataset.

Linear regression is performed with SGD on the energy dataset \cite{CANDANEDO201781}, which contains $19,735$ records of measurements from multiple sensors and the energy consumptions of appliances and lights. The model learns to predict the appliance energy consumption from sensor measurements.

K-means is performed with DGD on the user knowledge modeling dataset \cite{Kahraman}, which has $403$ samples each with $5$ attributes summarizing the user interaction with a web environment. The samples can be grouped into $4$ clusters representing different knowledge levels, but we assume that we do not have prior knowledge of this grouping.

CNN is trained using SGD on three different datasets, including MNIST-O as described above, the fashion MNIST dataset (referred to as \emph{MNIST-F}) which has the same format as MNIST-O but includes images of fashion items instead of digits \cite{FashionMNIST}, and the CIFAR-10 dataset which includes $60,000$ color images  ($50,000$ for training and $10,000$ for testing) of $10$ different types of objects \cite{CIFAR10}.
A separate CNN model is trained on each dataset, to perform multi-class classification among the $10$ different labels in the dataset.

\subsubsection{Data Distribution at Different Nodes (Cases 1--\,4)}
\label{subsec:DataDistributionCases}
For the distributed settings, we consider four different ways of distributing the data into different nodes.
In \emph{Case~1}, each data sample is randomly assigned to a node, thus each node has uniform (but not full) information.
In \emph{Case~2}, all the data samples in each node have the same label\footnote{When there are more labels than nodes, each node may have data with more than one label, but the number of labels at each node is no more than the total number of labels divided by the total number of nodes rounded to the next integer.}. This represents the case where each node has non-uniform information, because the entire dataset has samples with multiple different labels.
In \emph{Case~3}, each node has the entire dataset (thus full information).
In \emph{Case~4}, data samples with the first half of the labels are distributed to the first half of the nodes as in Case 1; the other samples are distributed to the second half of the nodes as in Case 2. This represents a combined uniform and non-uniform case.
For datasets that do not have ground truth labels, such the energy dataset used with linear regression, the data to node assignment is based on labels generated from an unsupervised clustering approach.

\subsubsection{Training and Control Parameters}
\label{subsec:experiments:trainingCtrlParams}
In all our experiments, we set the search range parameter $\gamma=10$, the maximum $\tau$ value $\tau_\textrm{max} = 100$. Unless otherwise specified, we set the control parameter $\varphi=0.025$ for SVM, linear regression, and K-means, and $\varphi=5\times 10^{-5}$ for CNN. The gradient descent step size is $\eta=0.01$. The resource (time) budget is set as $R = 15$ seconds unless otherwise specified. Except for the instantaneous results in Section~\ref{subsec:results:instantaneous}, the average results of $15$ independent experiment/simulation runs are shown.

\subsection{Results}

\subsubsection{Loss and Accuracy Values}
\label{sec:LossAndAccuracyFromExperiment}
In our first set of experiments, the SVM, linear regression, and K-means models were trained on the prototype system. 
Due to the resource limitation of Raspberry Pi devices, the CNN model was trained in a simulated environment of $5$~nodes, with resource consumptions generated in the way described in Section~\ref{subsec:experimentation:resource}.

We compare the loss function values of our proposed algorithm (with adaptive $\tau$) to baseline approaches, and also compare the classification accuracies for the SVM and CNN classifiers.
The results are shown in Fig.~\ref{fig:LossAndAccuracyExperimentResults}. We note that \emph{the proposed approach only has one data point (represented by a single marker in the figure) in each case}, because the value of $\tau$ is adaptive in this case and the marker location shows the average $\tau^*$ with the corresponding loss or accuracy. The centralized case also only has one data point but we show a flat line across different values of $\tau$ for the ease of comparison. We see that the proposed approach performs close to the optimal point for all cases and all models\footnote{Note that the loss and accuracy values shown in Fig.~\ref{fig:LossAndAccuracyExperimentResults} can be improved if we allow a longer training time. For example, the accuracy of CNN on MNIST data can become close to $1.0$ if we allow a long enough time for training. The goal of our experiments here is to show that our proposed approach can operate close to the optimal point with a \emph{fixed and limited} amount of training time (resource budget) as defined in Section~\ref{subsec:experiments:trainingCtrlParams}. }. We also see that the (empirically) optimal value of $\tau$ is different for different cases and models, so a fixed value of $\tau$ does not work well for all cases. In some cases, the distributed approach can perform better than the centralized approach, because for a given amount of time budget, federated learning is able to make use of the computation resource at multiple nodes. 
For DGD approaches, Case~3 does not perform as well as Case~1, because the amount of data at each node in Case~3 is larger than that in Case~1, and DGD processes the entire amount of data thus Case~3 requires more resource  for each local update.

Due to the high complexity of evaluating CNN models and the fact that linear regression and K-means models do not provide accuracy values, we focus on the SVM model in the following and provide further insights on the system.

\subsubsection{Varying Number of Nodes}

\begin{figure}
    \centering
    \begin{subfigure}{0.45\textwidth}
        \centering
	   \begin{subfigure}{0.45\textwidth}
             \includegraphics[width=1\linewidth]{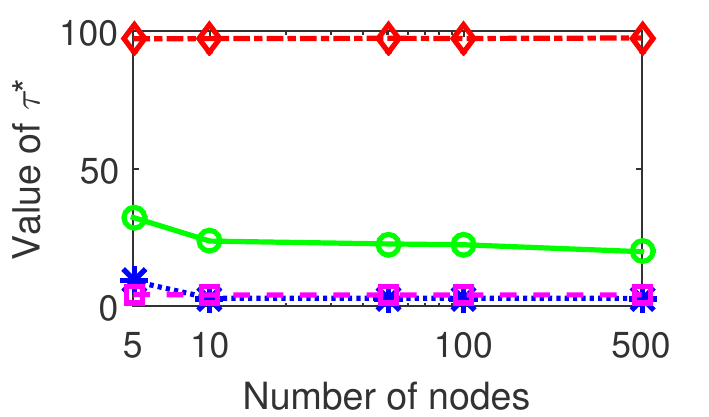}
        \end{subfigure}%
        ~\,\,\,\,
        \begin{subfigure}{0.18\textwidth}
             \includegraphics[width=1\linewidth]{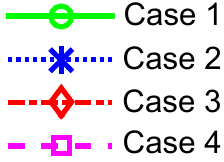}
             \vspace{0.2in}
        \end{subfigure}%
        \caption{$\tau^*$ in proposed algorithm}
    \end{subfigure}%
    \vspace{0.1in}

    \begin{subfigure}{0.45\textwidth}
        \centering
	    \begin{subfigure}{0.24\textwidth}
	        \centering
	        \hspace{1\linewidth}
	    \end{subfigure}%
	    ~
	    \begin{subfigure}{0.24\textwidth}
	        \centering
	        \includegraphics[width=0.9\linewidth]{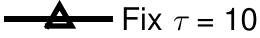}
	    \end{subfigure}%
	    ~
	    \begin{subfigure}{0.24\textwidth}
	        \centering
	        \includegraphics[width=0.9\linewidth]{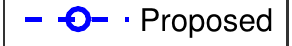}
	    \end{subfigure}%
	    ~
	    \begin{subfigure}{0.24\textwidth}
	        \centering
	        \hspace{1\linewidth}
	    \end{subfigure}%

	    \begin{subfigure}[b]{0.035\textwidth}
	        \centering
	        \includegraphics[width=1\linewidth]{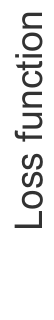}
	    \end{subfigure}%
	    ~
	    \begin{subfigure}[b]{0.23\textwidth}
	        \centering
	        \includegraphics[width=1\linewidth]{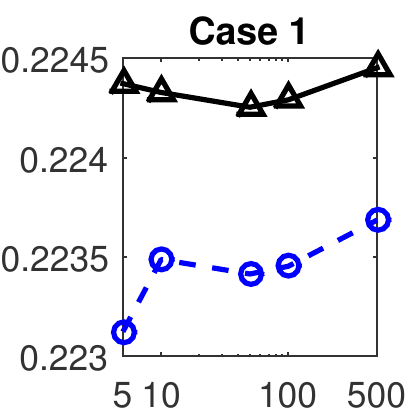}
	    \end{subfigure}%
	    ~
	    \begin{subfigure}[b]{0.23\textwidth}
	        \centering
	        \includegraphics[width=1\linewidth]{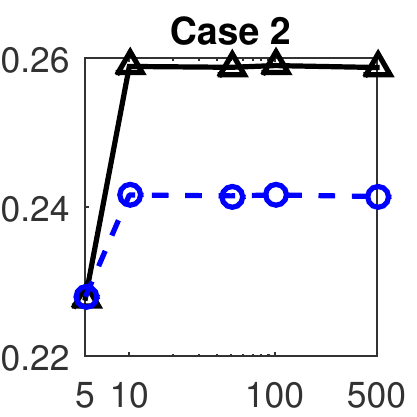}
	    \end{subfigure}%
	    ~
	    \begin{subfigure}[b]{0.23\textwidth}
	        \centering
	        \includegraphics[width=1\linewidth]{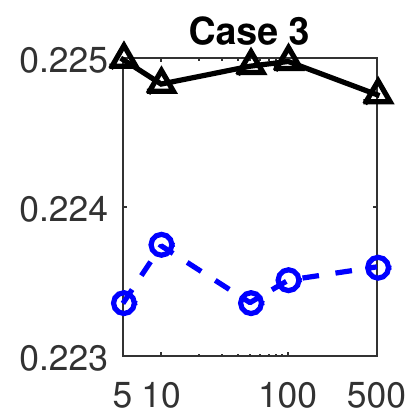}
	    \end{subfigure}%
	    ~
	    \begin{subfigure}[b]{0.23\textwidth}
	        \centering
	        \includegraphics[width=1\linewidth]{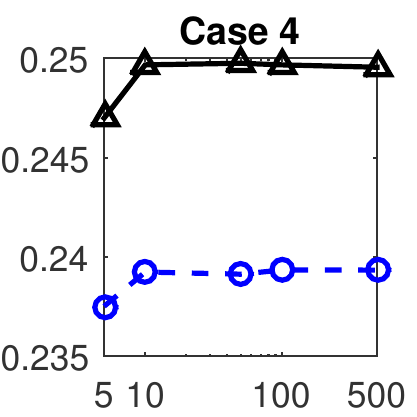}
	    \end{subfigure}%

	    \begin{subfigure}[b]{0.035\textwidth}
	        \centering
	        \includegraphics[width=1\textwidth]{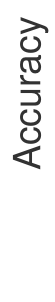}
	    \end{subfigure}%
	    ~
	    \begin{subfigure}[b]{0.23\textwidth}
	        \centering
	        \includegraphics[width=1\textwidth]{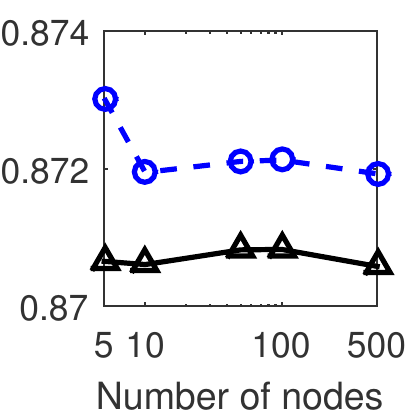}
	    \end{subfigure}%
	    ~
	    \begin{subfigure}[b]{0.23\textwidth}
	        \centering
	        \includegraphics[width=1\linewidth]{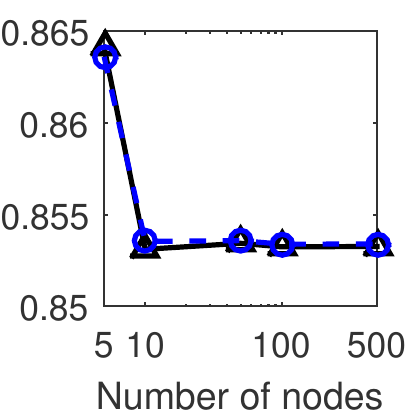}
	    \end{subfigure}%
	    ~
	    \begin{subfigure}[b]{0.23\textwidth}
	        \centering
	        \includegraphics[width=1\linewidth]{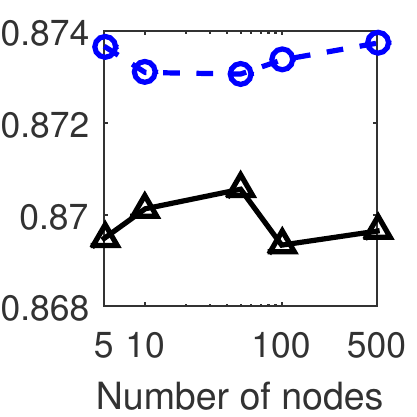}
	    \end{subfigure}%
	    ~
	    \begin{subfigure}[b]{0.23\textwidth}
	        \centering
	        \includegraphics[width=1\linewidth]{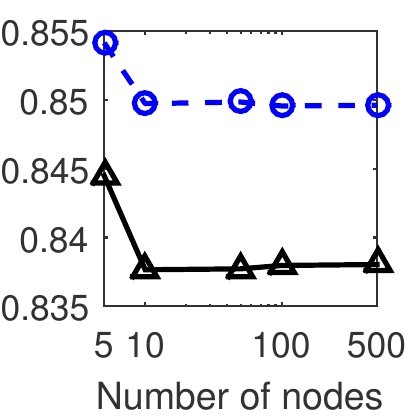}
	    \end{subfigure}%

        \caption{Loss function values and classification accuracy}
    \end{subfigure}%

\caption{SVM (SGD) with different numbers of nodes.}
\label{fig:MultiNodeExperimentResults}
\end{figure}

Results of SVM (SGD) for the number of nodes varying from $5$ to $500$ are shown in Fig.~\ref{fig:MultiNodeExperimentResults}, which are obtained in the simulated environment. Our proposed approach performs better than or similar to the fixed $\tau=10$ baseline in all cases, where we choose fixed $\tau=10$ as the baseline in this and the following evaluations because it is empirically a good value for non-adaptive $\tau$ in different cases according to the results in Fig.~\ref{fig:LossAndAccuracyExperimentResults}.

\subsubsection{Varying Global Aggregation Time}

\begin{figure}
    \centering
    \begin{subfigure}{0.45\textwidth}
        \centering
	   \begin{subfigure}{0.45\textwidth}
             \includegraphics[width=1\linewidth]{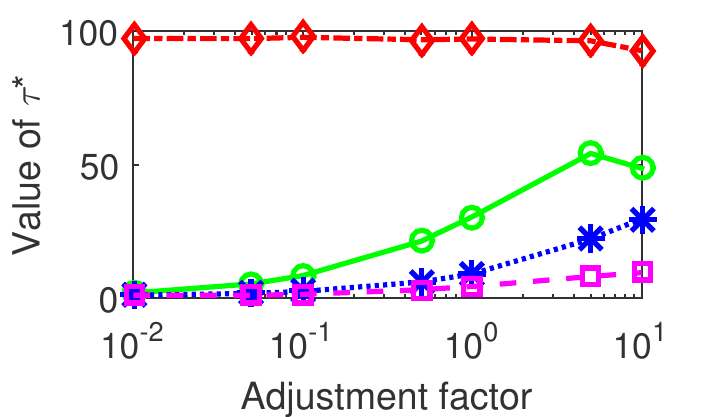}
        \end{subfigure}%
        ~\,\,\,\,
        \begin{subfigure}{0.18\textwidth}
             \includegraphics[width=1\linewidth]{Varying_Tau_Legend.pdf}
             \vspace{0.2in}
        \end{subfigure}%
        \caption{$\tau^*$ in proposed algorithm}
    \end{subfigure}%
    \vspace{0.1in}

    \begin{subfigure}{0.45\textwidth}
        \centering
	    \begin{subfigure}{0.24\textwidth}
	        \centering
	        \hspace{1\linewidth}
	    \end{subfigure}%
	    ~
	    \begin{subfigure}{0.24\textwidth}
	        \centering
	        \includegraphics[width=0.9\linewidth]{Varying_EachCase_Legend1.pdf}
	    \end{subfigure}%
	    ~
	    \begin{subfigure}{0.24\textwidth}
	        \centering
	        \includegraphics[width=0.9\linewidth]{Varying_EachCase_Legend2.pdf}
	    \end{subfigure}%
	    ~
	    \begin{subfigure}{0.24\textwidth}
	        \centering
	        \hspace{1\linewidth}
	    \end{subfigure}%

	    \begin{subfigure}[b]{0.035\textwidth}
	        \centering
	        \includegraphics[width=1\linewidth]{LossLabel.pdf}
	    \end{subfigure}%
	    ~
	    \begin{subfigure}[b]{0.23\textwidth}
	        \centering
	        \includegraphics[width=1\linewidth]{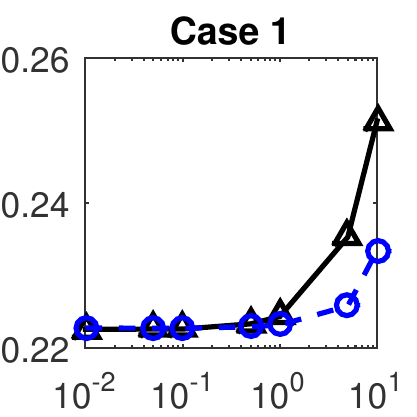}
	    \end{subfigure}%
	    ~
	    \begin{subfigure}[b]{0.23\textwidth}
	        \centering
	        \includegraphics[width=1\linewidth]{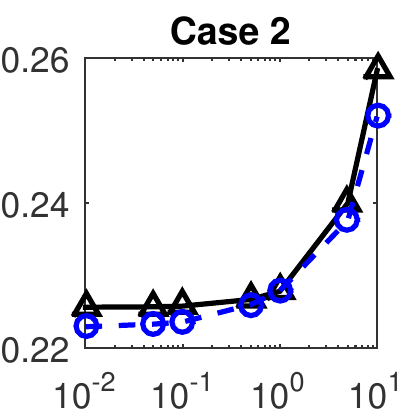}
	    \end{subfigure}%
	    ~
	    \begin{subfigure}[b]{0.23\textwidth}
	        \centering
	        \includegraphics[width=1\linewidth]{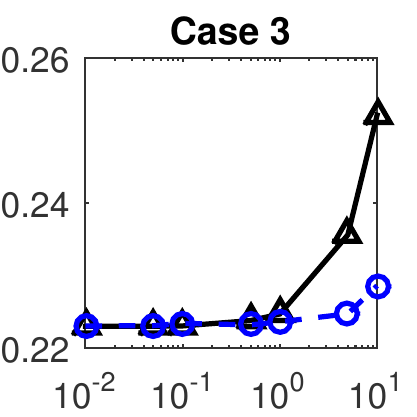}
	    \end{subfigure}%
	    ~
	    \begin{subfigure}[b]{0.23\textwidth}
	        \centering
	        \includegraphics[width=1\linewidth]{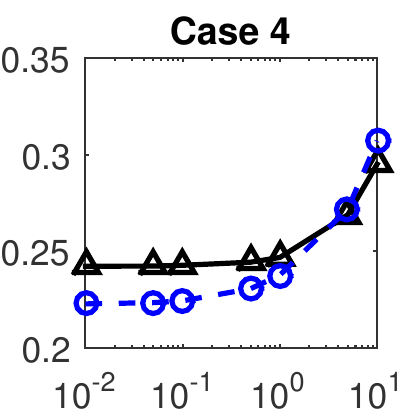}
	    \end{subfigure}%

	    \begin{subfigure}[b]{0.035\textwidth}
	        \centering
	        \includegraphics[width=1\textwidth]{AccuracyLabel.pdf}
	    \end{subfigure}%
	    ~
	    \begin{subfigure}[b]{0.23\textwidth}
	        \centering
	        \includegraphics[width=1\textwidth]{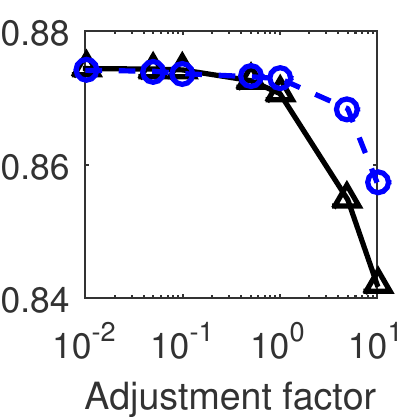}
	    \end{subfigure}%
	    ~
	    \begin{subfigure}[b]{0.23\textwidth}
	        \centering
	        \includegraphics[width=1\linewidth]{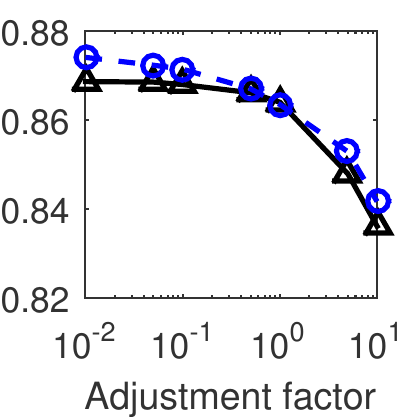}
	    \end{subfigure}%
	    ~
	    \begin{subfigure}[b]{0.23\textwidth}
	        \centering
	        \includegraphics[width=1\linewidth]{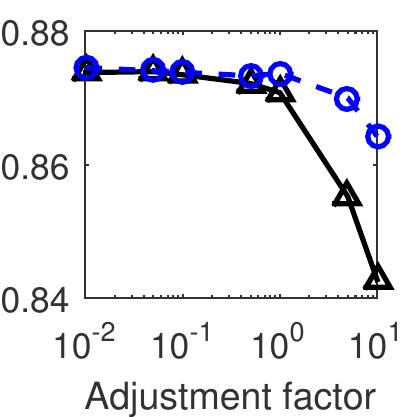}
	    \end{subfigure}%
	    ~
	    \begin{subfigure}[b]{0.23\textwidth}
	        \centering
	        \includegraphics[width=1\linewidth]{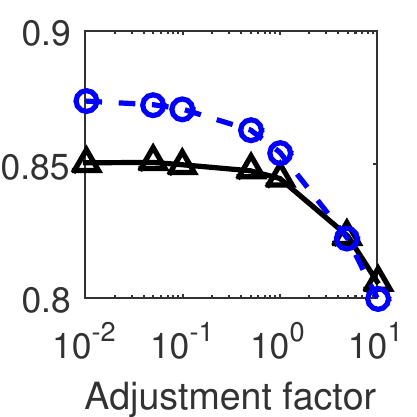}
	    \end{subfigure}%

        \caption{Loss function values and classification accuracy}
    \end{subfigure}%

\caption{SVM (SGD) with different global aggregation times.}
\label{fig:VaryingGlobalAggSGDExperimentResults}
\end{figure}

To study the impact of different resource consumption (time) for global aggregation, we modify the simulation environment so that the global aggregation time is scaled by an \emph{adjustment factor}. The actual time of global aggregation is equal to the original global aggregation time multiplied by the adjustment factor, thus a small adjustment factor corresponds to a small global aggregation time. The results for SVM (SGD) are shown in Fig.~\ref{fig:VaryingGlobalAggSGDExperimentResults}.
\if\citetechreport1
Additional results for SVM (DGD) are included in \cite[Appendix~\ref{append:VaryingGlobalAggDGDExperimentResults}]{JournalTechReport}.
\else
Additional results for SVM (DGD) are included in Appendix~\ref{append:VaryingGlobalAggDGDExperimentResults}.
\fi
We can see that as one would intuitively expect, a larger global aggregation time generally results in a larger $\tau^*$ for the proposed algorithm, because when it takes more time to perform global aggregation, the system should perform global aggregation less frequently, to make the best use of available time (resource). The fact that $\tau^*$ slightly decreases when the adjustment factor is large is because in this case, the global aggregation time is so large that only a few rounds of global aggregation can be performed before reaching the resource budget, and the value of $\tau^*$ will be decreased in the last round to remain within the resource budget (see Line~\ref{alg:protocolAggregator:decreaseTauToRemainWithinBudget} of Algorithm~\ref{alg:protocolAggregator}).
Comparing to the fixed $\tau=10$ baseline, the proposed algorithm performs better in (almost) all cases.

\subsubsection{Varying Total Time Budget}

\begin{figure}
    \centering
    \begin{subfigure}{0.45\textwidth}
        \centering
	   \begin{subfigure}{0.45\textwidth}
             \includegraphics[width=1\linewidth]{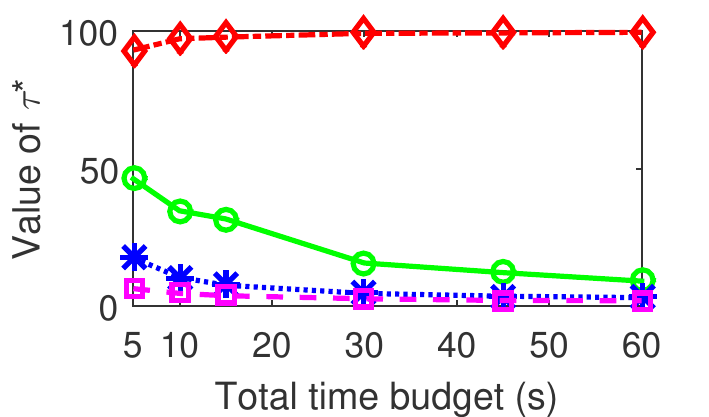}
        \end{subfigure}%
        ~\,\,\,\,
        \begin{subfigure}{0.18\textwidth}
             \includegraphics[width=1\linewidth]{Varying_Tau_Legend.pdf}
             \vspace{0.2in}
        \end{subfigure}%
        \caption{$\tau^*$ in proposed algorithm}
    \end{subfigure}%
    \vspace{0.1in}

    \begin{subfigure}{0.45\textwidth}
        \centering
	    \begin{subfigure}{0.24\textwidth}
	        \centering
	        \hspace{1\linewidth}
	    \end{subfigure}%
	    ~
	    \begin{subfigure}{0.24\textwidth}
	        \centering
	        \includegraphics[width=0.9\linewidth]{Varying_EachCase_Legend1.pdf}
	    \end{subfigure}%
	    ~
	    \begin{subfigure}{0.24\textwidth}
	        \centering
	        \includegraphics[width=0.9\linewidth]{Varying_EachCase_Legend2.pdf}
	    \end{subfigure}%
	    ~
	    \begin{subfigure}{0.24\textwidth}
	        \centering
	        \hspace{1\linewidth}
	    \end{subfigure}%

	    \begin{subfigure}[b]{0.035\textwidth}
	        \centering
	        \includegraphics[width=1\linewidth]{LossLabel.pdf}
	    \end{subfigure}%
	    ~
	    \begin{subfigure}[b]{0.23\textwidth}
	        \centering
	        \includegraphics[width=1\linewidth]{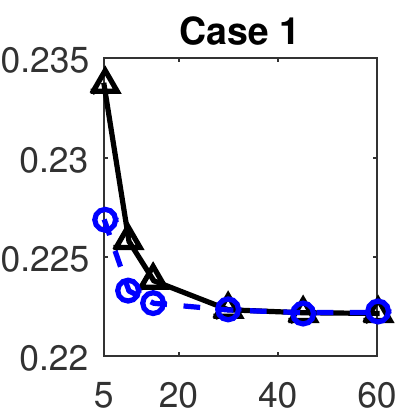}
	    \end{subfigure}%
	    ~
	    \begin{subfigure}[b]{0.23\textwidth}
	        \centering
	        \includegraphics[width=1\linewidth]{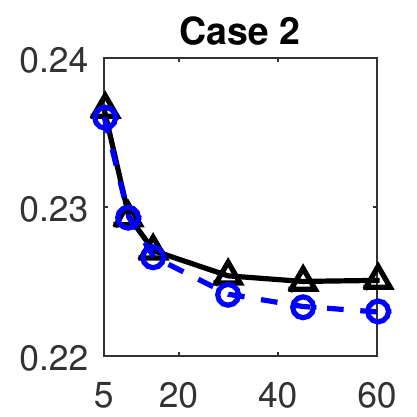}
	    \end{subfigure}%
	    ~
	    \begin{subfigure}[b]{0.23\textwidth}
	        \centering
	        \includegraphics[width=1\linewidth]{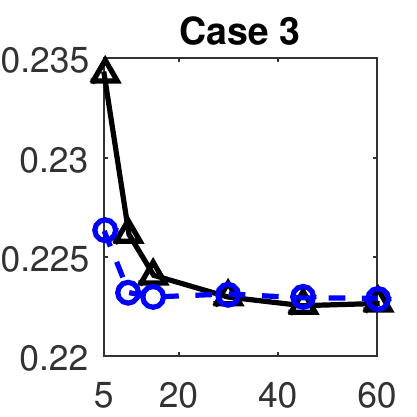}
	    \end{subfigure}%
	    ~
	    \begin{subfigure}[b]{0.23\textwidth}
	        \centering
	        \includegraphics[width=1\linewidth]{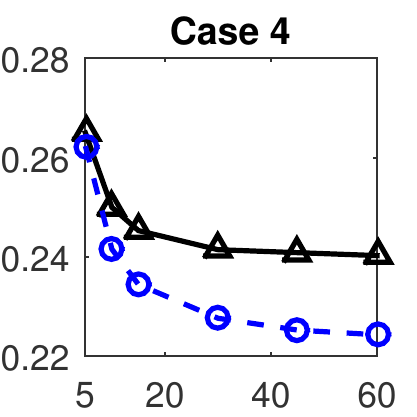}
	    \end{subfigure}%

	    \begin{subfigure}[b]{0.035\textwidth}
	        \centering
	        \includegraphics[width=1\textwidth]{AccuracyLabel.pdf}
	    \end{subfigure}%
	    ~
	    \begin{subfigure}[b]{0.23\textwidth}
	        \centering
	        \includegraphics[width=1\textwidth]{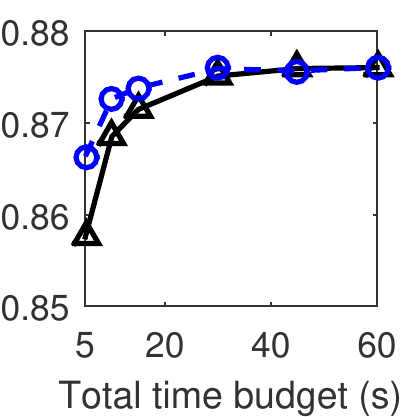}
	    \end{subfigure}%
	    ~
	    \begin{subfigure}[b]{0.23\textwidth}
	        \centering
	        \includegraphics[width=1\linewidth]{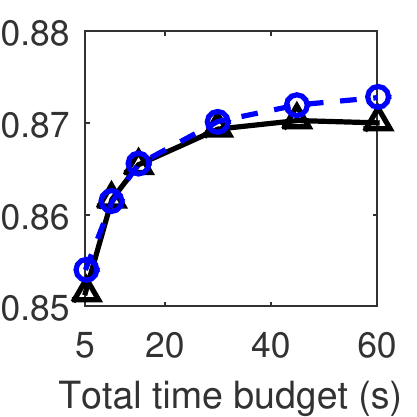}
	    \end{subfigure}%
	    ~
	    \begin{subfigure}[b]{0.23\textwidth}
	        \centering
	        \includegraphics[width=1\linewidth]{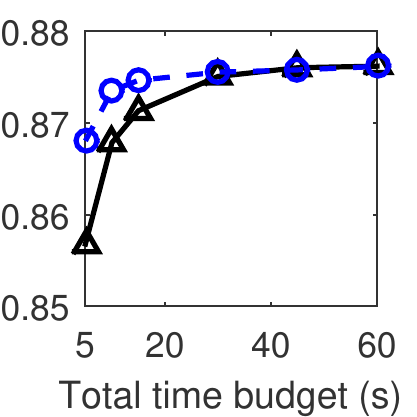}
	    \end{subfigure}%
	    ~
	    \begin{subfigure}[b]{0.23\textwidth}
	        \centering
	        \includegraphics[width=1\linewidth]{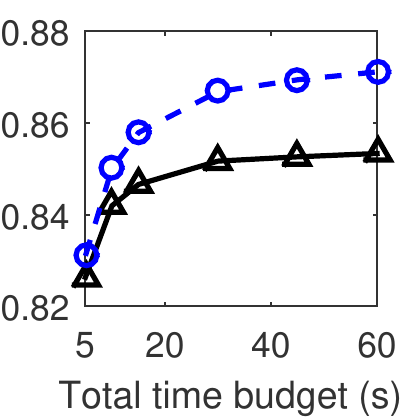}
	    \end{subfigure}%

        \caption{Loss function values and classification accuracy}
    \end{subfigure}%

\caption{SVM (SGD) with different total time budgets.}
\label{fig:VaryingTotalTimeSGDExperimentResults}
\end{figure}

We evaluate the impact of the total time (resource) budget on the prototype system. Results for SVM (SGD) are shown in Fig.~\ref{fig:VaryingTotalTimeSGDExperimentResults}.
\if\citetechreport1
Further results for SVM (DGD) are included in \cite[Appendix~\ref{append:VaryingTotalTimeDGDExperimentResults}]{JournalTechReport}.
\else
Further results for SVM (DGD) are included in Appendix~\ref{append:VaryingTotalTimeDGDExperimentResults}.
\fi
We see that except for Case 3 where all nodes have the same dataset, the value of $\tau^*$ of the proposed algorithm decreases with the total time budget. This aligns with the discussion in Section~\ref{subsec:approxSolutionControlAlg} that $\tau^*$ becomes close to one when the resource budget is large enough. We also see that the proposed algorithm performs better than or similar to the fixed $\tau=10$ baseline in all cases.

\subsubsection{Instantaneous Behavior}
\label{subsec:results:instantaneous}

We further study the instantaneous behavior of our system for a single run of $30$ seconds (for each case) on the prototype system.  Results for SVM (DGD) is shown in Fig.~\ref{fig:InstantSVMDGDExperimentResults}.
\if\citetechreport1
Further results for SVM (SGD) are available in \cite[Appendix~\ref{append:InstantSVMSGDExperimentResults}]{JournalTechReport}.
\else
Further results for SVM (SGD) are available in Appendix~\ref{append:InstantSVMSGDExperimentResults}. 
\fi
We see that the value of $\tau^*$ remains stable after an initial adaptation period, showing that the control algorithm is stable. The value of $\tau^*$ decreases at the end due to adjustment caused by the system reaching the resource budget (see Line~\ref{alg:protocolAggregator:decreaseTauToRemainWithinBudget} of Algorithm~\ref{alg:protocolAggregator}).
As expected, the gradient deviation $\delta$ is larger for Cases~2 and 4 where the data samples at different nodes are non-uniform. The same is observed for $\rho$ and $\beta$, indicating that the model parameter $\mathbf{w}$ is in a less smooth region for Cases~2 and 4. In Case 3, the data at different nodes are equal so we always have $\mathbf{w}_i (t) = \mathbf{w}(t)$ regardless of whether global aggregation is performed in iteration $t$. Thus, the estimated $\rho$ and $\beta$ values are zero by definition, as explained in the remark in Section~\ref{subsec:ControlAlgParamEst}.
Case 3 of SVM (DGD) has a much larger value of $c$ because it processes more data than in other cases and thus takes more time, as explained before. The value of $b$ exhibits fluctuations because of the randomness of the wireless channel.

\begin{figure}
    \centering
    \begin{subfigure}[b]{0.11\textwidth}
        \centering
        \includegraphics[width=0.9\linewidth]{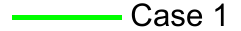}
    \end{subfigure}%
    \begin{subfigure}[b]{0.11\textwidth}
        \centering
        \includegraphics[width=0.9\linewidth]{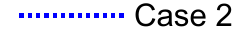}
    \end{subfigure}%
    \begin{subfigure}[b]{0.11\textwidth}
        \centering
        \includegraphics[width=0.9\linewidth]{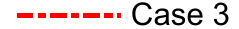}
    \end{subfigure}%
    \begin{subfigure}[b]{0.11\textwidth}
        \centering
        \includegraphics[width=0.9\linewidth]{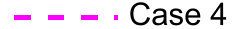}
    \end{subfigure}%

    \begin{subfigure}[b]{0.12\textwidth}
        \centering
        \includegraphics[width=1\linewidth]{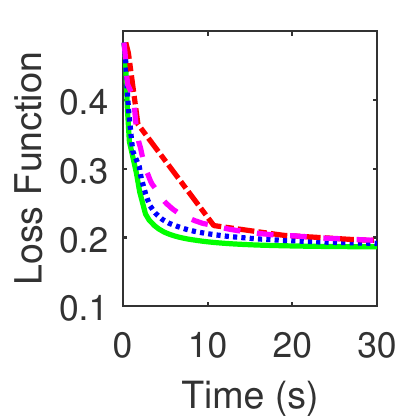}
    \end{subfigure}%
    ~\hspace{-0.1in}
    \begin{subfigure}[b]{0.12\textwidth}
        \centering
        \includegraphics[width=1\linewidth]{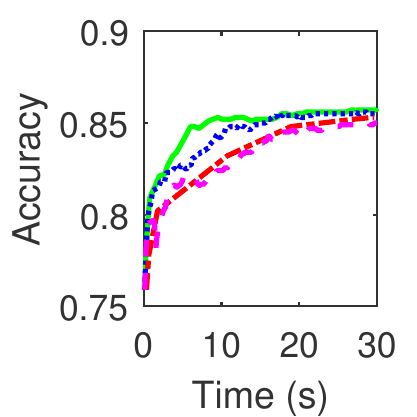}
    \end{subfigure}%
    ~\hspace{-0.1in}
    \begin{subfigure}[b]{0.12\textwidth}
        \centering
        \includegraphics[width=1\linewidth]{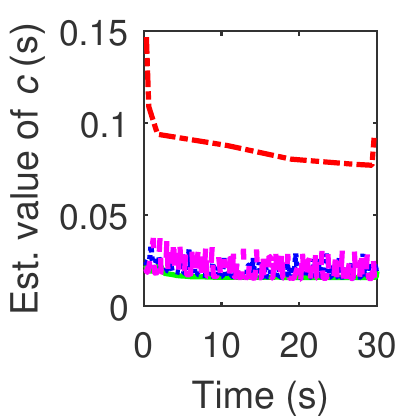}
    \end{subfigure}%
    ~\hspace{-0.1in}
    \begin{subfigure}[b]{0.12\textwidth}
        \centering
        \includegraphics[width=1\linewidth]{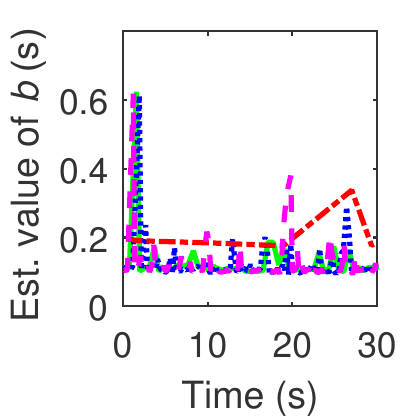}
    \end{subfigure}%
    \vspace{-0.1in}

    \begin{subfigure}[b]{0.12\textwidth}
        \centering
        \includegraphics[width=1\linewidth]{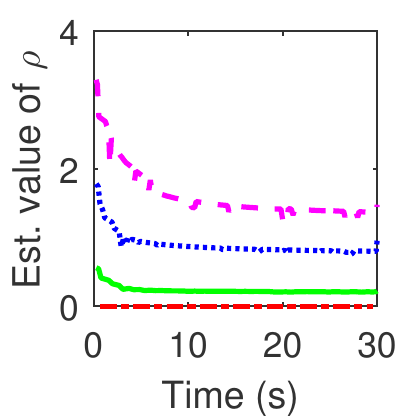}
    \end{subfigure}%
    ~\hspace{-0.1in}
    \begin{subfigure}[b]{0.12\textwidth}
        \centering
        \includegraphics[width=1\linewidth]{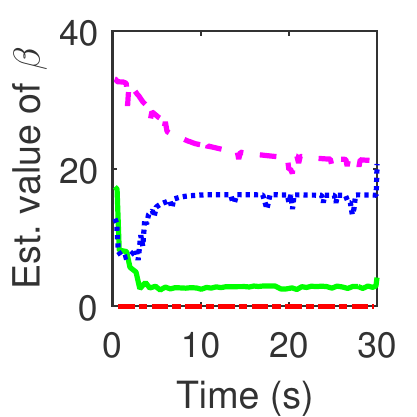}
    \end{subfigure}%
    ~\hspace{-0.1in}
    \begin{subfigure}[b]{0.12\textwidth}
        \centering
        \includegraphics[width=1\linewidth]{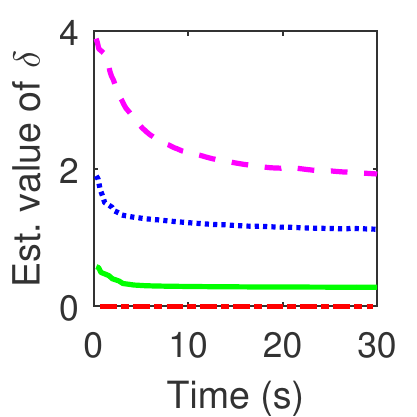}
    \end{subfigure}%
    ~\hspace{-0.1in}
    \begin{subfigure}[b]{0.12\textwidth}
        \centering
        \includegraphics[width=1\linewidth]{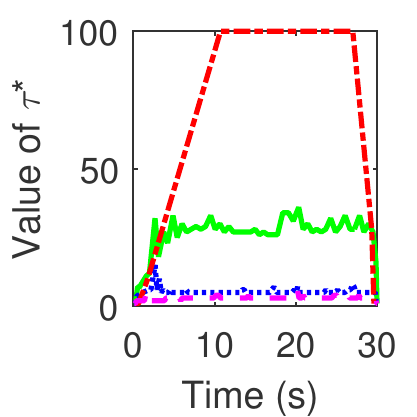}
    \end{subfigure}%

\caption{Instantaneous results of SVM (DGD) with the proposed algorithm.}
\label{fig:InstantSVMDGDExperimentResults}
\end{figure}

\subsubsection{Sensitivity of $\varphi$}
\label{subsubsec:sensitivityOfPhiResult}

\begin{figure}[t]
    \centering
    \begin{subfigure}[b]{0.18\textwidth}
        \centering
        \includegraphics[width=1\linewidth]{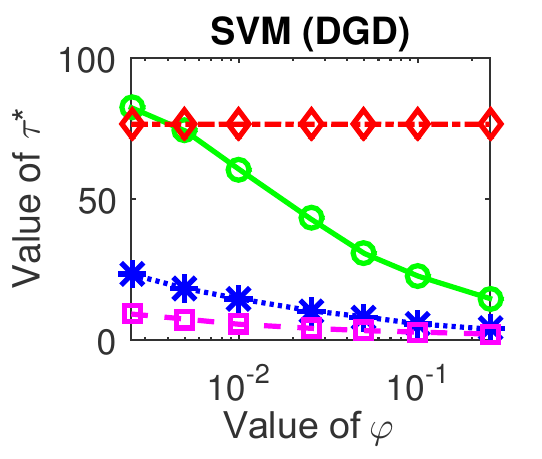}
    \end{subfigure}%
    ~
    \begin{subfigure}[b]{0.18\textwidth}
        \centering
        \includegraphics[width=1\linewidth]{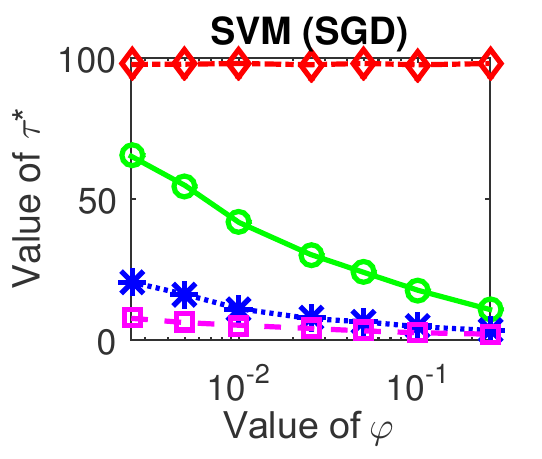}
    \end{subfigure}%
    ~
    \begin{subfigure}[b]{0.08\textwidth}
        \centering
        \includegraphics[width=1\linewidth]{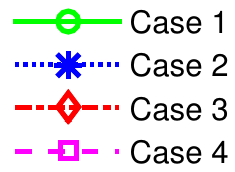}
        \vspace{0.4in}
    \end{subfigure}%
    \caption{Impact of $\varphi$ on the average value of $\tau^*$ in the proposed algorithm.}
    \label{fig:sensitivityControlParameterExperimentationResult}
\end{figure}

The sensitivity of the control parameter $\varphi$ evaluated on the prototype system is shown in Fig.~\ref{fig:sensitivityControlParameterExperimentationResult}. We see that the relationship among $\tau^*$ in different cases is mostly maintained with different values of $\varphi$. The value of $\tau^*$ decreases approximately linearly with $\log\varphi$, which is consistent with the fact that there is an exponential term w.r.t. $\tau$ in $h(\tau)$ (and thus $G(\tau)$). 
For Case~3, $\tau^*$ remains the same with different $\varphi$, because $h(\tau)=0$ in this case by definition (see the remark in Section~\ref{subsec:ControlAlgParamEst}) and the value of $\varphi$ does not affect $\tau^*$, as $G(\tau) \propto \frac{1}{\varphi}$ independently of $\tau$ in this case according to (\ref{eq:GTauDef}).
We also see that small changes of $\varphi$ does not change $\tau^*$ much, indicating that one can take big steps when tuning $\varphi$ in practice and the tuning is not difficult.

\subsubsection{Comparison to Asynchronous Distributed Gradient Descent}

Asynchronous gradient descent~\cite{Chen2016} is an alternative to the typically used synchronous gradient descent in federated learning. With asynchronous gradient descent, the edge nodes operate in an asynchronous manner. Each edge node pulls the most up-to-date model parameter from the aggregator, computes the gradient on its local dataset, then sends the gradient back to the aggregator. The aggregator performs gradient descent according to the step size $\eta$ weighted by the dataset sizes of each node, similar to the combination of (\ref{eq:localUpdate}) and (\ref{eq:globalAverage}). The process repeats until the training finishes. 
Asynchronous gradient descent is able to fully utilize the available computational resource at each node by running more gradient descent steps at more powerful (faster) nodes. However, the asynchronism may hurt the overall performance.

It was shown in~\cite{Chen2016} that synchronous gradient descent has benefits over asynchronous gradient descent in a datacenter setting. Here, we study their differences in the edge computing setting with heterogeneous resources (laptops and Raspberry Pis in our experiment) and different data distributions (Cases 1--4). The results for DGD and SGD with SVM are shown in Figs.~\ref{fig:asyncDGD} and \ref{fig:asyncSGD}, respectively. We see that the performance of asynchronous gradient descent is much worse than synchronous gradient descent for non-uniform data distribution in Cases 2 and 4, with slower convergence, sudden changes (indicating instability of the training process), and convergence to higher loss and lower accuracy values. This is because the model tends overfit the datasets on the faster nodes, as many more steps of gradient descent are performed on these nodes compared to the slower nodes. With uniform data distribution (Cases 1 and 3), asynchronous gradient descent performs similar as or slightly better than synchronous gradient descent, because when the datasets at different nodes are similar (Case~1) or equal (Case 3), there is not much harm caused by overfitting the data on the faster nodes. 

Considering the overall performance in all Cases 1--4, we can conclude that it is still better to perform federated learning with synchronous gradient descent as we do throughout this paper. However, how to make more efficient use of heterogeneous resources is something worth investigating in the future.

\begin{figure}[t]
    \centering
    \begin{subfigure}[b]{0.11\textwidth}
        \centering
        \includegraphics[width=1\textwidth]{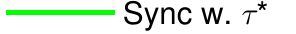}
    \end{subfigure}%
    ~
    \begin{subfigure}[b]{0.11\textwidth}
        \centering
        \includegraphics[width=1\textwidth]{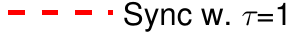}
    \end{subfigure}%
    ~\,\,\,\,
    \begin{subfigure}[b]{0.11\textwidth}
        \centering
        \includegraphics[width=1\textwidth]{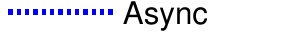}
    \end{subfigure}%
        
    \begin{subfigure}[b]{0.015\textwidth}
        \centering
        \includegraphics[width=1\textwidth]{LossLabel.pdf}
    \end{subfigure}%
    ~
    \begin{subfigure}[b]{0.11\textwidth}
        \centering
        \includegraphics[width=1\textwidth]{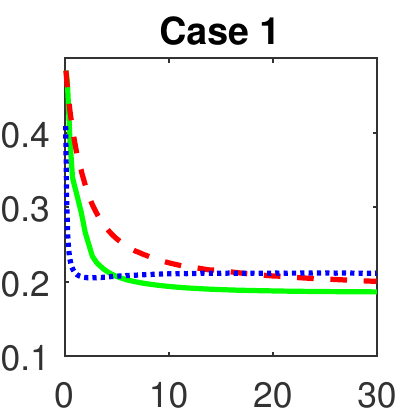}
    \end{subfigure}%
    ~
    \begin{subfigure}[b]{0.11\textwidth}
        \centering
        \includegraphics[width=1\linewidth]{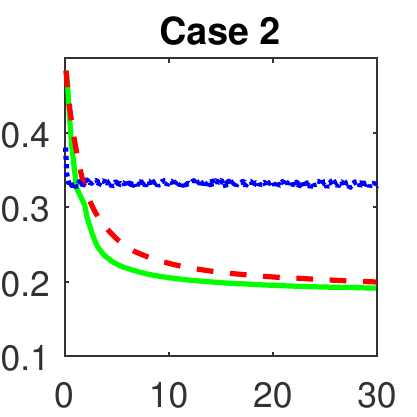}
    \end{subfigure}%
    ~
    \begin{subfigure}[b]{0.11\textwidth}
        \centering
        \includegraphics[width=1\linewidth]{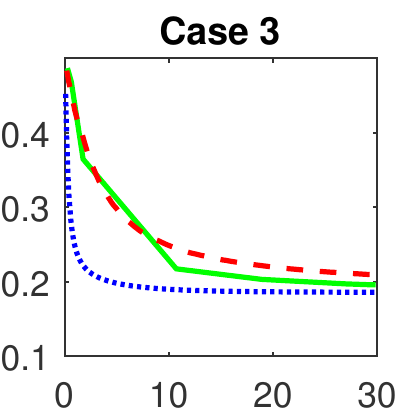}
    \end{subfigure}%
    ~
    \begin{subfigure}[b]{0.11\textwidth}
        \centering
        \includegraphics[width=1\linewidth]{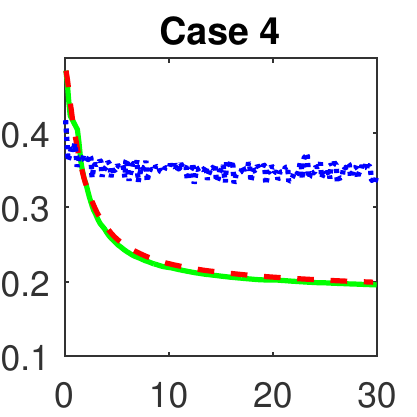}
    \end{subfigure}%

    \begin{subfigure}[b]{0.015\textwidth}
        \centering
        \includegraphics[width=1\textwidth]{AccuracyLabel.pdf}
    \end{subfigure}%
    ~
    \begin{subfigure}[b]{0.11\textwidth}
        \centering
        \includegraphics[width=1\textwidth]{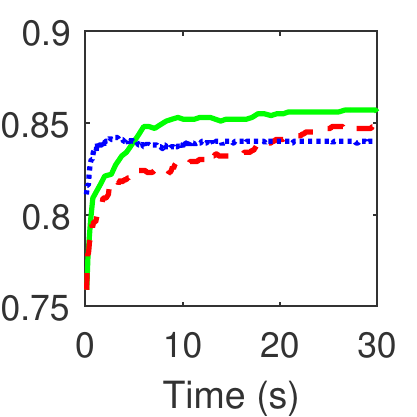}
    \end{subfigure}%
    ~
    \begin{subfigure}[b]{0.11\textwidth}
        \centering
        \includegraphics[width=1\linewidth]{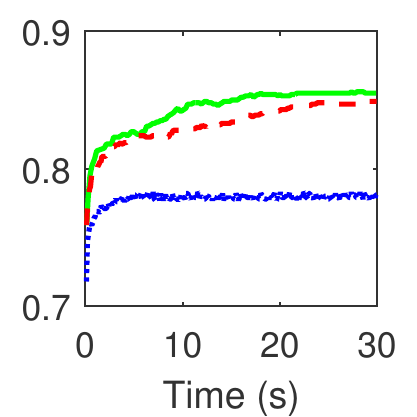}
    \end{subfigure}%
    ~
    \begin{subfigure}[b]{0.11\textwidth}
        \centering
        \includegraphics[width=1\linewidth]{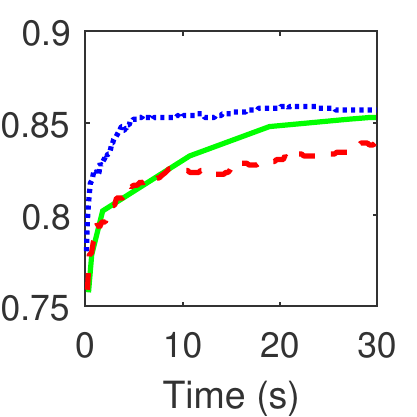}
    \end{subfigure}%
    ~
    \begin{subfigure}[b]{0.11\textwidth}
        \centering
        \includegraphics[width=1\linewidth]{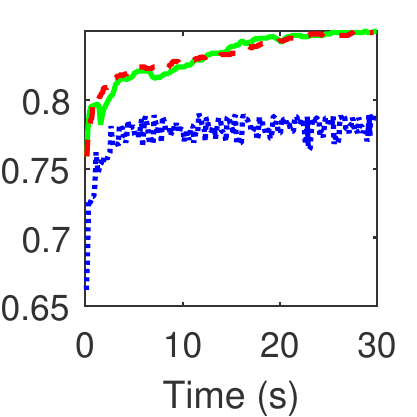}
    \end{subfigure}%
    \caption{Synchronous vs. asynchronous distributed DGD with SVM.}
    \label{fig:asyncDGD}
\end{figure}

\begin{figure}[t]
    \centering
    \begin{subfigure}[b]{0.11\textwidth}
        \centering
        \includegraphics[width=1\textwidth]{Async_Legend1.pdf}
    \end{subfigure}%
    ~
    \begin{subfigure}[b]{0.11\textwidth}
        \centering
        \includegraphics[width=1\textwidth]{Async_Legend2.pdf}
    \end{subfigure}%
    ~\,\,\,\,
    \begin{subfigure}[b]{0.11\textwidth}
        \centering
        \includegraphics[width=1\textwidth]{Async_Legend3.pdf}
    \end{subfigure}%
        
    \begin{subfigure}[b]{0.015\textwidth}
        \centering
        \includegraphics[width=1\textwidth]{LossLabel.pdf}
    \end{subfigure}%
    ~
    \begin{subfigure}[b]{0.11\textwidth}
        \centering
        \includegraphics[width=1\textwidth]{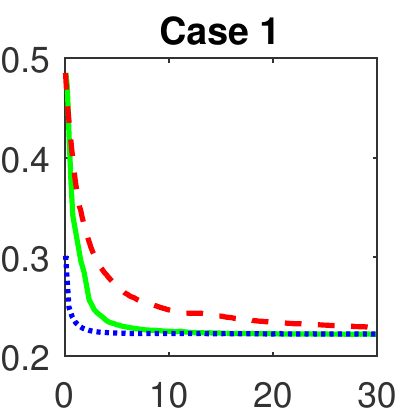}
    \end{subfigure}%
    ~
    \begin{subfigure}[b]{0.11\textwidth}
        \centering
        \includegraphics[width=1\linewidth]{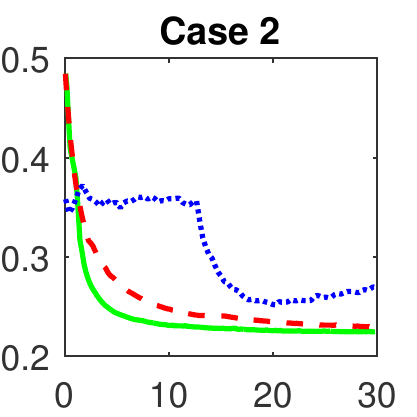}
    \end{subfigure}%
    ~
    \begin{subfigure}[b]{0.11\textwidth}
        \centering
        \includegraphics[width=1\linewidth]{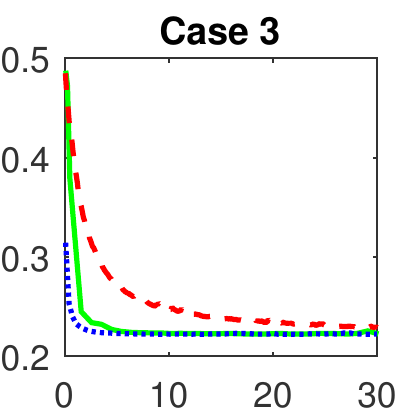}
    \end{subfigure}%
    ~
    \begin{subfigure}[b]{0.11\textwidth}
        \centering
        \includegraphics[width=1\linewidth]{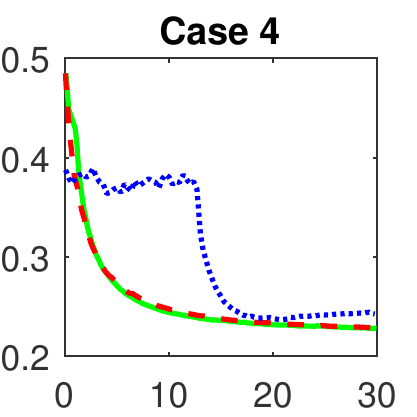}
    \end{subfigure}%

    \begin{subfigure}[b]{0.015\textwidth}
        \centering
        \includegraphics[width=1\textwidth]{AccuracyLabel.pdf}
    \end{subfigure}%
    ~
    \begin{subfigure}[b]{0.11\textwidth}
        \centering
        \includegraphics[width=1\textwidth]{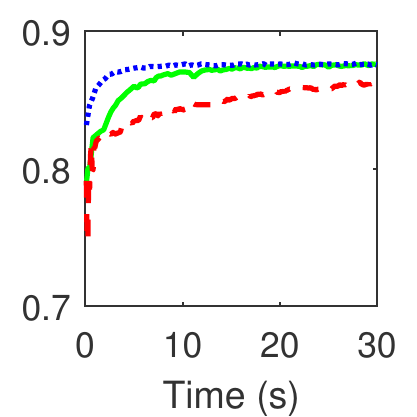}
    \end{subfigure}%
    ~
    \begin{subfigure}[b]{0.11\textwidth}
        \centering
        \includegraphics[width=1\linewidth]{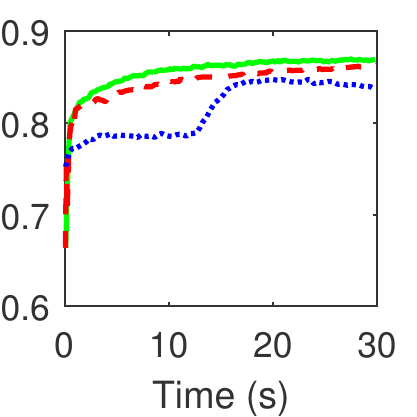}
    \end{subfigure}%
    ~
    \begin{subfigure}[b]{0.11\textwidth}
        \centering
        \includegraphics[width=1\linewidth]{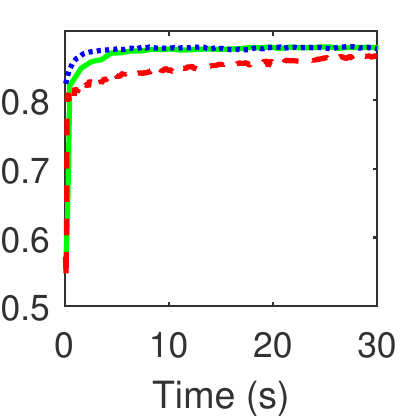}
    \end{subfigure}%
    ~
    \begin{subfigure}[b]{0.11\textwidth}
        \centering
        \includegraphics[width=1\linewidth]{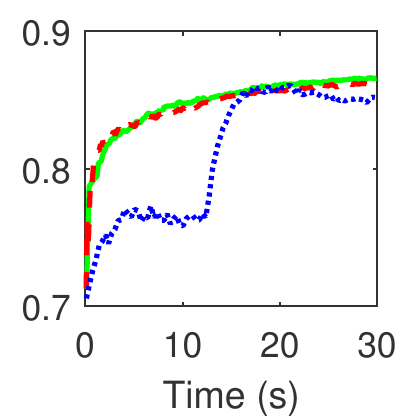}
    \end{subfigure}%
    \caption{Synchronous vs. asynchronous distributed SGD with SVM.}
    \label{fig:asyncSGD}
\end{figure}

\section{Conclusion}

In this paper, we have focused on gradient-descent based federated
learning that include local update and global aggregation steps.
Each step of local update and global aggregation consumes
resources. We have analyzed the convergence bound for federated learning with non-i.i.d. data distributions.  Using this theoretical bound, a control algorithm has been proposed to achieve the desirable trade-off between local update and global aggregation in order to
minimize the loss function under a resource budget constraint.
Extensive experimentation results confirm the effectiveness of our proposed
algorithm.
Future work can investigate how to make the most efficient use of heterogeneous resources for distributed learning, as well as the theoretical convergence analysis of some form of non-convex loss functions representing deep neural networks.

\bibliographystyle{IEEEtran}
\bibliography{ref}

\if\citetechreport1
    \clearpage
\fi

\appendix

\subsection{Distributed vs. Centralized Gradient Descent}
\label{append:optimalityOfDistributedGradDescent}

\begin{proposition}
\label{prop:distGradTauEqualOne}
When $\tau=1$, Algorithm~\ref{alg:distGradDescent} yields the following recurrence relation for $\mathbf{w}{(t)}$:
\begin{equation}
\mathbf{w}{(t)}=\mathbf{w}{(t-1)}-\eta\nabla F\left(\mathbf{w}{(t-1)}\right)
\label{eq:recurrenceCentralizedGradDescent}
\end{equation}
\end{proposition}
\begin{proof}
When $\tau=1$, we have $\widetilde{\mathbf{w}}_{i}{(t)}=\mathbf{w}{(t)}$ for all $t$. Thus,
\begin{align*}
\mathbf{w}{(t)} & = \frac{\sum_{i=1}^{N}D_{i}\mathbf{w}_{i}{(t)}}{D} \\ 
& = \frac{\sum_{i=1}^{N}D_{i}\left(  \widetilde{\mathbf{w}}_{i}{(t-1)}-\eta\nabla F_{i}\left(\widetilde{\mathbf{w}}_{i}{(t-1)}\right)  \right)  }{D}   \\
& = \frac{\sum_{i=1}^{N}D_{i} \mathbf{w}{(t-1)}      }{D}  - \eta \frac{\sum_{i=1}^{N}D_{i}\nabla F_{i}\left(\mathbf{w}{(t-1)}\right)    }{D}   \\
& = \mathbf{w}{(t-1)}-\eta\nabla F\left(\mathbf{w}{(t-1)}\right)
\end{align*}
where the second term in the last equality is because 
\begin{align*}
\frac{\sum_{i=1}^{N}D_{i}\nabla F_{i}\left(\mathbf{w}\right)}{D} = \nabla \left( \frac{\sum_{i=1}^{N}D_{i} F_{i}\left(\mathbf{w}\right)}{D} \right)
=\nabla F\left(\mathbf{w}\right)
\end{align*}
due to the linearity of the gradient operator.
\end{proof}

We note that (\ref{eq:recurrenceCentralizedGradDescent}) is the recurrence relation for centralized gradient decent on the global loss $F(\mathbf{w})$. 
Therefore, the distributed gradient descent algorithm presented in Algorithm~\ref{alg:distGradDescent} is logically equivalent to centralized gradient descent for $\tau=1$.

\subsection{Proof of Theorem~\ref{theorem:wBound}}
\label{append:proofWBound}

To prove Theorem~\ref{theorem:wBound}, we first introduce the following lemma.

\begin{lemma}
\label{lemma:wiBound}
For any interval $[k]$, and $t \in [(k-1)\tau, k\tau)$, we have
$$\left\Vert \widetilde{\mathbf{w}}_{i}{(t)}-\mathbf{v}_{[k]}{(t)}\right\Vert \leq g_i(t- (k-1)\tau)$$
where we define the function $g_i(x)$ as
\begin{align*}
g_{i}(x) \triangleq \frac{\delta_{i}}{\beta}\left((\eta\beta+1)^{x}-1 \right)
\end{align*}
\end{lemma}

\begin{proof}

We show by induction that $\left\Vert \mathbf{w}_{i}{(t)}-\mathbf{v}_{[k]}{(t)}\right\Vert \leq g(t - (k-1)\tau)$ for all $t \in ((k-1)\tau, k\tau]$.

When $t=(k-1)\tau$, we know that $\widetilde{\mathbf{w}}_{i}{(t)}=\mathbf{v}_{[k]}{(t)}$  by the definition of $\mathbf{v}_{[k]}{(t)}$, and we have $\left\Vert \widetilde{\mathbf{w}}_{i}{(t)}-\mathbf{v}_{[k]}{(t)}\right\Vert = g_i(0)$.

We note that  $\widetilde{\mathbf{w}}_{i}{(t)}=\mathbf{w}_{i}{(t)}$ for $t \in ((k-1)\tau, k\tau)$ because there is no global aggregation within this interval. Combining this with  (\ref{eq:localUpdate}), for $t \in ((k-1)\tau, k\tau)$, we have
\begin{align}
\widetilde{\mathbf{w}}_{i}{(t)}=\widetilde{\mathbf{w}}_{i}{(t-1)}-\eta\nabla F_{i}(\widetilde{\mathbf{w}}_{i}{(t-1)})
\label{eq:proofFirstBoundEqWUpdtModified}
\end{align}

For the induction, we assume that
\begin{equation}
\left\Vert \widetilde{\mathbf{w}}_{i}{(t-1)} - \mathbf{v}_{[k]}{(t-1)}\right\Vert \allowbreak \leq g_{i} (t-1 - (k-1)\tau)
\label{eq:proofFirstBoundEqIndAssumption}
\end{equation} 
holds for some $t \in ((k-1)\tau, k\tau) $. We now show that $\left\Vert {\widetilde{\mathbf{w}}}_{i}{(t)}-\mathbf{v}_{[k]}{(t)}\right\Vert \leq g_{i}{(t- (k-1)\tau)}$ holds for $t$. We have
\begin{align*}
& \left\Vert \widetilde{\mathbf{w}}_{i}{(t)}-\mathbf{v}_{[k]}{(t)}\right\Vert  \\
& =  \big\Vert \left(\widetilde{\mathbf{w}}_{i}{(t-1)}-\eta\nabla F_{i}(\widetilde{\mathbf{w}}_{i}{(t-1)})\right) \\ 
& \quad -\left(\mathbf{v}_{[k]}{(t-1)}-\eta\nabla F(\mathbf{v}_{[k]}{(t-1)})\right) \big\Vert \tag{from (\ref{eq:updateV}), (\ref{eq:proofFirstBoundEqWUpdtModified})}\\
& =  \big\Vert \left(\widetilde{\mathbf{w}}_{i}{(t-1)}-\mathbf{v}_{[k]}{(t-1)}\right)-\eta\big[\nabla F_{i}(\widetilde{\mathbf{w}}_{i}{(t-1)}) \\
& \quad - \! \nabla F_{i}(\mathbf{v}_{[k]}{(t-1)}) \! + \! \nabla F_{i}(\mathbf{v}_{[k]}{(t-1)}) \! - \! \nabla F(\mathbf{v}_{[k]}{(t-1)})\big]\big\Vert \tag{adding a zero term and rearranging}\\
& \leq  \left\Vert \widetilde{\mathbf{w}}_{i}{(t-1)}-\mathbf{v}_{[k]}{(t-1)}\right\Vert \\
& \quad +\eta\left\Vert \nabla F_{i}(\widetilde{\mathbf{w}}_{i}{(t-1)})-\nabla F_{i}(\mathbf{v}_{[k]}{(t-1)})\right\Vert \\
& \quad +\eta\left\Vert \nabla F_{i}(\mathbf{v}_{[k]}{(t-1)})-\nabla F(\mathbf{v}_{[k]}{(t-1)})\right\Vert \tag{from triangle inequality} \\
& \leq  (\eta\beta+1)\left\Vert \widetilde{\mathbf{w}}_{i}{(t-1)}-\mathbf{v}_{[k]}{(t-1)}\right\Vert +\eta\delta_{i} \tag{from the $\beta$-smoothness of $F_i(\cdot )$ and (\ref{eq:deltaiDef})} \\
& \leq  (\eta\beta+1)g_{i}{(t-1- (k-1)\tau)}+\eta\delta_{i}  \tag{from the induction assumption in (\ref{eq:proofFirstBoundEqIndAssumption})} \\
& =  (\eta\beta+1)\left(\frac{\delta_{i}}{\beta}\left((\eta\beta+1)^{t-1- (k-1)\tau}-1\right)\right)+\eta\delta_{i}\\
& = \frac{\delta_{i}}{\beta} (\eta\beta+1)^{t- (k-1)\tau}-\frac{\delta_{i}}{\beta} (\eta\beta+1) +\eta\delta_i\\
& = \frac{\delta_{i}}{\beta} (\eta\beta+1)^{t- (k-1)\tau}-\frac{\delta_i}{\beta}\\
& = \frac{\delta_{i}}{\beta}((\eta\beta+1)^{t- (k-1)\tau}-1)\\
& = g_i(t- (k-1)\tau)
\end{align*}

Using the above induction, we have shown that $\left\Vert \widetilde{\mathbf{w}}_{i}{(t)}-\mathbf{v}_{[k]}{(t)}\right\Vert \leq g_{i}(t- (k-1)\tau)$ for all $t \in [(k-1)\tau, k\tau)$.
\end{proof}

We are now ready to prove Theorem~\ref{theorem:wBound}.

\begin{proof}[Proof of Theorem~\ref{theorem:wBound}]

From (\ref{eq:localUpdate}) and (\ref{eq:globalAverage}), we have
\begin{equation}
\mathbf{w}{(t)}=\mathbf{w}{(t-1)}-\eta\frac{\sum_{i}D_{i}\nabla F_{i}(\widetilde{\mathbf{w}}_{i}{(t-1)})}{D}
\label{eq:proogBound2wUpdate}
\end{equation}

Then, for $t\in ((k-1)\tau, k\tau]$, we have
\begin{align*}
& \left\Vert \mathbf{w}{(t)}-\mathbf{v}_{[k]}{(t)}\right\Vert \\
& = \bigg \Vert \mathbf{w}{(t-1)}-\eta\frac{\sum_{i}D_{i}\nabla F_{i}(\widetilde{\mathbf{w}}_{i}{(t-1)})}{D}-\mathbf{v}_{[k]}{(t-1)}  \\
& \quad +\eta\nabla F(\mathbf{v}_{[k]}{(t-1)})\bigg\Vert \tag{from (\ref{eq:proogBound2wUpdate}) and (\ref{eq:updateV})} \\
& = \bigg\Vert \mathbf{w}{(t-1)}-\mathbf{v}_{[k]}{(t-1)} \\
& \quad -\eta\left(\frac{\sum_{i}D_{i}\nabla F_{i}(\widetilde{\mathbf{w}}_{i}{(t-1)})}{D}  -\nabla F(\mathbf{v}_{[k]}{(t-1)})\right)\bigg\Vert \\
& = \bigg\Vert \mathbf{w}{(t-1)}-\mathbf{v}_{[k]}{(t-1)} \\
& \quad -\eta\left(\frac{\sum_{i}D_{i}(\nabla F_{i}(\widetilde{\mathbf{w}}_{i}{(t-1)})-\nabla F_{i}(\mathbf{v}_{[k]}{(t-1)}))}{D}\right)\bigg\Vert \\
& \leq \left\Vert \mathbf{w}{(t-1)}-\mathbf{v}_{[k]}{(t-1)}\right\Vert  \\
& \quad +\eta\left(\frac{\sum_{i}D_{i}\left\Vert \nabla F_{i}(\widetilde{\mathbf{w}}_{i}{(t-1)})-\nabla F_{i}(\mathbf{v}_{[k]}{(t-1)})\right\Vert }{D}\right) \tag{from triangle inequality}\\
& \leq \left\Vert \mathbf{w}{(t-1)}-\mathbf{v}_{[k]}{(t-1)}\right\Vert  \\
& \quad +\eta\beta\left(\frac{\sum_{i}D_{i}\left\Vert \widetilde{\mathbf{w}}_{i}{(t-1)}-\mathbf{v}_{[k]}{(t-1)}\right\Vert }{D}\right)\tag{because $F(\cdot ) $ is $\beta$-smooth}\\
& \leq \left\Vert \mathbf{w}{(t \! - \! 1)} \! - \! \mathbf{v}_{[k]}{(t \! - \! 1)}\right\Vert  \! + \! \eta\beta\left(\frac{\sum_{i}D_{i} g_i(t \! - \! 1 \! - (k \! - \! 1)\tau)}{D}\right)\tag{from Lemma \ref{lemma:wiBound}}\\
& = \left\Vert \mathbf{w}{(t-1)}-\mathbf{v}_{[k]}{(t-1)}\right\Vert +\eta\delta\left((\eta\beta+1)^{t-1- (k-1)\tau}-1\right)
\end{align*}
where the last equality is because for any $x$,
\begin{align*}
\eta\beta\left(\frac{\sum_{i}D_{i} g_i(x)}{D}\right) 
 & = \eta\beta\left(\frac{\sum_{i}D_{i} \frac{\delta_{i}}{\beta} \left( (\eta\beta+1)^{x}-1 \right)  }{D}\right) \\
& = \eta\left(\frac{\sum_{i}D_{i} \delta_{i}}{D}\right) \left((\eta\beta+1)^{x}-1 \right)\\
& = \eta\delta\left((\eta\beta+1)^{x}-1\right)
\end{align*}
Equivalently,
\begin{align}
& \left\Vert \mathbf{w}{(t)}-\mathbf{v}_{[k]}{(t)}\right\Vert - \left\Vert \mathbf{w}{(t-1)}-\mathbf{v}_{[k]}{(t-1)}\right\Vert \nonumber \\
& \leq \eta\delta\left((\eta\beta+1)^{t-1- (k-1)\tau}-1\right)
\label{eq:wvProof}
\end{align}
When $t=(k-1)\tau$, we have $\mathbf{w}{(t)} = \mathbf{v}_{[k]}{(t)}$ according to the definition, thus $\left\Vert \mathbf{w}{(t)}-\mathbf{v}_{[k]}{(t)}\right\Vert=0$.
For $t\in ((k-1)\tau, k\tau]$, by summing up (\ref{eq:wvProof}) over different values of $t$, we have
\begin{align*}
& \left\Vert \mathbf{w}{(t)}-\mathbf{v}_{[k]}{(t)}\right\Vert \\
& = \sum_{y=(k-1)\tau+1}^{t} \left\Vert \mathbf{w}{(y)}-\mathbf{v}_{[k]}{(y)}\right\Vert - \left\Vert \mathbf{w}{(y-1)}-\mathbf{v}_{[k]}{(y-1)}\right\Vert\\
& \leq  \eta\delta  \sum_{y=(k-1)\tau+1}^{t}  \left( (\eta\beta+1)^{y-1-(k-1)\tau}- 1 \right)\\
& =  \eta\delta  \sum_{z=1}^{t-(k-1)\tau}  \left( (\eta\beta+1)^{z-1}- 1 \right)\\
& =  \eta\delta  \sum_{z=1}^{t-(k-1)\tau}   (\eta\beta+1)^{z-1}- \eta\delta(t-(k-1)\tau) \\
& = \eta\delta\frac{(1-(\eta\beta+1)^{t-(k-1)\tau})}{-\eta\beta}-\eta\delta (t-(k-1)\tau)\\
& = \eta\delta\frac{(\eta\beta+1)^{t-(k-1)\tau}-1}{\eta\beta}-\eta\delta(t-(k-1)\tau)\\
& = \frac{\delta}{\beta}\left((\eta\beta+1)^{t-(k-1)\tau}-1\right)-\eta\delta (t-(k-1)\tau)\\
& = h(t-(k-1)\tau)
\end{align*}
\end{proof}

\subsection{Proof of Lemma~\ref{lemma:convergenceUpperBound}}
\label{append:proofConvergenceUpperBound}

To prove Lemma~\ref{lemma:convergenceUpperBound}, we first introduce some additional definitions and lemmas.

\begin{definition} For an interval $[k]$, we define $\theta_{[k]}(t)=F(\mathbf{v}_{[k]}(t))-F(\mathbf{w}^*)$, for a fixed $k$, $t$ is defined between $(k-1)\tau \leq t \leq k\tau$.
\end{definition}

According to the convergence lower bound of gradient descent given in \cite[Theorem 3.14]{convex}, we always have 
\begin{equation}
\theta_{[k]}(t) > 0
\label{eq:thetaLowerBound}
\end{equation}
for any finite $t$ and $k$.

\begin{lemma}
When $\eta \leq \frac{1}{\beta}$, for any $k$, and $t \in [(k-1)\tau, k\tau]$, we have that $ \left\Vert \mathbf{v}_{[k]}(t)- \mathbf{w}^*\right\Vert $ does not increase with $t$, where $\mathbf{w}^*$ is the optimal parameter defined in (\ref{eq:learningProblem}). 
\label{lemma:vwDecrease}
\end{lemma}

\begin{proof}
\begin{align*}
& \left\Vert \mathbf{v}_{[k]}(t+1)- \mathbf{w}^*\right\Vert^2 \\
& = \left\Vert \mathbf{v}_{[k]}(t)-\eta \nabla F(\mathbf{v}_{[k]}(t))-\mathbf{w}^*\right\Vert^2\\
&=\left\Vert \mathbf{v}_{[k]}(t)- \mathbf{w}^*\right\Vert^2 -2\eta\nabla F(\mathbf{v}_{[k]}(t))^\mathrm{T}(\mathbf{v}_{[k]}(t)-\mathbf{w}^*) \\
& \quad +\eta^2 \left\Vert \nabla F(\mathbf{v}_{[k]}(t)) \right\Vert^2
\end{align*}

Because $F(\cdot )$ is $\beta$-smooth, from (\ref{eq:thetaLowerBound}) and \cite[Lemma 3.5]{convex}, we have
\begin{align*}
0 & <  \theta_{[k]}(t) \leq \nabla F(\mathbf{v}_{[k]}(t))^\mathrm{T}(\mathbf{v}_{[k]}(t) \! - \! \mathbf{w}^*) \! - \! \frac{\left\Vert \nabla F(\mathbf{v}_{[k]}(t))\right\Vert^2}{2\beta}
\end{align*}
Thus,
$$ -\nabla F(\mathbf{v}_{[k]}(t))^\mathrm{T}(\mathbf{v}_{[k]}(t)-\mathbf{w}^*) < - \frac{\left\Vert \nabla F(\mathbf{v}_{[k]}(t))\right\Vert^2}{2\beta} $$

Therefore,
\begin{align*}
& \left\Vert \mathbf{v}_{[k]}(t+1)- \mathbf{w}^*\right\Vert^2 \\ 
& = \left\Vert \mathbf{v}_{[k]}{(t)} - \eta\nabla F(\mathbf{v}_{[k]}{(t)}) - \mathbf{w}^*\right\Vert^2  \tag{from (\ref{eq:updateV})}\\
& = \left\Vert \mathbf{v}_{[k]}(t)- \mathbf{w}^*\right\Vert^2 -2\eta\nabla F(\mathbf{v}_{[k]}(t))^\mathrm{T}(\mathbf{v}_{[k]}(t)-\mathbf{w}^*) \tag{expanding the squared norm}\\
& \quad +\eta^2 \left\Vert \nabla F(\mathbf{v}_{[k]}(t)) \right\Vert^2 \\
& < \! \left\Vert \mathbf{v}_{[k]}(t) \!-\! \mathbf{w}^*\right\Vert^2  \!-\! \eta \frac{\left\Vert \nabla F(\mathbf{v}_{[k]}(t))\right\Vert^2}{\beta} + \eta^2 \left\Vert \nabla F(\mathbf{v}_{[k]}(t))\right\Vert^2\\
& = \left\Vert \mathbf{v}_{[k]}(t)- \mathbf{w}^*\right\Vert^2 -\eta \left(\frac{1}{\beta} - \eta \right)\left\Vert \nabla F(\mathbf{v}_{[k]}(t)) \right\Vert^2 
\end{align*}

When $\eta \leq \frac{1}{\beta}$, we obtain
\begin{align*}
\left\Vert \mathbf{v}_{[k]}(t+1)- \mathbf{w}^*\right\Vert^2 & \leq \left\Vert \mathbf{v}_{[k]}(t)- \mathbf{w}^*\right\Vert^2 
\end{align*}
\end{proof}

\begin{lemma}
\label{lemma:FvDecrease}
For any $k$, when $\eta \leq \frac{1}{\beta}$ and $t \in [(k-1)\tau, k\tau)$, we have
\begin{align}
F(\mathbf{v}_{[k]}(t+1))-F(\mathbf{v}_{[k]}(t))  \leq -\eta \left(1-\frac{\beta\eta}{2}\right)\left\Vert\nabla F(\mathbf{v}_{[k]}(t))\right\Vert^2 \label{eq:FvDecrease}
\end{align}
\end{lemma}
\begin{proof}
Because $F(\cdot )$ $\beta$-smooth, from \cite[Lemma 3.4]{convex}, we have
\begin{align*}
F(\mathbf{x}) \leq F(\mathbf{y}) +\nabla F(\mathbf{y})^\mathrm{T}(\mathbf{x}-\mathbf{y})+\frac{\beta}{2}{\left\Vert \mathbf{x}-\mathbf{y}\right\Vert}^2
\end{align*}
for arbitrary $\mathbf{x}$ and $\mathbf{y}$.
Thus,
\begin{align}
& F(\mathbf{v}_{[k]}(t+1))-F(\mathbf{v}_{[k]}(t)) \nonumber\\
&\leq \nabla F(\mathbf{v}_{[k]}(t))^\mathrm{T}(\mathbf{v}_{[k]}(t+1)-\mathbf{v}_{[k]}(t)) \nonumber \\
& \quad +\frac{\beta}{2}{\left\Vert\mathbf{v}_{[k]}(t+1)-\mathbf{v}_{[k]}(t) \right\Vert}^2 \nonumber\\
& \leq -\eta \nabla F(\mathbf{v}_{[k]}(t))^\mathrm{T}\nabla F(\mathbf{v}_{[k]}(t))+\frac{\beta\eta^2}{2}{\left\Vert\nabla F(\mathbf{v}_{[k]}(t)) \right\Vert}^2 \tag{from (\ref{eq:updateV})}\nonumber\\
& \leq -\eta \left(1-\frac{\beta\eta}{2}\right)\left\Vert\nabla F(\mathbf{v}_{[k]}(t))\right\Vert^2  \nonumber
\end{align}
\end{proof}

\begin{lemma}
\label{lemma:thetaBound}
For any $k$, when $\eta \leq \frac{1}{\beta}$ and $t \in [(k-1)\tau, k\tau)$, we have
\begin{align}
\frac{1}{\theta_{[k]}(t+1)}-\frac{1}{\theta_{[k]}(t)} \geq \omega \eta \left(1-\frac{\beta\eta}{2}\right) 
\end{align}
where $\omega=\min_k \frac{1}{{\left\Vert \mathbf{v}_{[k]}((k-1)\tau)-\mathbf{w}^* \right\Vert^2 }}$
\end{lemma}

\begin{proof}

By definition,
$\theta_{[k]}(t)=F(\mathbf{v}_{[k]}(t))-F(\mathbf{w}^*)$ and $\theta_{[k]}(t+1)=F(\mathbf{v}_{[k]}(t+1))-F(\mathbf{w}^*)$ 
Substituting this into (\ref{eq:FvDecrease}) in Lemma~\ref{lemma:FvDecrease} yields
\begin{align*}
\theta_{[k]}(t+1)-\theta_{[k]}(t) &\leq  -\eta \left(1-\frac{\beta\eta}{2}\right)\left\Vert\nabla F(\mathbf{v}_{[k]}(t))\right\Vert^2
\end{align*}
Equivalently,
\begin{align}
\theta_{[k]}(t+1) &\leq \theta_{[k]}(t)  -\eta \left(1-\frac{\beta\eta}{2}\right)\left\Vert\nabla F(\mathbf{v}_{[k]}(t))\right\Vert^2\label{eq:proofThetaDiffBound}
\end{align}

The convexity condition gives 
\begin{align*}   
\theta_{[k]}(t) & = F(\mathbf{v}_{[k]}(t))-F(\mathbf{w}^*)  \leq  \nabla F(\mathbf{v}_{[k]}(t))^\mathrm{T}(\mathbf{v}_{[k]}(t)-\mathbf{w}^*) \\
& \leq \left\Vert \nabla F(\mathbf{v}_{[k]}(t))\right\Vert \left\Vert \mathbf{v}_{[k]}(t)-\mathbf{w}^* \right\Vert
\end{align*}
where the last inequality is from the Cauchy-Schwarz inequality.
Hence,
\begin{align}   
\frac{\theta_{[k]}(t)}{\left\Vert \mathbf{v}_{[k]}(t)-\mathbf{w}^* \right\Vert} \leq \left\Vert \nabla F(\mathbf{v}_{[k]}(t))\right\Vert \label{eq5}
\end{align} 

Substituting (\ref{eq5}) into (\ref{eq:proofThetaDiffBound}), we get
\begin{align*}
\theta_{[k]}(t+1) & \leq \theta_{[k]}(t)-\frac{\eta\left(1-\frac{\beta\eta}{2}\right)\theta_{[k]}(t)^2}{\left\Vert \mathbf{v}_{[k]}(t)-\mathbf{w}^* \right\Vert^2} \\
& \leq \theta_{[k]}(t)-\omega{\eta\left(1-\frac{\beta\eta}{2}\right)\theta_{[k]}(t)^2}
\end{align*}
where the last inequality in the above is explained as follows. From Lemma \ref{lemma:vwDecrease}, we know that for each interval of $[k]$, ${\left\Vert\mathbf{v}_{[k]}(t)-\mathbf{w}^*\right\Vert}$ does not increase with $t$ when $t \in [(k-1)\tau, k\tau]$. Hence, $\left\Vert \mathbf{{v}}_{[k]}((k-1)\tau)\ - \mathbf{w}^*\right\Vert \geq \left\Vert \mathbf{{v}}_{[k]}(t)\ - \mathbf{w}^*\right\Vert $.
Recall that we defined $\omega=\min_k \frac{1}{{\left\Vert \mathbf{v}_{[k]}((k-1)\tau)-\mathbf{w}^*\right\Vert^2 }}$, we have
 $-\omega \geq \frac{-1}{\left\Vert \mathbf{{v}}_{[k]}((k-1)\tau)\ - \mathbf{w}^*\right\Vert^2} \geq \frac{-1}{ \left\Vert \mathbf{{v}}_{[k]}(t)\ - \mathbf{w}^*\right\Vert^2}$ and the inequality follows.

As $\theta_{[k]}(t+1)\theta_{[k]}(t) > 0$ according to (\ref{eq:thetaLowerBound}), dividing both sides by $\theta_{[k]}(t+1)\theta_{[k]}(t)$, we obtain
\begin{align*}
\frac{1}{\theta_{[k]}(t)}\leq \frac{1}{\theta_{[k]}(t+1)}-\frac{\omega{\eta\left(1-\frac{\beta\eta}{2}\right)\theta_{[k]}(t)}}{\theta_{[k]}(t+1)}
\end{align*}

We have $0 < \theta_{[k]}(t+1) \leq \theta_{[k]}(t)$ from (\ref{eq:thetaLowerBound}) and (\ref{eq:proofThetaDiffBound}), thus $\frac{\theta_{[k]}(t) }{\theta_{[k]}(t+1) } \geq 1$. Hence, 
\begin{align*}
\frac{1}{\theta_{[k]}(t+1)}-\frac{1}{\theta_{[k]}(t)}\geq \frac{\omega{\eta\left(1-\frac{\beta\eta}{2}\right)\theta_{[k]}(t)}}{\theta_{[k]}(t+1)} \geq \omega{\eta\left(1-\frac{\beta\eta}{2}\right)}
\end{align*}
\end{proof}

We are now ready to prove Lemma~\ref{lemma:convergenceUpperBound}.

\begin{proof}[Proof of Lemma~\ref{lemma:convergenceUpperBound}]

Using Lemma \ref{lemma:thetaBound} and considering $t\in [(k-1)\tau,k\tau]$, we have
\begin{align*}
 \frac{1}{\theta_{[k]}(k\tau)}-\frac{1}{\theta_{[k]}((k-1)\tau)} 
& = \sum_{z=(k-1)\tau}^{k\tau -1}  \left( \frac{1}{\theta_{[k]}(t+1)}-\frac{1}{\theta_{[k]}(t)} \right) \\
& \geq \tau\omega\eta\left(1-\frac{\beta\eta}{2}\right)
\end{align*}
Summing up the above for all $k=1,2...,K$ yields
\begin{align*}
\sum_{k=1}^{K} \left( \frac{1}{\theta_{[k]}(k\tau)}-\frac{1}{\theta_{[k]}((k-1)\tau)} \right) & \geq \sum_{k=1}^{K}  \tau\omega\eta \left(1-\frac{\beta\eta}{2}\right) \\
& = K\tau\omega\eta \left(1-\frac{\beta\eta}{2}\right)
\end{align*}
Rewriting the left-hand side and noting that $T=K\tau$ yields
\begin{align*}
&\frac{1}{\theta_{\left[K\right]}(T)}-\frac{1}{\theta_{[1]}(0)}-\sum_{k=1}^{K-1}\left( \frac{1}{\theta_{[k+1]}(k\tau)}-\frac{1}{\theta_{[k]}(k\tau)}\right) \\
& \geq T\omega\eta \left(1-\frac{\beta\eta}{2}\right)
\end{align*}
which is equivalent to
\begin{align}
\label{etoile}
& \frac{1}{\theta_{[K]}(T)}-\frac{1}{\theta_{[1]}(0)} \nonumber \\
& \geq T\omega\eta \left(1-\frac{\beta\eta}{2}\right) +\sum_{k=1}^{K-1}\left( \frac{1}{\theta_{[k+1]}(k\tau)}-\frac{1}{\theta_{[k]}(k\tau)}\right) 
\end{align}
Each term in the sum in right-hand side of (\ref{etoile}) can be further expressed as
\begin{align}
\frac{1}{\theta_{[k+1]}(k\tau)}-\frac{1}{\theta_{[k]}(k\tau)} 
& = \frac{\theta_{[k]}(k\tau)-{\theta_{[k+1]}(k\tau)}}{{\theta_{[k]}(k\tau)}{\theta_{[k+1]}(k\tau)}} \nonumber \\
& = \frac{F(\mathbf{v}_{[k]}(k\tau))-F(\mathbf{v}_{[k+1]}(k\tau))}{\theta_{[k]}(k\tau)\theta_{[k+1]}(k\tau)} \nonumber \\
& \geq \frac{-\rho h(\tau)}{\theta_{[k]}(k\tau)\theta_{[k+1]}(k\tau)} \label{eq:convergenceProofBoundWithFuncH}
\end{align}
where the last inequality is obtained using Theorem~\ref{theorem:wBound} and noting that, according to the definition, $\mathbf{v}_{[k+1]}(k\tau) = \mathbf{w}(k\tau)$, thus $F(\mathbf{v}_{[k+1]}(k\tau))=F(\mathbf{w}(k\tau))$. 

It is assumed that $F(\mathbf{v}_{[k]}(k\tau))-F(\mathbf{w}^*) \geq \varepsilon$ for all $k$. According to Lemma~\ref{lemma:FvDecrease}, $F(\mathbf{v}_{[k]}(t)) \geq F(\mathbf{v}_{[k]}(t+1))$ for any $t \in [(k-1)\tau, k\tau)$. Therefore, we have 
$\theta_{[k]}(t) = F(\mathbf{v}_{[k]}(t))-F(\mathbf{w}^*) \geq \varepsilon$
for all $t$ and $k$ for which $\mathbf{v}_{[k]}(t)$ is defined.
Consequently,
\begin{align}
\theta_{[k]}(k\tau)\theta_{[k+1]}(k\tau) & \geq \varepsilon^2 \nonumber \\
\frac{-1}{\theta_{[k]}(k\tau)\theta_{[k+1]}(k\tau)} & \geq - \frac{1}{\varepsilon^2}
\label{eq:convergenceProofEpsilon}
\end{align}

Combining (\ref{eq:convergenceProofEpsilon}) with (\ref{eq:convergenceProofBoundWithFuncH}), the sum in the right-hand side of (\ref{etoile}) can be bounded by
\begin{align*}
\sum_{k=1}^{K-1}\left(\frac{1}{\theta_{[k+1]}(k\tau)}-\frac{1}{\theta_{[k]}(k\tau)}\right) 
& \geq -\sum_{k=1}^{K-1}\frac{\rho h(\tau)}{\varepsilon^2}  \\
& = -\left(K-1\right)\frac{\rho h(\tau)}{\varepsilon^2}
\end{align*}
Substituting the above into (\ref{etoile}), we get
\begin{align}
\frac{1}{\theta_{[K]}(T)}-\frac{1}{\theta_{[1]}(0)} 
& \geq T\omega\eta \left(1-\frac{\beta\eta}{2}\right) - (K-1)\frac{\rho h(\tau)}{\varepsilon^2} 
\label{eq:convergenceProofReciprocal}
\end{align}

It is also assumed that $ F\left(\mathbf{w}(T)\right)-F(\mathbf{w}^*) \geq \varepsilon$. Using the same argument as for obtaining (\ref{eq:convergenceProofEpsilon}), we have
\begin{align}
\frac{-1}{\left( F(\mathbf{w}(T))-F(\mathbf{w}^*) \right) \theta_{[K]}(T)} & \geq - \frac{1}{\varepsilon^2}
\label{eq:convergenceProofEpsilon2}
\end{align}
We then have
\begin{align}
 \frac{1}{F(\!\mathbf{w}(T))\!-\!F(\!\mathbf{w}^*\!)}\! -\! \frac{1}{\theta_{[K]}(T)}\!
& =\! \frac{\theta_{[K]}(T) \!-\! \left(F(\mathbf{w}(T))\!-\!F(\mathbf{w}^*)\right) }{\left(F(\mathbf{w}(T))\!-\!F(\mathbf{w}^*)\right)\theta_{[K]}(T)} \nonumber \\
& = \frac{F(\mathbf{v}_{[K]}(T)) - F(\mathbf{w}(T))}{\left(F(\mathbf{w}(T))-F(\mathbf{w}^*)\right)\theta_{[K]}(T)} \nonumber \\
& \geq \frac{-\rho h(\tau) }{\left(F(\mathbf{w}(T))-F(\mathbf{w}^*)\right)\theta_{[K]}(T)} \nonumber \\
& \geq - \frac{\rho h(\tau) }{\varepsilon^2}
\label{eq:convergenceProofReciprocal2}
\end{align}
where the first inequality is from Theorem~\ref{theorem:wBound} and the second inequality is from (\ref{eq:convergenceProofEpsilon2}).

Summing up (\ref{eq:convergenceProofReciprocal}) and (\ref{eq:convergenceProofReciprocal2}), we have
\begin{align*}
\frac{1}{F(\mathbf{w}(T))-F(\mathbf{w}^*)} - \frac{1}{\theta_{[1]}(0)} 
& \geq T\omega\eta \left(1-\frac{\beta\eta}{2}\right) - K\frac{\rho h(\tau)}{\varepsilon^2} \\
& = T\omega\eta \left(1-\frac{\beta\eta}{2}\right) - T\frac{\rho h(\tau)}{\tau\varepsilon^2} \\
& =  T\left(\omega\eta\left(1 \! - \! \frac{\beta \eta}{2}\right)-\frac{\rho h(\tau)}{\tau\varepsilon^2}\right)
\end{align*}
where the first equality is because $K=\frac{T}{\tau}$.

We note that
\begin{align*}
\frac{1}{F(\mathbf{w}(T))-F(\mathbf{w}^*)} & \geq \frac{1}{F(\mathbf{w}(T))-F(\mathbf{w}^*)} - \frac{1}{\theta_{[1]}(0)} \\
& \geq T\left(\omega\eta\left(1-\frac{\beta \eta}{2}\right)-\frac{\rho h(\tau)}{\tau\varepsilon^2}\right) > 0    
\end{align*}
where the first inequality is because $\theta_{[1]}(0)=F(\mathbf{v}_{[1]}(0))-F(\mathbf{w}^*) > 0$, and the last inequality is due to the assumption that $\omega\eta(1-\frac{\beta \eta}{2})-\frac{\rho h(\tau)}{\tau\varepsilon^2} > 0$.
Taking the reciprocal of the above inequality yields
\begin{align*}
F(\mathbf{w}(T))-F(\mathbf{w}^*)  & \leq \frac{1}{T\left(\omega\eta\left(1-\frac{\beta \eta}{2}\right)-\frac{\rho h(\tau)}{\tau\varepsilon^2}\right)} \\
& = \frac{1}{T\left(\eta\varphi -\frac{\rho h(\tau)}{\tau\varepsilon^2}\right)}
\end{align*}
\end{proof}
\vspace{-0.5in}

\subsection{Proof of Proposition~\ref{prop:TauOptBounded}}
\label{append:proofGTauOptBounded}

We first prove that $\tau_0$  is finite.
According to the definition of $\nu$, we have $\frac{c_\nu}{R'_\nu} \geq \frac{c_m}{R'_m}$ for all $m$, thus $c_\nu R'_m -c_m R'_\nu \geq 0$. For any $m$, we consider the following two cases.
\begin{enumerate}
\item When $c_\nu R'_m -c_m R'_\nu > 0$, it is obvious that $\frac{b_m R'_\nu - b_\nu R'_m}{c_\nu R'_m -c_m R'_\nu}$ is finite. 
\item When $c_\nu R'_m -c_m R'_\nu = 0$, according to the definition of $\nu$, we have $\frac{b_\nu}{R'_\nu} \geq \frac{b_m}{R'_m}$ thus $b_m R'_\nu - b_\nu R'_m \leq 0$. We further consider two cases as follows.
\begin{enumerate}
\item If $b_m R'_\nu - b_\nu R'_m < 0$, we have $\frac{b_m R'_\nu - b_\nu R'_m}{c_\nu R'_m -c_m R'_\nu} = -\infty$.
\item If $b_m R'_\nu - b_\nu R'_m = 0$, because we define $\frac{0}{0} \triangleq 0$, we have $\frac{b_m R'_\nu - b_\nu R'_m}{c_\nu R'_m -c_m R'_\nu} = 0$.
\end{enumerate}
\end{enumerate}

Combining the above, we know that $\max_m  \frac{b_m R'_\nu - b_\nu R'_m}{c_\nu R'_m -c_m R'_\nu}$ is finite. Then, we can easily see that $\tau_0$  is finite.

Now, we prove that $\tau^* \leq \tau_0$. We will first show that for any $\tau > \tau_0$, we have
\begin{align}
\max_m \frac{c_m\tau + b_m}{R'_m \tau} = \frac{c_\nu\tau + b_\nu}{R'_\nu \tau}.
\label{eq:proofGTauOptBounded:max_nu_equiv}
\end{align}
To see this, we note that when $\tau > \tau_0$, we have
\begin{align}
\tau > \tau_0 \geq \max_m  \frac{b_m R'_\nu - b_\nu R'_m}{c_\nu R'_m -c_m R'_\nu} \geq  \frac{b_m R'_\nu - b_\nu R'_m}{c_\nu R'_m -c_m R'_\nu}
\label{eq:proofGTauOptBounded:tau_nu_ineq}
\end{align}
for any $m$.
As mentioned above, we have $c_\nu R'_m -c_m R'_\nu \geq 0$ according to the definition of $\nu$.
We consider the following two cases for any $m$.
\begin{enumerate}
\item When $c_\nu R'_m -c_m R'_\nu > 0$, we can rearrange (\ref{eq:proofGTauOptBounded:tau_nu_ineq}) and obtain
\begin{align*}
\frac{c_\nu\tau + b_\nu}{R'_\nu \tau} >  \frac{c_m\tau + b_m}{R'_m \tau}.
\end{align*}
\item When $c_\nu R'_m -c_m R'_\nu = 0$ (i.e., $\frac{c_\nu}{R'_\nu} = \frac{c_m}{R'_m}$), we have $\frac{b_\nu}{R'_\nu} \geq \frac{b_m}{R'_m}$ according to the definition of $\nu$. Then, it is obvious that
\begin{align*}
\frac{c_\nu\tau + b_\nu}{R'_\nu \tau} \geq  \frac{c_m\tau + b_m}{R'_m \tau}.
\end{align*}
\end{enumerate}
Combining these two cases, we have proven (\ref{eq:proofGTauOptBounded:max_nu_equiv}).
In the following, we define $c \triangleq c_\nu$ and $b \triangleq b_\nu$ for simplicity.

It follows that for $\tau > \tau_0$, we can rewrite $G(\tau)$ as
\begin{equation}
G(\tau) = H_1(\tau) + \sqrt{H_2(\tau)}
\end{equation}
where
\begin{align*}
H_1(\tau) & \triangleq \frac{c\tau + b}{C_1 \tau} + \rho h(\tau) \\
& = \frac{c}{C_1} +\frac{b}{C_1 \tau} + \frac{\rho\delta(B^\tau - 1)}{\beta} - \rho\eta\delta\tau \\
H_2(\tau) & \triangleq \frac{(c\tau + b)^2}{C_2\tau^2} + \frac{\rho h(\tau)}{\eta\varphi\tau} \\
& = \frac{c^2}{C_2} + \frac{2cb}{C_2 \tau} + \frac{b^2}{C_2 \tau^2} + \frac{\rho\delta(B^\tau - 1)}{\eta\beta\varphi\tau} - \frac{\rho\delta}{\varphi}.
\end{align*}

Next, we consider continuous values of $\tau$ (with a slight abuse of notation) and continue to assume that $\tau > \tau_0$. We show that both $H_1(\tau)$ and $H_2(\tau)$ increase with $\tau$ in this case.

Taking the first order derivative of $H_1(\tau)$, we have
\begin{align*}
\frac{d H_1(\tau)}{d\tau} & = - \frac{b}{C_1 \tau^2} + \frac{\rho\delta B^\tau \log B}{\beta} - \rho\eta\delta  \\
& \geq - \frac{b}{C_1 \tau^2} + \frac{\rho\delta \log B}{\beta} (1 + \eta\beta\tau) - \rho\eta\delta   \tag{Bernoulli's inequality}  \\
& >  - \frac{b}{C_1} + \frac{\rho\delta \log B}{\beta} (1 + \eta\beta\tau) - \rho\eta\delta  \tag{$\tau > 1$, $\frac{b}{C_1} > 0$}  \\
& > 0 \tag{$\tau > \tau_0 > \frac{1}{\rho\delta\eta \log B} \left(\frac{b}{C_1} + \rho\eta\delta\right)   - \frac{1}{\eta\beta}$}
\end{align*}

Taking the first order derivative of $H_2(\tau)$, we have
\begin{align}
& \frac{d H_2(\tau)}{d\tau} \nonumber \\
& = - \frac{2cb}{C_2 \tau^2} - \frac{2b^2}{C_2 \tau^3} + \frac{\rho\delta}{\eta\beta\varphi}\left( \frac{ B^\tau \log B}{\tau} - \frac{B^\tau - 1}{\tau^2} \right) \nonumber \\
& = \frac{1}{\tau^2} \left( - \frac{2cb}{C_2} - \frac{2b^2}{C_2 \tau} + \frac{\rho\delta}{\eta\beta\varphi}\left( \tau B^\tau \log B - (B^\tau - 1) \right)  \right) \label{eq:TauOptBoundedProof1}
\end{align}
From \cite{topsok2006some}, we know that $\log B \geq \frac{2 \eta\beta}{2 + \eta\beta}$. We thus have
\begin{align}
\tau B^\tau \log B - (B^\tau - 1) & \geq  \left( \frac{2 \eta\beta \tau}{2 + \eta\beta}  - 1 \right) B^\tau + 1 \nonumber \\
& > \frac{2 \eta\beta \tau}{2 + \eta\beta}   \label{eq:TauOptBoundedProof2}
\end{align}
where the last inequality is because $ \frac{2 \eta\beta \tau}{2 + \eta\beta}  - 1 > 0$ due to $\tau > \tau_0 \geq \frac{1}{\eta\beta} + \frac{1}{2} = \frac{2 + \eta\beta}{2 \eta\beta}$, and $B^\tau > 1$ due to $B > 1$ and $\tau > \tau_0 > 1$.
Plugging (\ref{eq:TauOptBoundedProof2}) into (\ref{eq:TauOptBoundedProof1}), we have
\begin{align*}
\frac{d H_2(\tau)}{d\tau} & > \frac{1}{\tau^2} \left( - \frac{2cb}{C_2} - \frac{2b^2}{C_2 \tau} + \frac{2 \rho\delta\tau}{\varphi(2 + \eta\beta)}  \right)  \\
& > \frac{1}{\tau^2} \left( - \frac{2cb}{C_2} - \frac{2b^2}{C_2} + \frac{2 \rho\delta\tau}{\varphi(2 + \eta\beta)}  \right)  \tag{$\tau > 1$, $\frac{2b^2}{C_2}>0$} \\
& > 0 \tag{$\tau > \tau_0 \geq \frac{\varphi(2 + \eta\beta)}{2 \rho\delta}  \left( \frac{2cb}{C_2} + \frac{2b^2}{C_2} \right)$}.
\end{align*}

We have now proven that $\frac{d H_1(\tau)}{d\tau} > 0$ and $\frac{d H_2(\tau)}{d\tau} > 0$. Because $\sqrt{x}$ increases with $x$ for any $x\geq 0$, we conclude that $G(\tau)$ increases with $\tau$ for $\tau > \tau_0$. Hence, $\tau^* \leq \tau_0$.

\subsection{Parameters for Generating Resource Consumptions in Simulation}
\label{append:ExperimentSimParam}

The mean and standard deviation values for randomly generating resource consumptions in the simulation are shown in Tables~\ref{tab:DGDDistributedSimParam}, \ref{tab:SGDDistributedSimParam}, and \ref{tab:SGDCentralizedSimParam}. All these values are obtained from measurements on the prototype system when running with the SVM model. The distributed DGD uses different distributions for each case, because the amount of data samples processed in Case~3 is different from other cases. The distributed SGD uses the same distribution for all cases, because the mini-batch size remains the same among all cases. When running the centralized SGD in simulations, only the local update time is generated randomly, because the centralized gradient descent does not include any global aggregation step; thus Table~\ref{tab:SGDCentralizedSimParam}  only includes the mean and standard deviation values for local update. We note that we never simulated centralized DGD in our experiments, thus we do not include values for centralized DGD here.
\vspace{0.1in}

\begin{table}[h]
\caption {Parameters for generating resource consumptions for distributed DGD} \label{tab:DGDDistributedSimParam} 
\vspace{-0.1in}
{\footnotesize
\begin{center}
\begin{tabular}{clcc}
\hline
Case & Resource type & Mean  & Standard deviation  \\
\hline
1 & Local update (seconds) & 0.020613052 & 0.008154439  \\
& Global aggregation (seconds) & 0.137093837  &  0.05548447 \\
\hline
2 & Local update (seconds) & 0.021810727 & 0.008042984  \\
& Global aggregation (seconds) & 0.12322071 & 0.048079171  \\
\hline
3 & Local update (seconds) & 0.095353094 &  0.016688657 \\
& Global aggregation (seconds) & 0.157255906 & 0.066722225  \\
\hline
4 & Local update (seconds) & 0.022075891 & 0.008528005  \\
& Global aggregation (seconds) & 0.108598094 &  0.044627335 \\
\hline
\end{tabular}
\end{center}
}
\end{table}

\begin{table}[h]
\caption {Parameters for generating resource consumptions for distributed SGD} \label{tab:SGDDistributedSimParam} 
\vspace{-0.1in}
{\footnotesize
\begin{center}
\begin{tabular}{lcc}
\hline
Resource type & Mean  & Standard deviation  \\
\hline
Local update (seconds) & 0.013015156 &  0.006946299 \\
Global aggregation (seconds) & 0.131604348 & 0.053873234  \\
\hline
\end{tabular}
\end{center}
}
\end{table}

\begin{table}[h]
\caption {Parameters for generating resource consumptions for centralized SGD} \label{tab:SGDCentralizedSimParam} 
\vspace{-0.1in}
{\footnotesize
\begin{center}
\begin{tabular}{lcc}
\hline
Resource type & Mean  & Standard deviation  \\
\hline
Local update (seconds) & 0.009974248 &  0.011922926 \\
\hline
\end{tabular}
\end{center}
}
\end{table}
\vspace{-0.2in}

\subsection{Additional Results on Varying Global Aggregation Time}
\label{append:VaryingGlobalAggDGDExperimentResults}

See Fig.~\ref{fig:VaryingGlobalAggDGDExperimentResults}.
\vspace{-0.1in}

\begin{figure*}[b]
    \centering
    \begin{subfigure}{0.25\textwidth}
        \centering
        \hspace{1\linewidth}
    \end{subfigure}%
    ~
    \begin{subfigure}{0.18\textwidth}
        \centering
        \hspace{1\linewidth}
    \end{subfigure}%
    ~
    \begin{subfigure}{0.18\textwidth}
        \centering
        \includegraphics[width=0.7\linewidth]{Varying_EachCase_Legend1.pdf}
    \end{subfigure}%
    ~
    \begin{subfigure}{0.18\textwidth}
        \centering
        \includegraphics[width=0.7\linewidth]{Varying_EachCase_Legend2.pdf}
    \end{subfigure}%
    ~
    \begin{subfigure}{0.18\textwidth}
        \centering
        \hspace{1\linewidth}
    \end{subfigure}%

    \begin{subfigure}[b]{0.25\textwidth}
        \centering
        \includegraphics[width=0.4\linewidth]{Varying_Tau_Legend.pdf}
        \vspace{0.2in}
    \end{subfigure}%
    ~
    \begin{subfigure}[b]{0.18\textwidth}
        \centering
        \includegraphics[width=1\linewidth]{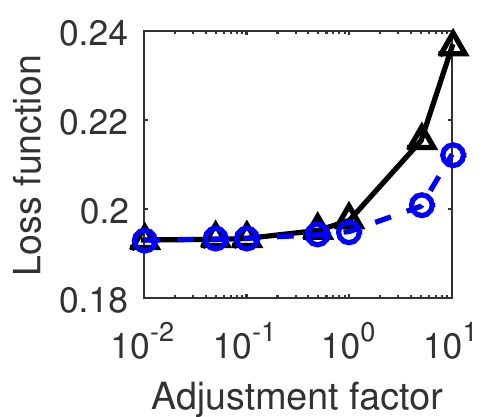}
    \end{subfigure}%
    ~
    \begin{subfigure}[b]{0.18\textwidth}
        \centering
        \includegraphics[width=1\linewidth]{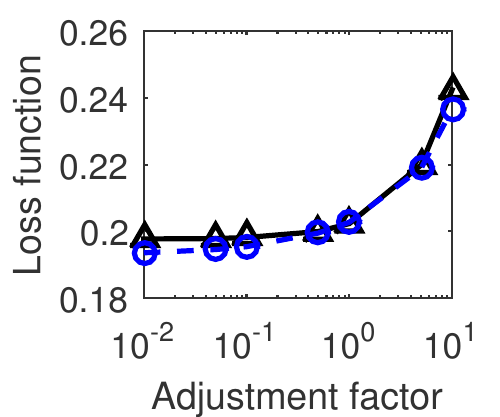}
    \end{subfigure}%
    ~
    \begin{subfigure}[b]{0.18\textwidth}
        \centering
        \includegraphics[width=1\linewidth]{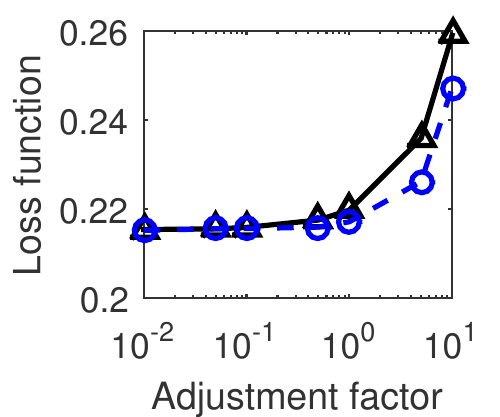}
    \end{subfigure}%
    ~
    \begin{subfigure}[b]{0.18\textwidth}
        \centering
        \includegraphics[width=1\linewidth]{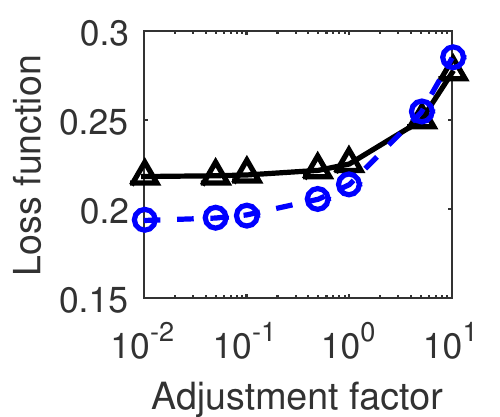}
    \end{subfigure}%
    \vspace{-0.1in}

    \begin{subfigure}[b]{0.25\textwidth}
        \centering
        \includegraphics[width=1\linewidth]{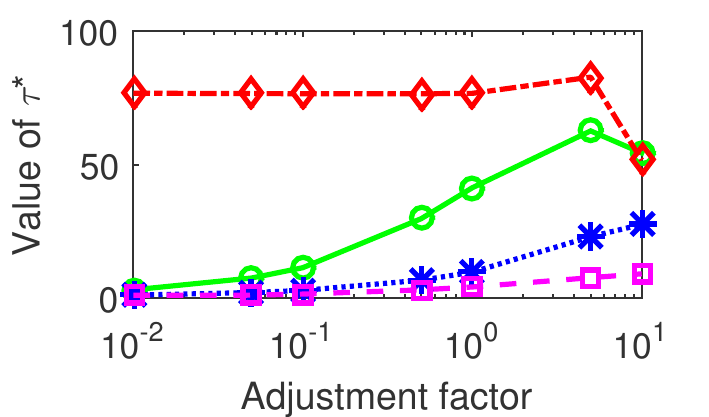}
        \caption{$\tau^*$ in proposed algorithm}
    \end{subfigure}%
    ~
    \begin{subfigure}[b]{0.18\textwidth}
        \centering
        \includegraphics[width=1\linewidth]{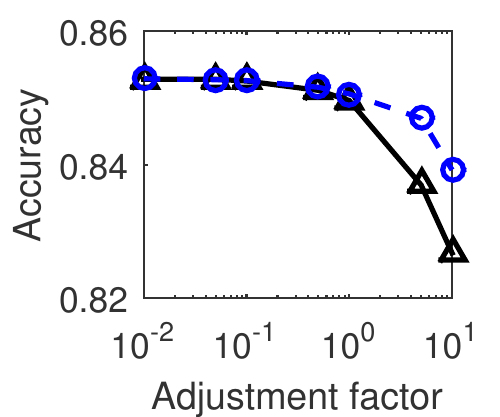}
        \caption{Case 1}
    \end{subfigure}%
    ~
    \begin{subfigure}[b]{0.18\textwidth}
        \centering
        \includegraphics[width=1\linewidth]{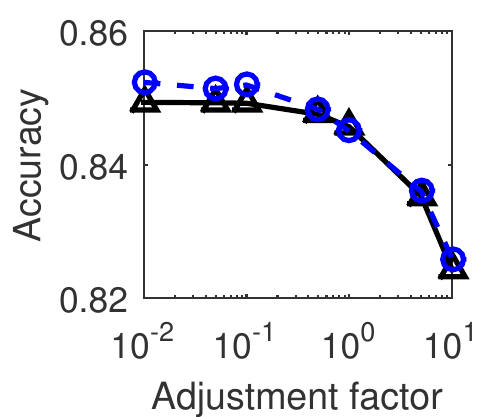}
        \caption{Case 2}
    \end{subfigure}%
    ~
    \begin{subfigure}[b]{0.18\textwidth}
        \centering
        \includegraphics[width=1\linewidth]{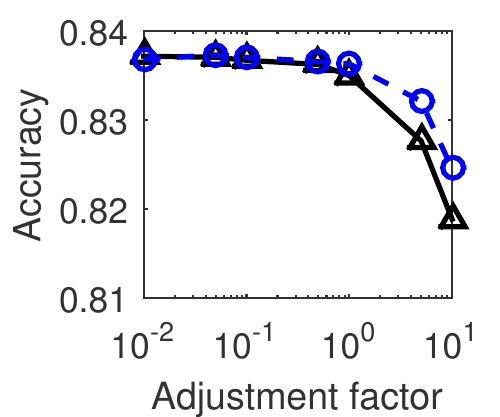}
        \caption{Case 3}
    \end{subfigure}%
    ~
    \begin{subfigure}[b]{0.18\textwidth}
        \centering
        \includegraphics[width=1\linewidth]{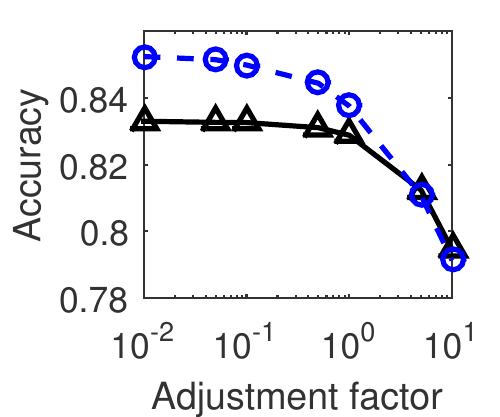}
        \caption{Case 4}
    \end{subfigure}%

\caption{Loss function values and classification accuracy with different global aggregation times for SVM (DGD).}
\label{fig:VaryingGlobalAggDGDExperimentResults}
\end{figure*}

\subsection{Additional Results on Varying Total Time Budget}
\label{append:VaryingTotalTimeDGDExperimentResults}

See Fig.~\ref{fig:VaryingTotalTimeDGDExperimentResults}.
\vspace{-0.1in}

\begin{figure*}[b]
    \centering
    \begin{subfigure}{0.25\textwidth}
        \centering
        \hspace{1\linewidth}
    \end{subfigure}%
    ~
    \begin{subfigure}{0.18\textwidth}
        \centering
        \hspace{1\linewidth}
    \end{subfigure}%
    ~
    \begin{subfigure}{0.18\textwidth}
        \centering
        \includegraphics[width=0.7\linewidth]{Varying_EachCase_Legend1.pdf}
    \end{subfigure}%
    ~
    \begin{subfigure}{0.18\textwidth}
        \centering
        \includegraphics[width=0.7\linewidth]{Varying_EachCase_Legend2.pdf}
    \end{subfigure}%
    ~
    \begin{subfigure}{0.18\textwidth}
        \centering
        \hspace{1\linewidth}
    \end{subfigure}%

    \begin{subfigure}[b]{0.25\textwidth}
        \centering
        \includegraphics[width=0.4\linewidth]{Varying_Tau_Legend.pdf}
        \vspace{0.2in}
    \end{subfigure}%
    ~
    \begin{subfigure}[b]{0.18\textwidth}
        \centering
        \includegraphics[width=1\linewidth]{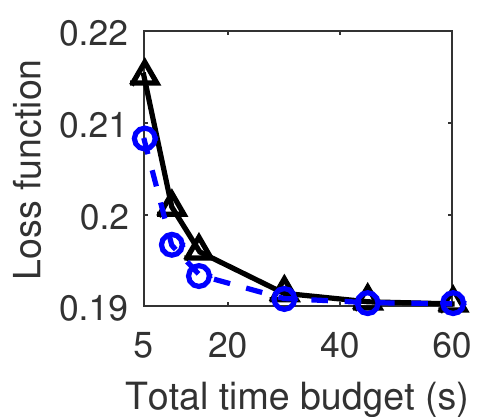}
    \end{subfigure}%
    ~
    \begin{subfigure}[b]{0.18\textwidth}
        \centering
        \includegraphics[width=1\linewidth]{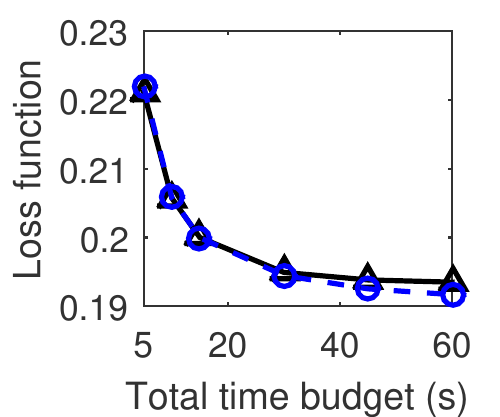}
    \end{subfigure}%
    ~
    \begin{subfigure}[b]{0.18\textwidth}
        \centering
        \includegraphics[width=1\linewidth]{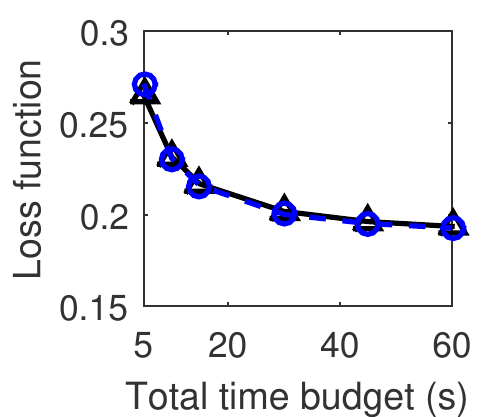}
    \end{subfigure}%
    ~
    \begin{subfigure}[b]{0.18\textwidth}
        \centering
        \includegraphics[width=1\linewidth]{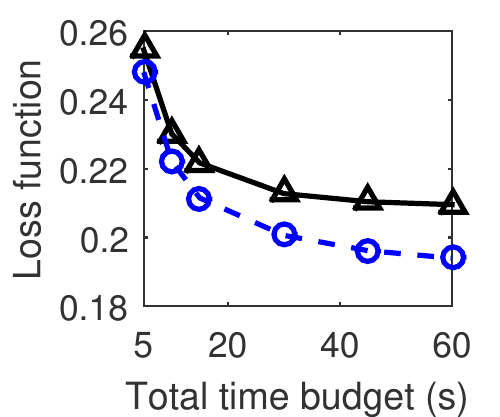}
    \end{subfigure}%
    \vspace{-0.1in}

    \begin{subfigure}[b]{0.25\textwidth}
        \centering
        \includegraphics[width=1\linewidth]{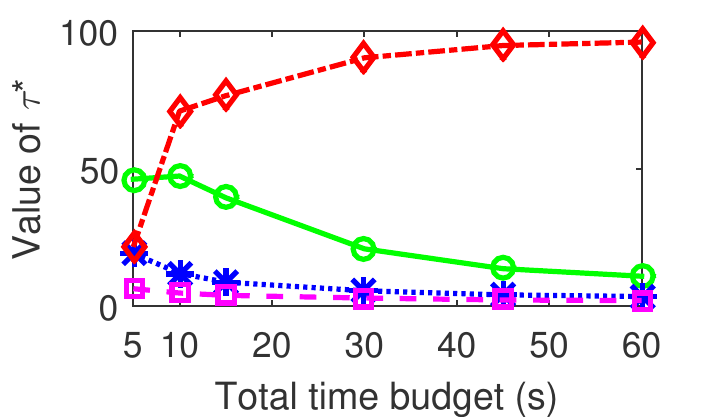}
        \caption{$\tau^*$ in proposed algorithm}
    \end{subfigure}%
    ~
    \begin{subfigure}[b]{0.18\textwidth}
        \centering
        \includegraphics[width=1\linewidth]{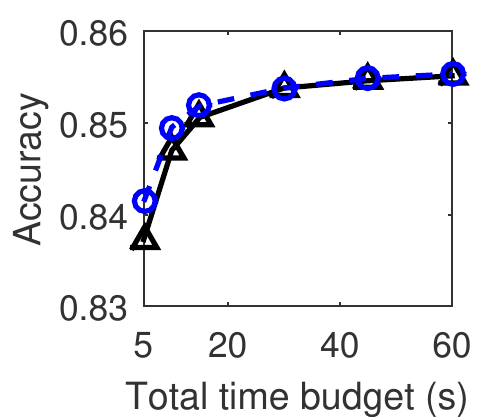}
        \caption{Case 1}
    \end{subfigure}%
    ~
    \begin{subfigure}[b]{0.18\textwidth}
        \centering
        \includegraphics[width=1\linewidth]{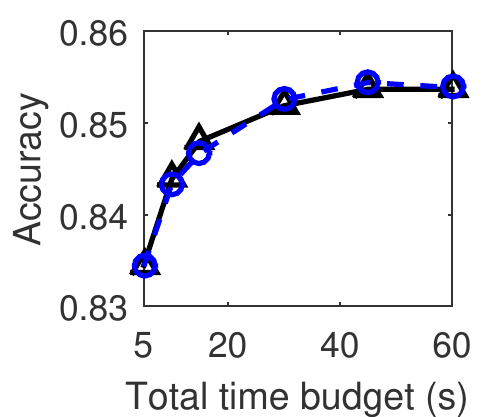}
        \caption{Case 2}
    \end{subfigure}%
    ~
    \begin{subfigure}[b]{0.18\textwidth}
        \centering
        \includegraphics[width=1\linewidth]{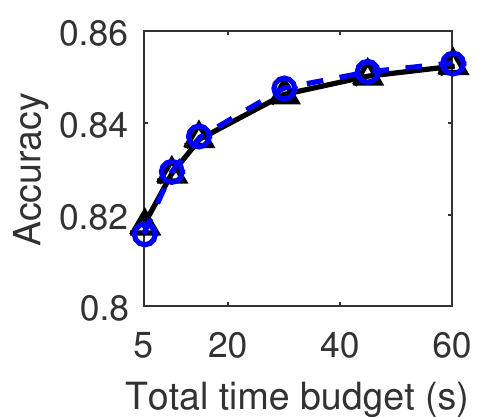}
        \caption{Case 3}
    \end{subfigure}%
    ~
    \begin{subfigure}[b]{0.18\textwidth}
        \centering
        \includegraphics[width=1\linewidth]{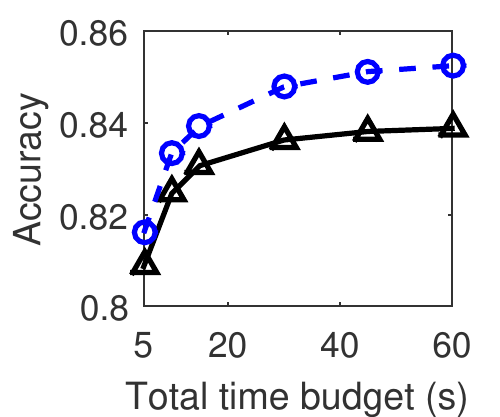}
        \caption{Case 4}
    \end{subfigure}%

\caption{Loss function values and classification accuracy with different total time budgets for SVM (DGD).}
\label{fig:VaryingTotalTimeDGDExperimentResults}
\end{figure*}

\subsection{Additional Results on Instantaneous Behavior}
\label{append:InstantSVMSGDExperimentResults}

See Fig.~\ref{fig:InstantSVMSGDExperimentResults}.

\begin{figure*}[b]
    \centering
    \begin{subfigure}[b]{0.18\textwidth}
        \centering
        \includegraphics[width=0.6\linewidth]{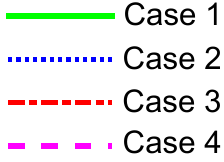}
        \vspace{0.2in}
    \end{subfigure}%
    ~
    \begin{subfigure}[b]{0.2\textwidth}
        \centering
        \includegraphics[width=1\linewidth]{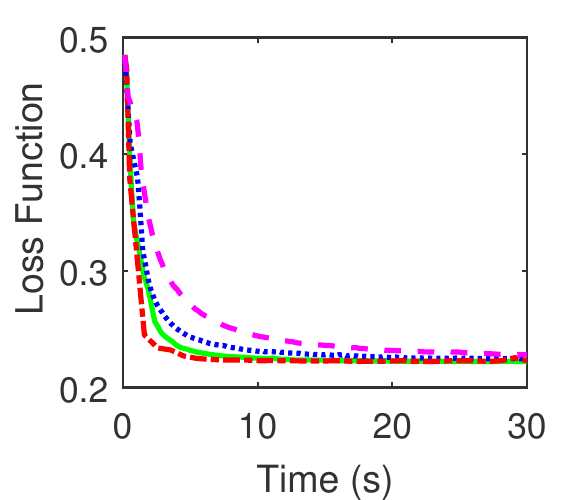}
    \end{subfigure}%
    ~
    \begin{subfigure}[b]{0.2\textwidth}
        \centering
        \includegraphics[width=1\linewidth]{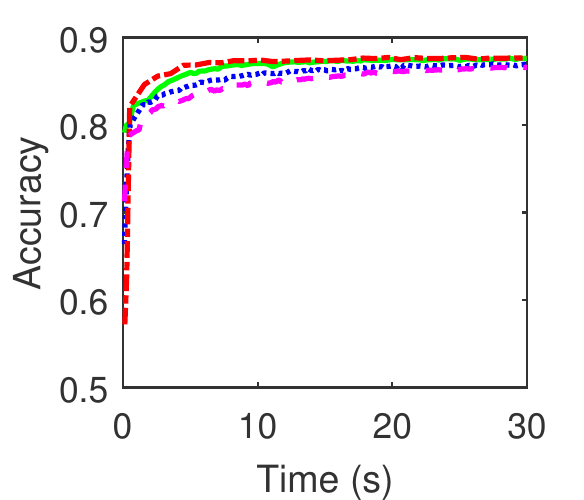}
    \end{subfigure}%
    ~
    \begin{subfigure}[b]{0.2\textwidth}
        \centering
        \includegraphics[width=1\linewidth]{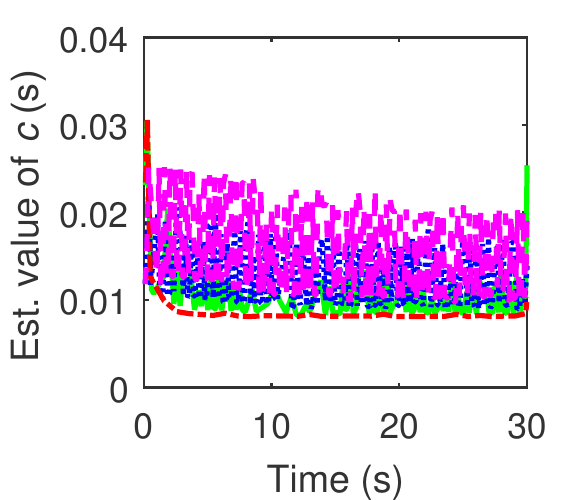}
    \end{subfigure}%
    ~
    \begin{subfigure}[b]{0.2\textwidth}
        \centering
        \includegraphics[width=1\linewidth]{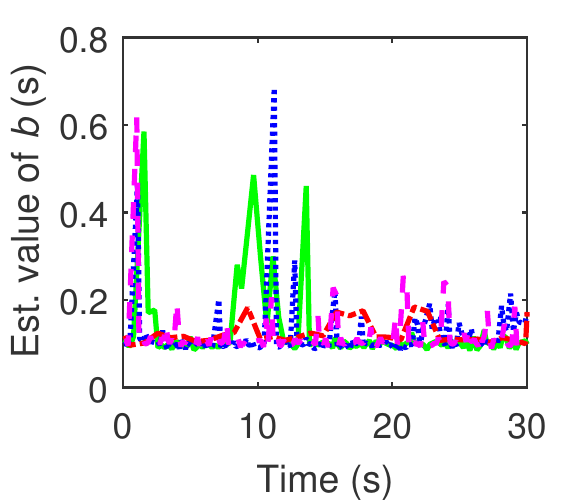}
    \end{subfigure}%
    \vspace{-0.1in}

    \begin{subfigure}[b]{0.18\textwidth}
        \centering
        \hspace{1\linewidth}
    \end{subfigure}%
    ~
    \begin{subfigure}[b]{0.2\textwidth}
        \centering
        \includegraphics[width=1\linewidth]{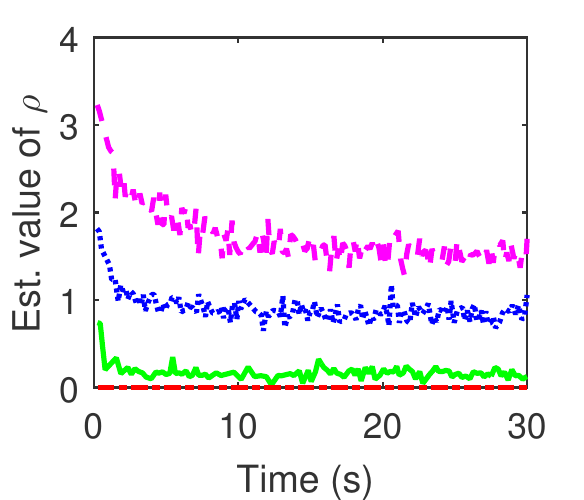}
    \end{subfigure}%
    ~
    \begin{subfigure}[b]{0.2\textwidth}
        \centering
        \includegraphics[width=1\linewidth]{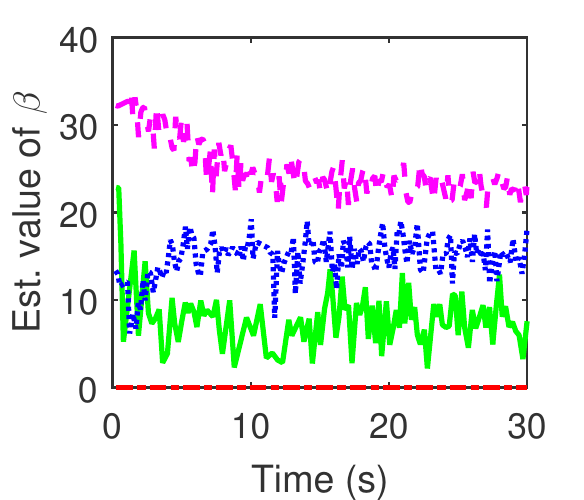}
    \end{subfigure}%
    ~
    \begin{subfigure}[b]{0.2\textwidth}
        \centering
        \includegraphics[width=1\linewidth]{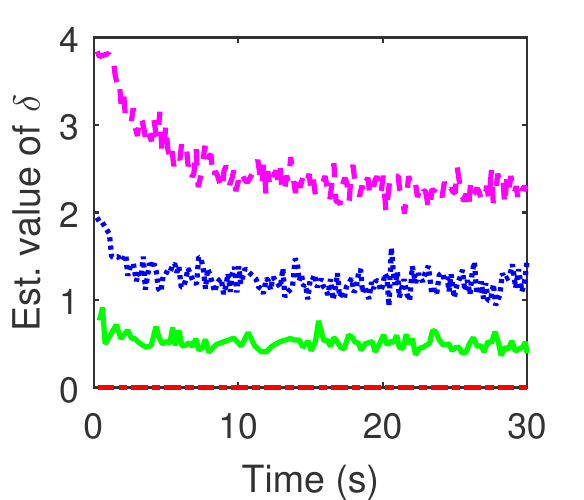}
    \end{subfigure}%
    ~
    \begin{subfigure}[b]{0.2\textwidth}
        \centering
        \includegraphics[width=1\linewidth]{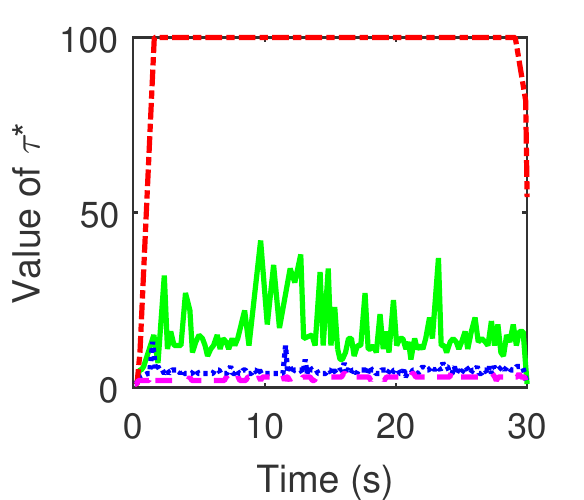}
    \end{subfigure}%

\caption{Instantaneous results of SVM (SGD) with the proposed algorithm.}
\label{fig:InstantSVMSGDExperimentResults}
\end{figure*}

\end{document}